\documentclass[12pt]{article}
\usepackage{authblk}
\usepackage{geometry}
\usepackage{amsthm}
\usepackage{mathrsfs}
\usepackage{customcommands}
\usepackage{bm}
\usepackage{Chrisstyle}

\usepackage{amsmath}
\usepackage{amssymb}
\usepackage{amsfonts}
\usepackage{braket}
\usepackage[normalem]{ulem}
\usepackage{color}
\usepackage{ulem}
\usepackage{mathtools}
\usepackage{tensor} 
\usepackage{overpic}
\usepackage{breqn} 
\usepackage{subcaption}
\usepackage{float}
\usepackage[export]{adjustbox}
\usepackage{ytableau}
\newtheorem{theorem}{Theorem}
\newtheorem{proposition}[theorem]{Proposition}

\DeclareMathOperator{\tr}{tr}

\DeclarePairedDelimiter\ceil{\lceil}{\rceil}
\DeclarePairedDelimiter\floor{\lfloor}{\rfloor}
\graphicspath{{./fig/}{./fig/sec3/}{./fig3/}}
\newtheorem{corollary}[theorem]{Corollary}
\newtheorem{lemma}[theorem]{Lemma}
\theoremstyle{definition}
\newtheorem{definition}{Definition}[section]

\theoremstyle{remark}

\title{Reflected entropy in random tensor networks II: a topological index from the canonical purification}

\author[1]{Chris Akers,}
\author[2]{Thomas Faulkner,}
\author[2]{Simon Lin,}
\author[3]{Pratik Rath}

\affiliation[1]{Center for Theoretical Physics,\\
Massachusetts Institute of Technology, Cambridge, MA 02139, USA}
\affiliation[2]{Department of Physics, University of Illinois,\\ 1110 W. Green St., Urbana, IL 61801-3080, USA}
\affiliation[3]{Department of Physics, University of California, Santa Barbara, CA 93106, USA}

\emailAdd{cakers@mit.edu}
\emailAdd{tomf@illinois.edu}
\emailAdd{shanlin3@illinois.edu}
\emailAdd{rath@ucsb.edu}

\abstract{In Ref.~\cite{Akers:2021pvd}, we analyzed the reflected entropy ($S_R$) in random tensor networks motivated by its proposed duality to the entanglement wedge cross section (EW) in holographic theories, $S_R=2 \frac{EW}{4G}$. In this paper, we discover further details of this duality by analyzing a simple network consisting of a chain of two random tensors. This setup models a multiboundary wormhole. We show that the reflected entanglement spectrum is controlled by representation theory of the Temperley-Lieb (TL) algebra. In the semiclassical limit motivated by holography, the spectrum takes the form of a sum over superselection sectors associated to different irreducible representations of the TL algebra and labelled by a topological index $k\in \mathbb{Z}_{\geq 0}$. Each sector contributes to the reflected entropy an amount $2k \frac{EW}{4G}$ weighted by its probability. We provide a gravitational interpretation in terms of fixed-area, higher-genus multiboundary wormholes with genus $2k-1$ initial value slices. These wormholes appear in the gravitational description of the canonical purification. 
We confirm the reflected entropy holographic duality away from phase transitions. 
We also find important non-perturbative contributions from the novel geometries with $k\geq 2$ near phase transitions, resolving the discontinuous transition in $S_R$. Along with analytic arguments, we provide numerical evidence for our results. 
We comment on the connection between TL algebras, Type II$_1$ von Neumann algebras and gravity.}

\begin{document}
\maketitle

\section{Introduction}\label{sec:intro}

The intriguing connection between geometry and entanglement in the context of holography has resulted in big leaps in our understanding of quantum gravity. The Ryu-Takayanagi (RT) formula \cite{Ryu:2006bv,Ryu:2006ef,Hubeny:2007xt} relating boundary entropy to the area of bulk extremal surfaces is the hallmark of such an emergence of spacetime from entanglement. In the pursuit of more such links, a proposal for the holographic dual to another geometric object, the entanglement wedge cross section, was made in Ref.~\cite{Dutta:2019gen}. The proposed dual, the reflected entropy, is a novel measure of correlation between bipartite mixed states, or equivalently tripartite pure states. 

The reflected entropy is defined as
\begin{equation}\label{eq:SR}
	S_R(A:B) = S(AA^*)_{\ket{\sqrt{\rho_{AB}}}},
\end{equation}
where the state $\ket{\sqrt{\rho_{AB}}}\in \mathcal{H}_{AB}\otimes \mathcal{H}_{A^*B^*}$ is the canonical purification of the density matrix $\rho_{AB}$. The subsystems $A^*,B^*$ are referred to as the reflected copies of the subsystems $A,B$ respectively. The holographic proposal then states
\begin{equation}\label{eq:EW}
	S_R(A:B) = \frac{2 EW(A:B)}{4G},
\end{equation} 
where $EW(A:B)$ is the minimal cross section splitting the entanglement wedge of $AB$ (see \figref{fig:EWCS}). In \Eqref{eq:EW}, we have ignored quantum corrections, as well as time dependence (see Refs.~\cite{Chandrasekaran:2020qtn,Hayden:2021gno} for details). For simplicity, this paper will be limited to discussing the static, classical proposal although there is no reason to suspect the results do not generalize.
\footnote{While time dependence is straightforward, quantum corrections would likely include subtleties arising from corrections to the QES formula \cite{Marolf:2020vsi,Dong:2020iod,Akers:2020pmf}.} 
This proposal has already been useful in demonstrating the need for large amounts of tripartite entanglement in holographic states \cite{Akers:2019gcv}. More so, it can be thought of as a generalization of the RT formula with a boundary dual that is rigorously well defined even in the continuum limit \cite{Dutta:2019gen}. Thus, it is interesting to find evidence for such a duality.

\begin{figure}[t]
    \centering
    \includegraphics[scale=0.5]{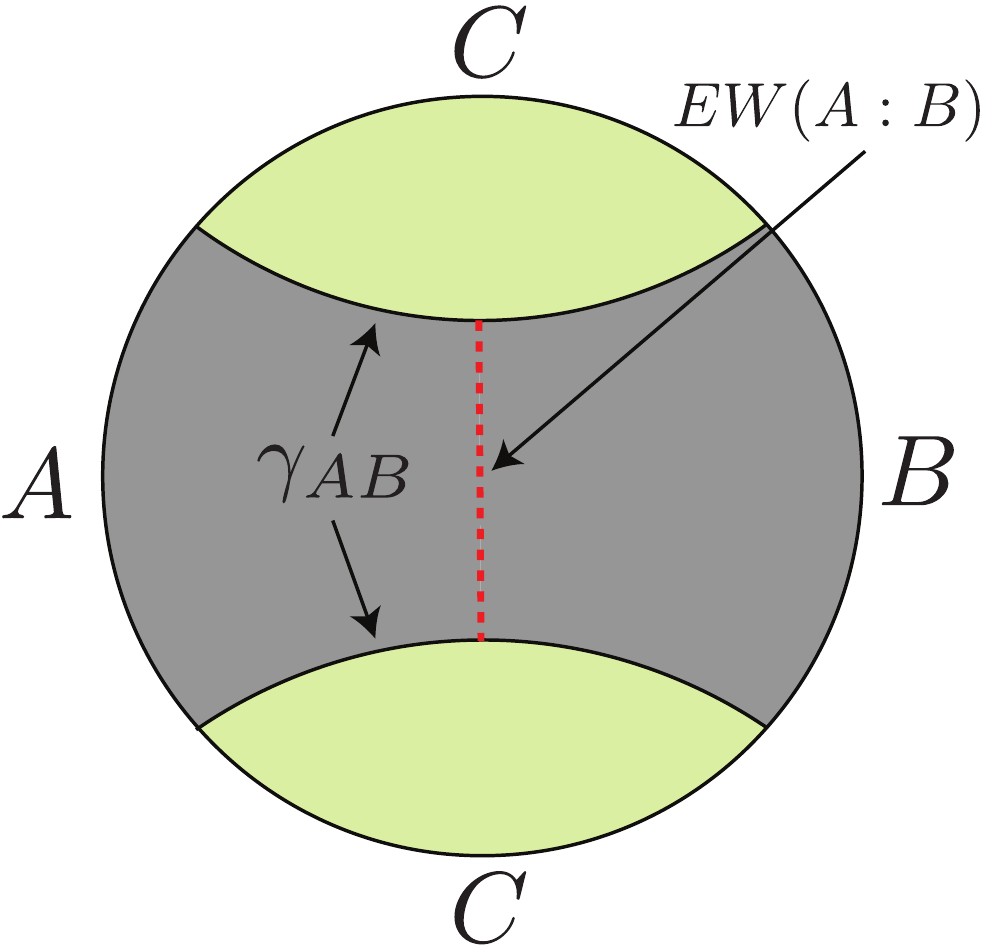}
    \caption{$EW(A:B)$ is the minimal area surface that divides the entanglement wedge of $AB$, bounded by the RT surface $\gamma_{AB}$, into regions homologous to subregions $A$ and $B$ respectively.}
    \label{fig:EWCS}
\end{figure}

The original argument for the proposal involved a two-parameter replica trick, followed by an analytic continuation \`a la Lewkowycz-Maldacena \cite{Lewkowycz:2013nqa}. While the proposal passes various sanity checks, it was noted in Ref.~\cite{Kusuki:2019evw} that the replica trick argument itself suffered from an order of limits issue. More so, the EW cross-section undergoes a discontinuous transition when the entanglement wedge changes from disconnected to connected. This raises the possibility of non-perturbative effects becoming important to resolve the phase transition. Thus, it is of interest to use solvable toy models to better understand the above issues. 

In Ref.~\cite{Akers:2021pvd}, we used random tensor networks (RTNs) \cite{Hayden:2016cfa} as a playground to understand the various subtleties associated with the replica trick argument. In particular, for a tripartite state generated from a single random tensor, we were able to use analytic and numerical techniques to extract the reflected spectrum, the entanglement spectrum of $\rho_{AA^*}$. A crucial role was played by the addition of a novel saddle which dominated in portions of parameter space and motivated a resolution to the order of limits issue. Our analysis provided evidence for the validity of the proposal in \Eqref{eq:EW}. Since the replica trick in RTNs involves a sum over permutations that is quite analogous to the sum over topologies in the gravitational path integral, there is good reason to believe that the analysis in RTNs is a faithful indicator of the calculation in gravity. Moreover, we were also able to solve the above problem in the West Coast model consisting of Jackiw-Teitelboim gravity coupled to end-of-the-world branes \cite{Penington:2019kki}, finding further evidence for \Eqref{eq:EW} along with novel features near phase transitions \cite{Akers:2022max}. 

\begin{figure}[t]
  \centering
  \raisebox{-0.5\height}{\includegraphics[scale=.3]{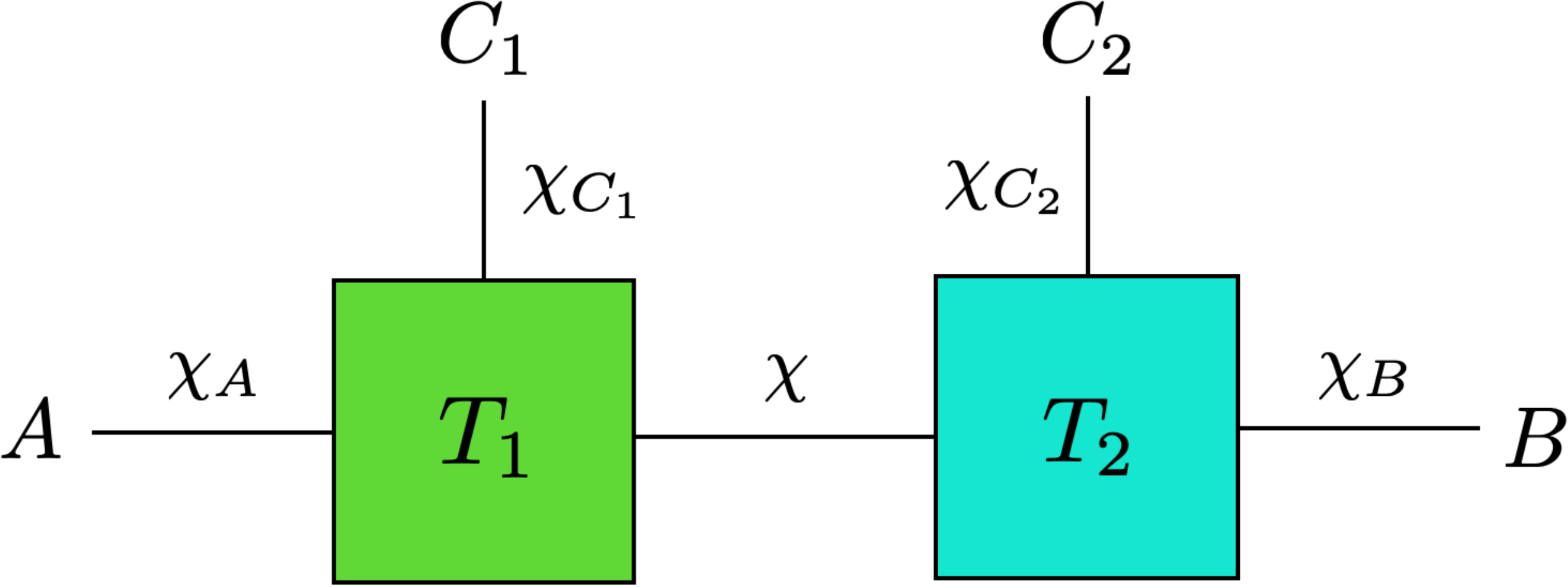}}
  \hspace{.2in}
  \raisebox{-0.5\height}{\includegraphics[scale=.3]{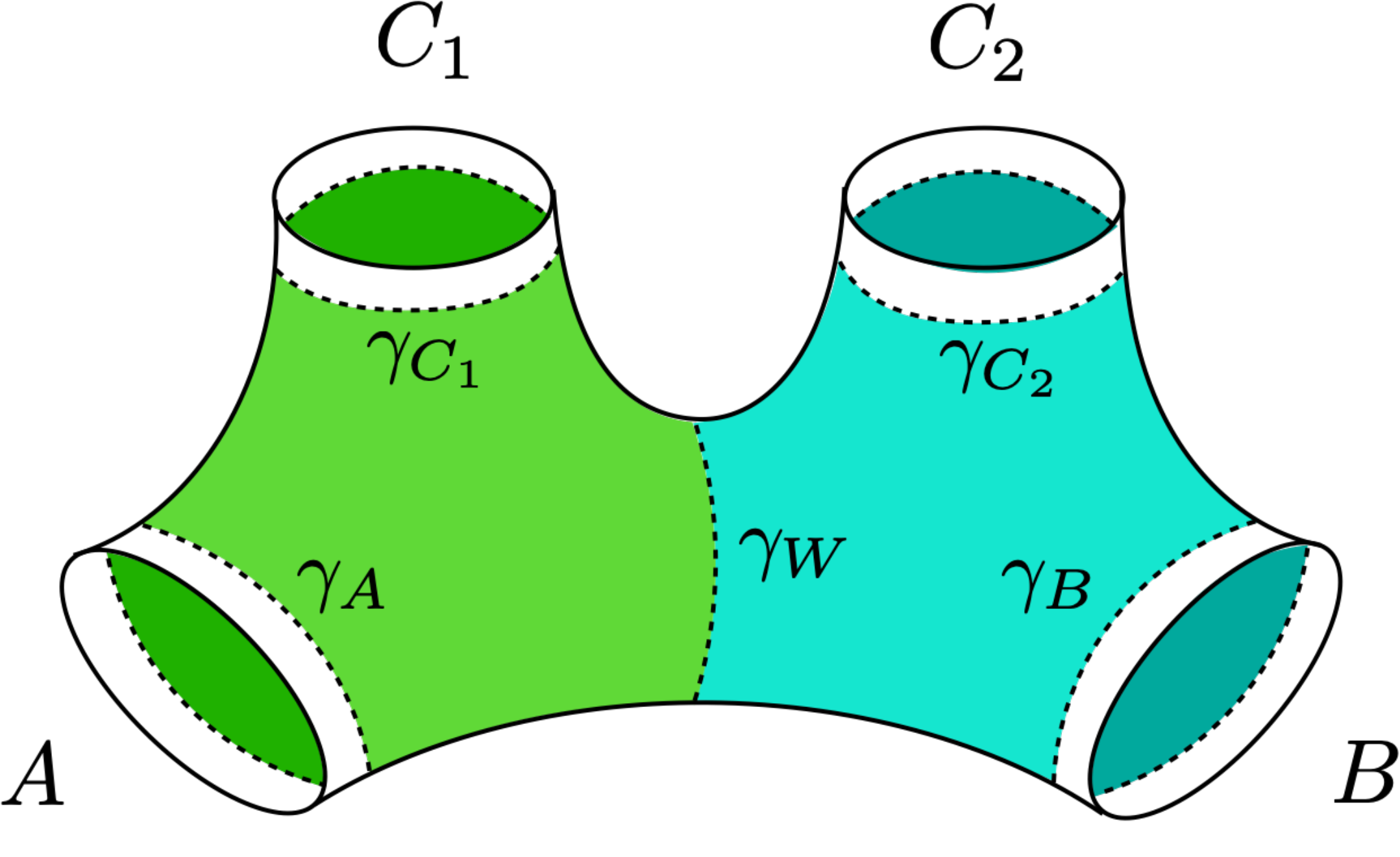}}
  \caption{(left) The 2TN tensor network considered in this section is built from two random tensors $T_1$ and $T_2$. The parameters are the boundary bond dimensions $\chi_A,\chi_B,\chi_{C_1},\chi_{C_2}$ and the internal bond dimension $\chi$.
    (right) The wormhole solution that is modeled by 2TN. The external bond dimensions corresponds to the three horizon areas and the internal bond dimension $\chi$ corresponds to the cross-section surface $\gamma_W$.}
  \label{fig:2TN}
\end{figure}

In this paper, we carry on with our analysis of reflected entropy in RTNs in the hope of finding other undiscovered aspects of the replica trick. In particular, we will focus on an RTN consisting of two random tensors, which we refer to as 2TN. 2TN can be interpreted as a model for a four-boundary wormhole as depicted in \figref{fig:2TN}, where the areas of the labelled surfaces are fixed to a narrow window \cite{Dong:2018seb, Akers:2018fow,Dong:2019piw}. More generally, we will provide heuristic arguments that the calculations in 2TN are also useful for more general settings, e.g., the familiar setup of two intervals in vacuum AdS depicted in \figref{fig:EWCS}. Since the bulk geometry is coarse-grained down to just two tensors, the model cannot capture any of the local dynamics. However, it does capture general topological aspects of the gravitational calculation which turn out to be the relevant aspect for the reflected entropy, including near phase transition effects.

In \secref{sec:gravity}, we start by motivating the gravitational construction of novel, higher genus saddles that contribute to the canonical purification. We consider the gravitational state corresponding to the four boundary wormhole depicted in \figref{fig:2TN}, prepared using a Euclidean path integral with fixed area boundary conditions. As discussed in Refs.~\cite{Dong:2018seb,Akers:2018fow,Dong:2019piw,Penington:2019kki,Akers:2020pmf,Marolf:2020vsi,Dong:2020iod}, the replica trick for such fixed-area states is simplified by the fact that one can simply glue together multiple copies of the original bulk geometry without having to solve for a new backreacted geometry. Thus, we have control over the different saddles contributing to the canonical purification. By doing a replica trick to construct the state $\ket{\rho_{AB}^{m/2}}$ for even integer $m$ \cite{Dutta:2019gen}, we find saddles labelled by a topological index $k\in \mathbb{Z}_{>0}$. They correspond to geometries with initial data slices obtained by gluing together $2 k$ copies of the shaded region (see \figref{fig:2TN}) of the connected entanglement wedge of $AB$ in the original state. Each such geometry contributes with an amplitude $\sqrt{p_k}$ computed from the path integral. The canonically purified state can then be obtained via analytic continuation to $m=1$, and is approximately given by a superposition over such geometries as shown in \figref{fig:genus}. Thus, we obtain a family of geometries that contribute to the entanglement wedge gluing procedure dual to the canonical purification \cite{Engelhardt:2017aux,Engelhardt:2018kcs,Dutta:2019gen,Marolf:2020vsi}. Finally, computing the reflected entropy, we see that the geometry labelled by $k$ contributes an amount $2k \frac{EW(A:B)}{4G}$ weighted by its probability.

\begin{figure}[t]
  \centering
  \includegraphics[width=.95\textwidth]{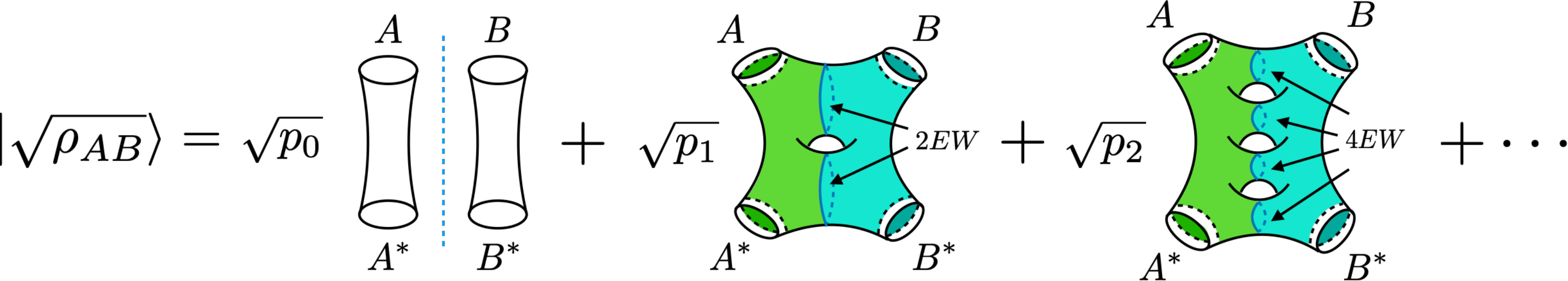}
  \caption{The canonical purification $\ket{\sqrt{\rho_{AB}}}$ consists of a superposition of a one parameter family of geometries labelled by $k$. They are obtained by gluing together $2k$ copies of the shaded portion (see \figref{fig:2TN}) of the connected entanglement wedge of $AB$ in the four boundary wormhole. The $k$-th geometry has $2k$ copies of the entanglement wedge cross section labelled $EW$.}
  \label{fig:genus}
\end{figure}

Having motivated the existence of these higher genus geometries from the gravitational path integral, we set up the 2TN problem in \secref{sec:setup} to get a better handle on such effects. The replica trick for reflected entropy, discussed in \secref{sub:rtns}, involves computing the so called $(m,n)$-R\'enyi reflected entropy \cite{Dutta:2019gen}. Here $n$ is the usual R\'enyi entropy index and $m$ labels the state $\ket{\rho_{AB}^{m/2}}$, a generalization of the canonical purification. For the reflected entropy, one needs to then analytically continue to $m,n\to 1$.
Analyzing the $(m,n)$ replica trick for the 2TN problem beyond the saddle point approximation requires a new tool, the Temperley-Lieb (TL) algebra \cite{Temperley:1971iq}, which we introduce in \secref{sub:TL}. 
Using the resolvent trick \cite{Penington:2019kki,Shapourian:2020mkc,Akers:2020pmf,Dong:2021oad,Vardhan:2021npf, Akers:2021pvd,Akers:2022max}, we show that the reflected spectrum can be categorized into different sectors in terms of the irreducible representations of the TL algebra. These irreps are labeled by an index $k$ which we call ``topological" since it corresponds to the topology of the higher genus saddles in the gravitational path integral.  

\begin{figure}[t]
  \centering
  \includegraphics[scale=.35]{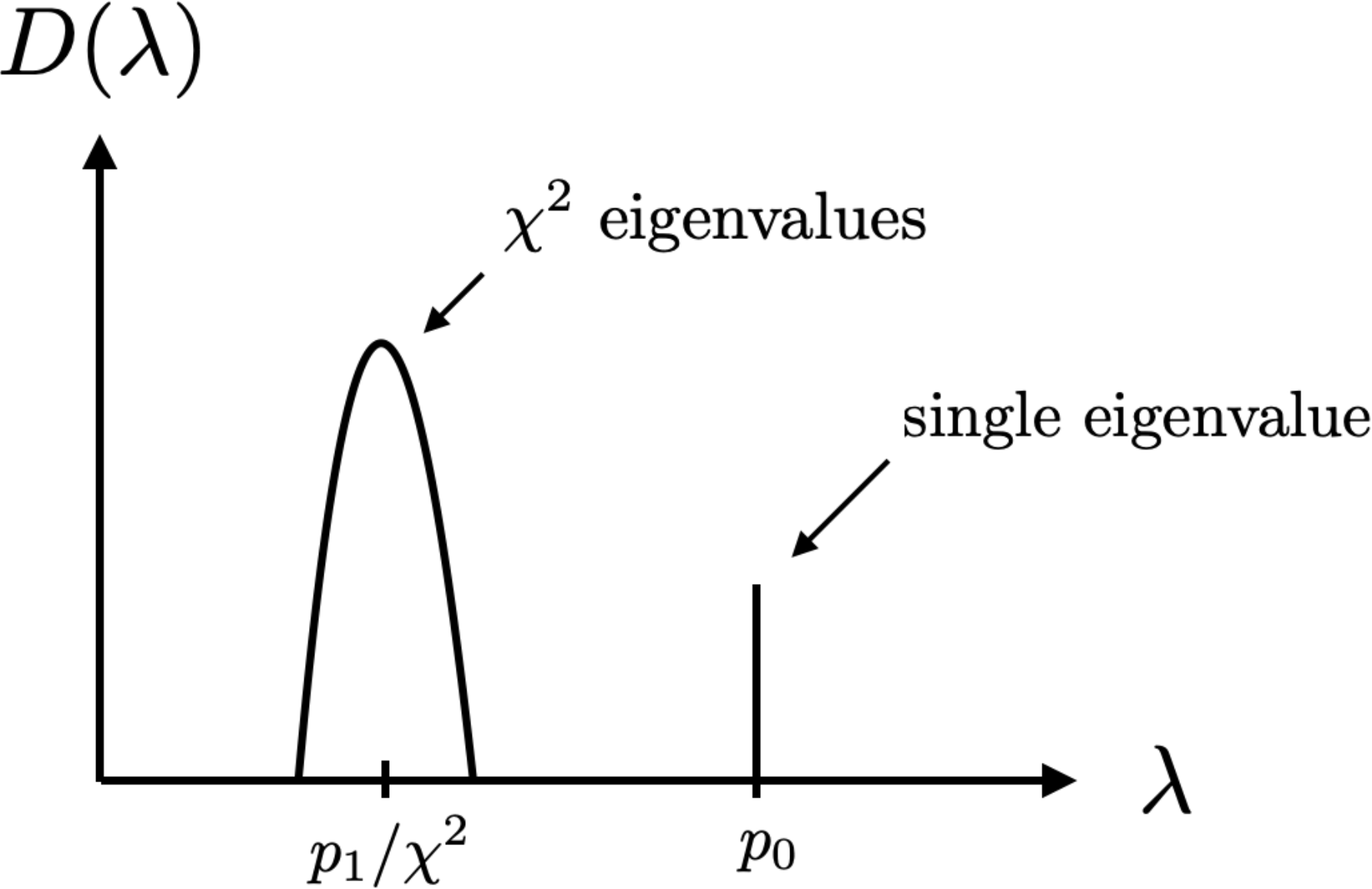}
  \includegraphics[scale=.35]{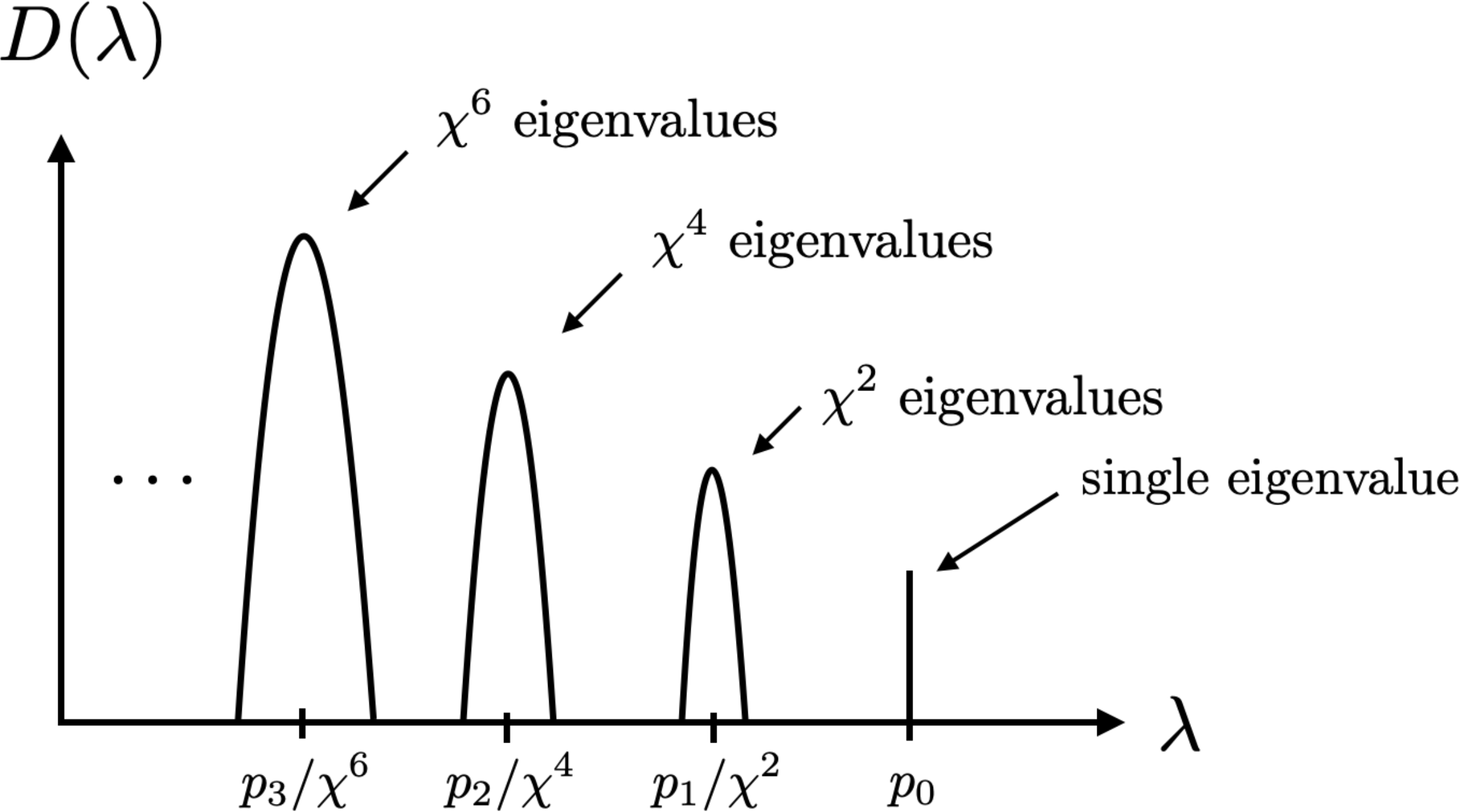}
  \caption{Sketch of spectrum of 1TN vs 2TN. While the 1TN model has two peaks corresponding to the connected and disconnected phases, the 2TN model has infinitely many peaks corresponding to the novel, higher-genus saddles discovered in \secref{sec:gravity}.}
  \label{fig:spectrum}
\end{figure}

With the formalism for computing the reflected spectrum set up, we compute and analyze the 2TN reflected spectrum and entropy using the TL algebra in \secref{sec:2TN}. As a proof of principle, we first solve for the spectrum at first few even integer values of $m$ in \secref{sub:finite}. In order to then relate to the semiclassical limit in gravity, we take the limit where $\chi$ is large in \secref{sub:large}. In this limit, we find the leading-$\chi$ contribution to reflected spectrum that can be analytically continued to $m=1$. The spectrum exhibits an infinite sequence of delta function peaks, labelled by the index $k\in\mathbb{Z}_{\ge 0}$. Each peak consists of $\chi^{2k}$ eigenvalues, thus contributing to the reflected entropy by an amount $2k EW(A:B)$.

This 2TN spectrum has a much richer structure than the single random tensor analyzed in Ref.~\cite{Akers:2021pvd}, see comparison in \figref{fig:spectrum}. We analyze the properties of the 2TN spectrum and its relation to emergent superselection sectors and quantum error correction in \secref{sub:sr}. Using the reflected spectrum, we find consistency with the holographic proposal, \Eqref{eq:EW}, away from phase transitions where either $k=0$ or $k=1$ dominates. More so, the new sectors $k\geq 2$ become important near the phase transition and smooth out the discontinuity in the reflected entropy. In \secref{sub:corrections} we consider the leading corrections to the large $\chi$ limit. We find that these corrections shift the locations of the delta functions and spread them into peaks with finite width containing $\chi^{2k}$ eigenvalues. We give an estimate of the shift and relevant widths. Finally, we demonstrate the consistency of our calculations with numerical results in \secref{sub:num}. 

With this, we conclude in \secref{sec:disc} with a discussion of how the results obtained here generalize to arbitrary RTNs which model multiboundary wormholes. We provide heuristic arguments that these additional sectors also contribute to the canonical purification in more general settings such as the two interval example in vacuum AdS. We also comment on the relation of our results to the emergence of non-trivial von Neumann algebras in gravitational theories.  More specifically we speculate that signatures of a non-trivial von Neumann algebra, connected to the TL algebra, will emerge from a modular flowed version of reflected entropy. 

We provide additional details about multiboundary wormholes in Appendix~\ref{sec:multi}, and the Temperley-Lieb algebra in Appendix~\ref{sec:TLalgebra}. Calculations of the leading corrections to the large $\chi$ limit can be found in Appendix~\ref{sec:finite_chi}. Proofs of various results used in \secref{sec:2TN} are available in Appendix~\ref{sec:proof}.

\section{Motivation: Canonical Purification in Gravity} 
\label{sec:gravity}

Before analyzing the 2TN model in detail, we provide motivation for the anticipated results by constructing novel gravitational saddles that contribute to the canonical purification. We will discover new features from the gravitational path integral that will be mirrored by the 2TN model in \secref{sec:2TN}. For simplicity, our discussion will focus on pure 3D gravity with a negative cosmological constant, where multiboundary wormholes are well understood \cite{Brill:1995jv,Krasnov:2000zq,Skenderis:2009ju}. In this context, we can make a direct connection between wormholes and RTNs, closely following and elaborating on the analysis in Ref.~\cite{Balasubramanian:2014hda}.

\begin{figure}[t]
  \centering
  \includegraphics[width=.3\textwidth]{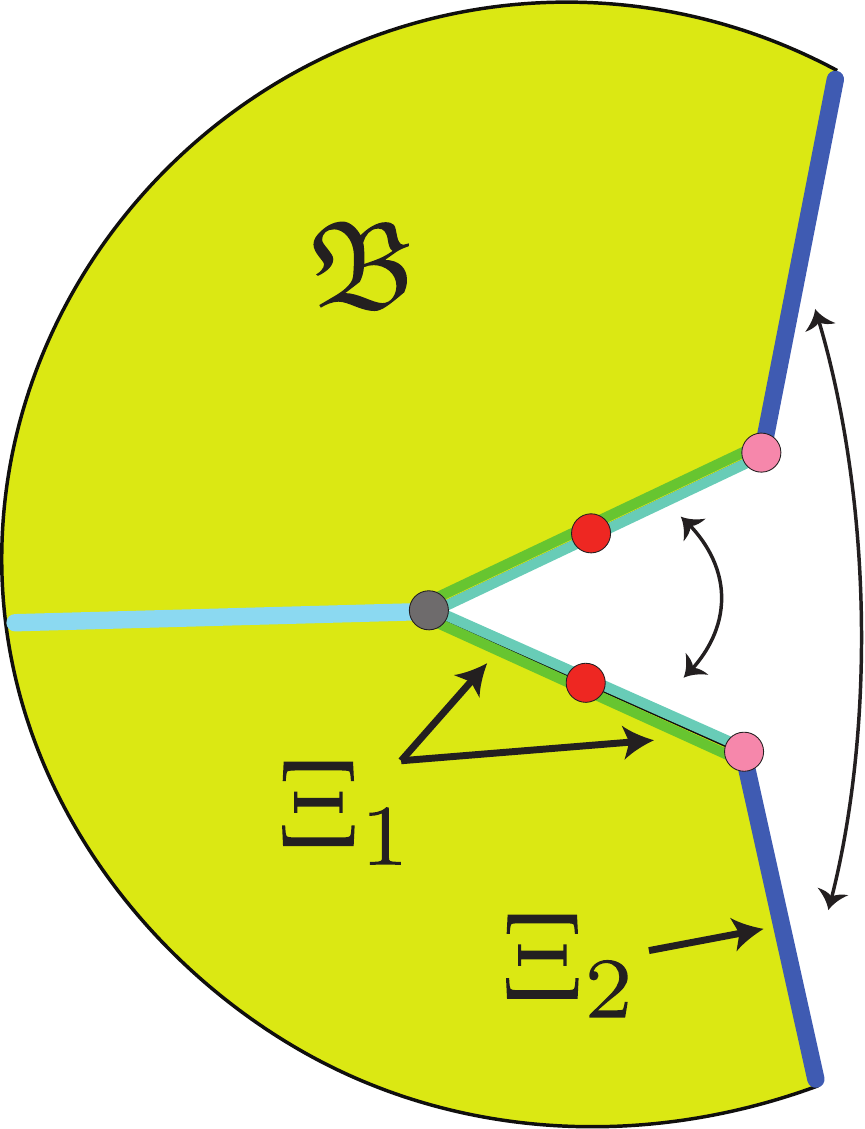}
  \hspace{.1\textwidth}
  \includegraphics[width=.5\textwidth]{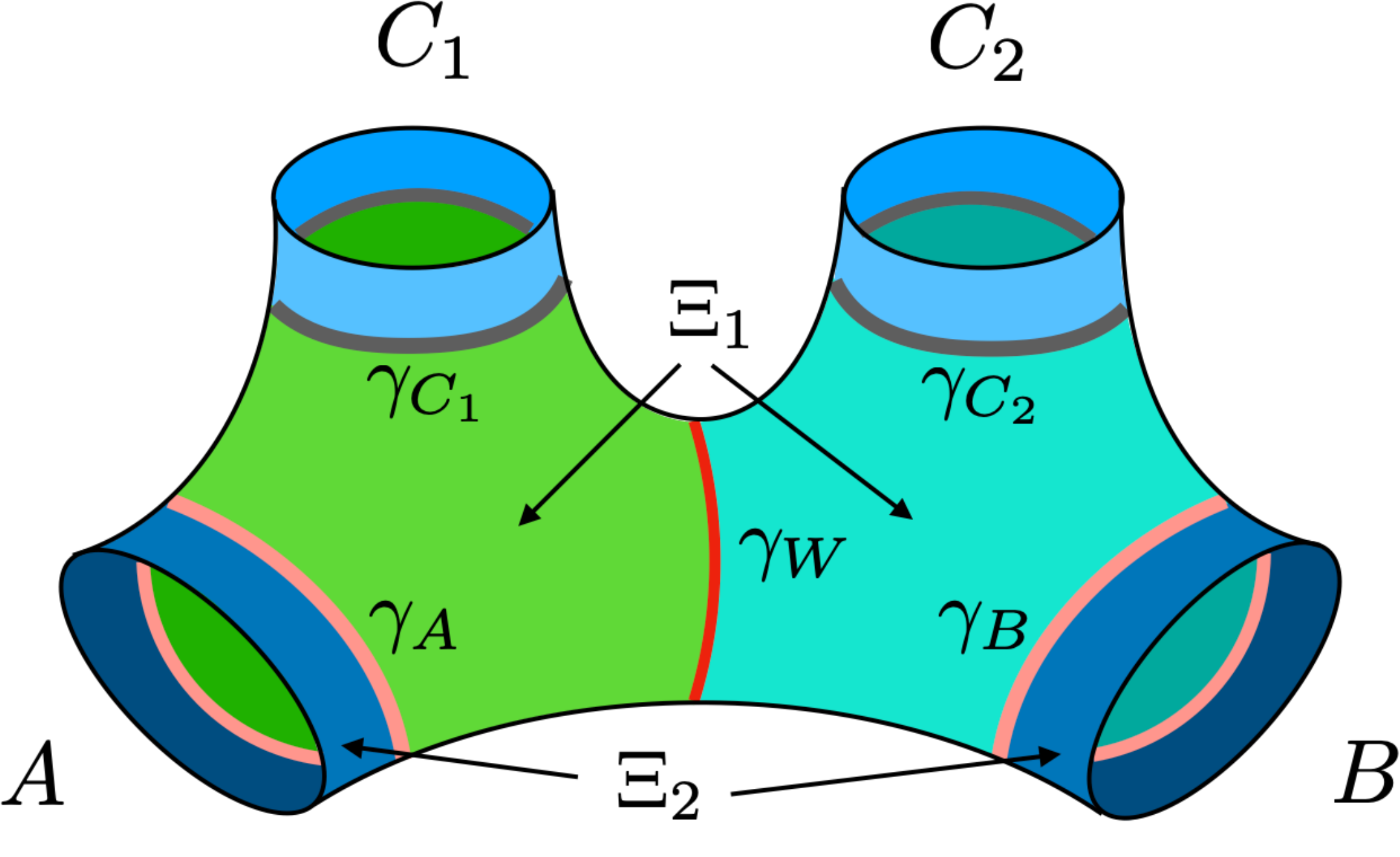}

  \caption{A fixed-area Euclidean saddle $\mathfrak{B}$ (left) computing the norm of the state $\ket{\psi}$ representing a four boundary wormhole (right). Due to fixed-area boundary conditions at the relevant surfaces $\gamma_{A,B,C,W}$, generically the saddle consists of conical defects at these locations. The $\mathbb{Z}_2$ symmetric slice of the saddle is a Cauchy slice $\Xi$ of the wormhole geometry with shaded regions and extremal surfaces represented on either side. The surfaces $\Xi_1$ (light and dark green) and $\Xi_2$ (dark blue) are identified as shown, but can be cut open to use as a building block for the replica trick.}
  \label{fig:fixedarea}
\end{figure}

The 2TN model can be directly translated into a four-boundary wormhole with a hyperbolic metric as shown in \figref{fig:2TN}. First, it is useful to match the degrees of freedom on either side. The parameters in the tensor network are the bond dimensions. On the other hand, the moduli of the wormhole can be understood by a pair-of-pants decomposition of the hyperbolic geometry into two constituent three-boundary wormholes. For each three-boundary wormhole, the moduli are the three horizon areas. Gluing them together removes one degree of freedom due to identification and simultaneously introduces additional Dehn twist moduli. In order to have a reflection symmetric Cauchy slice and be able to apply the RT formula, we can set the twist to zero \cite{Balasubramanian:2014hda}. This leaves us with the areas of the labelled extremal surfaces, each corresponding to a bond in the tensor network. For the calculation of reflected entropy, these are the only surfaces that are relevant.\footnote{There are other possible cross-sections in the geometry which could be minimal. We discuss these in Appendix~\ref{sec:multi} and for the comparison, one can restrict to a regime of parameters where the surface $\gamma_W$ is indeed minimal.} Thus, the 2TN model, despite being a rather coarse-grained description of the geometry, is sufficient to model the four-boundary wormhole accurately.

The construction of these wormhole geometries using a Euclidean path integral is also well understood (see Ref.~\cite{Balasubramanian:2014hda} and references therein). Namely, given a spatial geometry $\Xi$, a corresponding Euclidean spacetime geometry that prepares such initial data is given by
\begin{equation}\label{eq:euc}
	ds^2 = l^2\left(dt_E^2 + \cosh^2 t_E \,d\Xi^2 \right),
\end{equation}
where $l$ is the AdS scale and $t_E$ is the Euclidean time coordinate. However, since we are interested in preparing holographic states that are modelled by an RTN, it is important to fix the areas of the various surfaces that correspond to maximally entangled bonds \cite{Dong:2018seb,Akers:2018fow,Dong:2019piw}. Since these surfaces are all spacelike separated from each other, the areas can be simultaneously fixed. While the geometry in \Eqref{eq:euc} is a valid Euclidean saddle for the fixed-area problem, other ways of preparing the same state would generically contain conical defects at the fixed-area surfaces as shown in \figref{fig:fixedarea}. Moreover, Einstein's equations require the conical defects to be located at extremal surfaces \cite{Dong:2019piw}, and this is true by construction for the wormhole that we're interested in.\footnote{Note that despite the fact that \Eqref{eq:euc} provides a valid saddle, there is no guarantee that it dominates and in general, it is not completely well understood which saddles dominate the path integral \cite{Yin:2007at,Maxfield:2016mwh,Harlow:2018tng}. Our results should be understood to apply when such a dominant Euclidean saddle can be found.}

\begin{figure}[t]
  \centering
  \includegraphics[width=.9\textwidth]{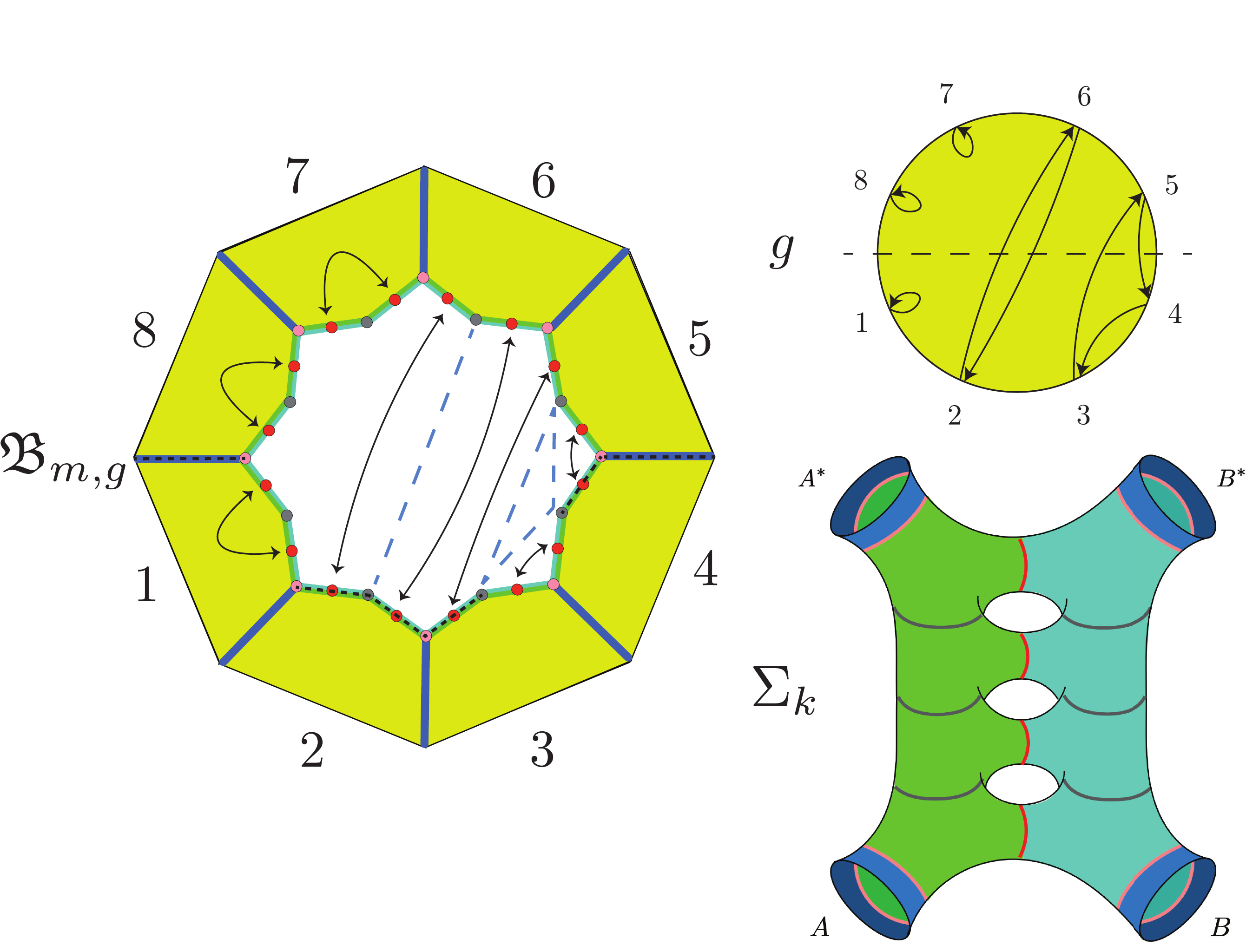}
  \caption{The computation of $\langle \rho_{AB}^{m/2}|\rho_{AB}^{m/2}\rangle=\tr(\rho_{AB}^m)$ (e.g. $m=8$) involves Euclidean saddles constructed by gluing $m$ copies of the original fixed area saddle in different ways. E.g., a particular saddle $\mathfrak{B}_{m,g}$ is in a one-to-one correspondence to permutation $g$ acting on $m$ elements. Here, we depict one such example. The $\mathbb{Z}_2$ symmetric slice (black dotted) is then a Cauchy slice $\Sigma_k$ for Lorentzian evolution and consists of $2 k$ copies of the entanglement wedge, where $k$ is the number of permutation cycles crossing the horizontal dotted line on the right (here $k=2)$.}
  \label{fig:gluing}
\end{figure}

Once we pick any such $\mathbb{Z}_2$ symmetric Euclidean geometry $\mathfrak{B}$, we can cut it open to obtain a preparation of the state $\ket{\psi}$ which has a spatial geometry $\Xi$. As usual, the norm of $\ket{\psi}$ is computed by $\mathfrak{B}$ as depicted in \figref{fig:fixedarea}. Moreover, the advantage of using fixed-area states is that we can find candidate geometries for computing $\tr\left(\rho_{AB}^m\right)$ by using $\mathfrak{B}$ as a building block. More specifically, we can cut open $\mathfrak{B}$ in the region $\Xi_1\cup \Xi_2$ such that $\partial \Xi_1=\gamma_A\cup \gamma_B\cup \gamma_{C_1} \cup \gamma_{C_2}$ and $\partial \Xi_2 = A\cup B \cup \gamma_A \cup \gamma_B$. We can then glue together $m$ copies cyclically in the region $\Xi_2$ as shown in \figref{fig:gluing}. We are then left with picking a way to glue together the remaining section $\Xi_1$. The different ways of gluing $\Xi_1$ are fixed by an element of the permutation group $S_m$. An example of this correspondence is demonstrated in \figref{fig:gluing}.\footnote{In principle, we have two independent permutations on the two different portions separated by the fixed-area surface $\gamma_W$ in correspondence with the two tensors in the RTN. However, saddles with different permutations are suppressed due to the cost of having a domain wall at $\gamma_W$.} The fixed-area boundary condition ensures that all the contributing geometries solve Einstein's equations and satisfy the correct boundary conditions. Any such geometry $B_{m,g}$, obtained by gluing in a manner corresponding to a permutation $g$, contributes to the computation of $\tr(\rho_{AB}^m)$. It is also easy to check that the Euclidean action agrees with the domain wall cost in the RTN \cite{Dong:2018seb}. Thus, it is clear that there is a direct correspondence between the RTN replica partition function and the gravitational path integral, which is a fact that has already been exploited in various calculations \cite{Dong:2018seb,Akers:2018fow,Akers:2020pmf,Marolf:2020vsi,Dong:2020iod}. 

We can now look at $B_{m,g}$ for even integer $m$, and interpret it as computing the norm of the state $\ket{\rho_{AB}^{m/2}}$. It is well-known from the correspondence to the RTN that only the saddles corresponding to non-crossing permutations are important \cite{Penington:2019kki,Akers:2020pmf,Marolf:2020vsi,Dong:2020iod}, thus we will neglect all other permutations. Once we do so, it is useful to classify the non-crossing permutations by the number of cycles crossing the horizontal middle line, which is representative of the $\mathbb{Z}_2$ symmetric Cauchy slice of the geometry.\footnote{The slice is only locally $\mathbb{Z}_2$ symmetric, but not globally so in general. The overall $\mathbb{Z}_2$ symmetry is restored by the sum over saddles.} For a crossing number $k$, it can be checked that the spatial geometry $\Sigma_k$ is obtained by gluing together sections of $2k$ copies of the original $\Xi$ at the horizons. We illustrate one such example in \figref{fig:gluing}. As described before, these are also hyperbolic geometries with a different pair-of-pants decomposition and by construction are prepared by a Euclidean path integral which solves Einstein equations with fixed-area boundary conditions.\footnote{It is important to note that in general there are also perturbative corrections which have completely neglected in this analysis since there is no corresponding feature in the RTN. Thus, even at this level the RTN only captures certain topological aspects of the gravitational path integral, but they are usually the important non-perturbative corrections near phase transitions.} Given the Cauchy data on $\Xi$ that satisfies the constraint equations, one can then evolve it in Lorentzian time to find the full spacetime geometry.\footnote{Depending on how sharply the areas have been fixed, there may or may not exist a smooth spacetime to the future and past of these fixed-area surfaces (see Ref.~\cite{Dong:2022ilf} for more details).} 

Since the spatial geometries on the $\mathbb{Z}_2$ symmetric slice are different for different values of $k$, the states are orthogonal in the gravitational path integral approximation. Thus, we can divide the path integral into a sum over the index $k$ as
\begin{equation}
	Z_m = \sum_{k=0}^{m/2} Z_{k,m},
\end{equation}
and for each value of $k$, we can write down
\begin{equation}
	Z_{k,m} = \langle \psi_{k,m} | \psi_{k,m}\rangle.
\end{equation}
The state $\ket{\psi_{k,m}}$ has an associated geometry $\Sigma_k$ and can be written as a superposition over different half-permutations with appropriate weight-factors that contribute to its norm. For example, we have (for $m=4$)
\begin{equation}\label{eq:super}
\begin{matrix}
\includegraphics[scale=0.5]{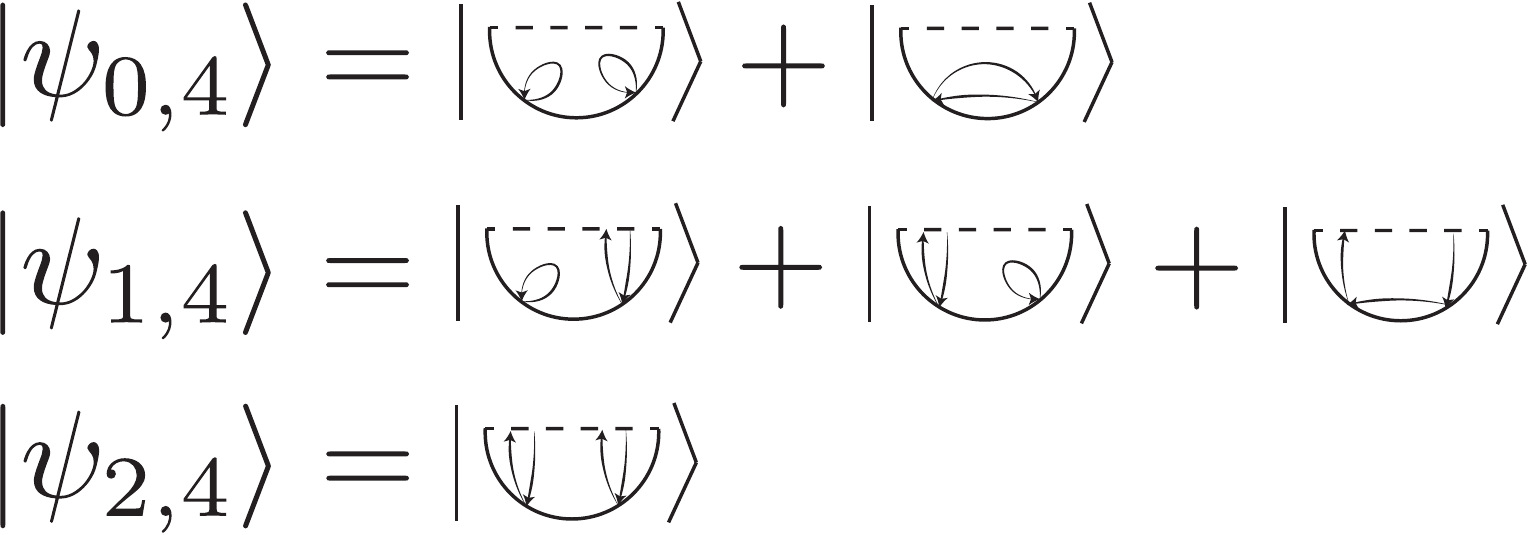},
\end{matrix}
\end{equation}
where the states written above are not normalized and the overlaps of these states can be computed by closing up the open ends of the permutation and computing the action of the corresponding Euclidean saddle. We will later see in \secref{sub:large} that the states $\ket{\psi_{k,m}}$ are naturally associated with specific states in the standard module of the Temperley-Lieb algebra that dominate in the large $\chi$ limit. Thus, we now have written the state $\ket{\rho_{AB}^{m/2}}$ in terms of a superposition of geometries with appropriate weights determined by the path integral.

The spatial geometry $\Sigma_k$ consists of $2k$ copies of the surface $\gamma_W$, which is the entanglement wedge cross section in the connected phase. Now we can use the fact that the entropy of a superposition of a small number of fixed-area states behaves approximately linearly as argued in Refs.~\cite{Almheiri:2016blp,Marolf:2020vsi,Akers:2020pmf}. Thus, the analog of reflected entropy for the state $\ket{\rho_{AB}^{m/2}}$ is computed by averaging over the different sectors, i.e.,
\begin{align}\label{eq:avg}
\begin{split}
	S_R^{(m)}(A:B) &= S(\rho_{AA^*})_{\ket{\rho_{AB}^{m/2}}}\\
	&=\sum_{k=0}^{m/2} p_k(m) (2k EW(A:B)) - p_k(m) \log p_k(m),
 \end{split}
\end{align}
where $EW(A:B)$ is used as a shorthand to represent the area of $\gamma_W$, independent of which phase dominates. The weights $p_k(m)=\langle \psi_{k,m} |\psi_{k,m}\rangle$ can be thought of as probability weights associated to the geometries $\Sigma_k$. We postpone the precise formulae for $p_k(m)$ to \secref{sub:large}, but note for now that we will find a function that is analytic in $m$. By extending the sum in \Eqref{eq:avg} to all integer $k$, we can thus analytically continue the answer away from even integer values of $m$. Doing so, we obtain an expression for the reflected entropy,
\begin{equation}\label{eq:refgrav}
    S_R(A:B) = \sum_{k=0}^{\infty} p_k (2k EW(A:B)) - p_k \log p_k,
\end{equation}
where $p_k\equiv p_k(1)$. Thus, we have found that the canonical purification is given by a superposition of geometries obtained by gluing together multiple copies of the entanglement wedge of $AB$ as demonstrated in \figref{fig:genus}. As we shall see later, the different geometries $\Sigma_k$ obtained by this construction can equivalently be interpreted in terms of new RTNs which are obtained by gluing multiple copies of the network at different bonds, see \figref{fig:interpret}. Thus, this provides a refined version of the effective description suggested in Ref.~\cite{Akers:2021pvd} for the canonical purification.


\section{Setup} 
\label{sec:setup}

In this section, we set up the problem of finding the reflected entropy in our model of interest, the 2TN model. \secref{sub:rtns} describes the replica trick for reflected entropy, discussing the relevant saddle point configurations. \secref{sub:TL} then sets up the problem of computing the resolvent for the reflected density matrix by introducing the Temperley-Lieb algebra. 

\subsection{Replica Trick for the 2TN Model} 
\label{sub:rtns}

RTNs prepare boundary states by choosing random tensors at the vertices of an arbitrary graph $G=(V,E)$ and contracting them by projecting onto maximally entangled pairs along the edges of $G$. Namely, the unnormalized state takes the form
\begin{equation}
	\ket{\psi} = \left(\otimes_{\{x,y\}\in E} \bra{xy}\right) \left(\otimes_{x\in V} \ket{V_x}\right),
\end{equation}
where $\ket{V_x}$ is a Haar random state chosen at vertex $x$ and $\ket{xy}$ represent maximally entangled pairs along the edge connecting vertices $x$ and $y$. 

The replica trick for the reflected entropy then involves computing the $(m,n)$ R\'enyi reflected entropy defined as \cite{Dutta:2019gen}
\begin{align}\label{eq:SRmn}
	S_R^{(m,n)}(A:B) =S_n(A A^*)_{\ket{\rho_{AB}^{m/2}}},	
\end{align}
which is the $n$th R\'enyi entropy of subregion $AA^*$ in a state $\ket{\rho_{AB}^{m/2}}$ that generalizes the canonical purification.\footnote{Note that we ignored normalization factors in the state for simplicity of notation.} For integer $n$ and $\frac{m}{2}$, \Eqref{eq:SRmn} can be computed using correlation functions of appropriately defined twist operators on $mn$ copies of the system \cite{Dutta:2019gen,Akers:2021pvd}, i.e.,
\begin{align}
	\label{eq:replica}
 \begin{split}
	S_R^{(m,n)}(A:B) &= \frac{1}{1-n}\log \left(\frac{Z_{m,n}}{\left(Z_{m,1}\right)^n}\right)\\
		Z_{m,n} &= \bra{\psi}^{\otimes mn} \Sigma_A(g_A) \Sigma_B(g_B) \ket{\psi}^{\otimes mn}\\
	Z_{m,1}&=\bra{\psi}^{\otimes m} \Sigma_{AB}(\tau_m) \ket{\psi}^{\otimes m},
 \end{split}
\end{align}
where $\Sigma_R(g)$ implement the permutation $g$ on subregion $R$. The permutation $\tau_m$ is a cyclic permutation on $m$ elements, while the permutations $g_A,g_B$ are depicted in \figref{fig:gAB}.

\begin{figure}[t]
  \centering
  \includegraphics[width=.7\textwidth]{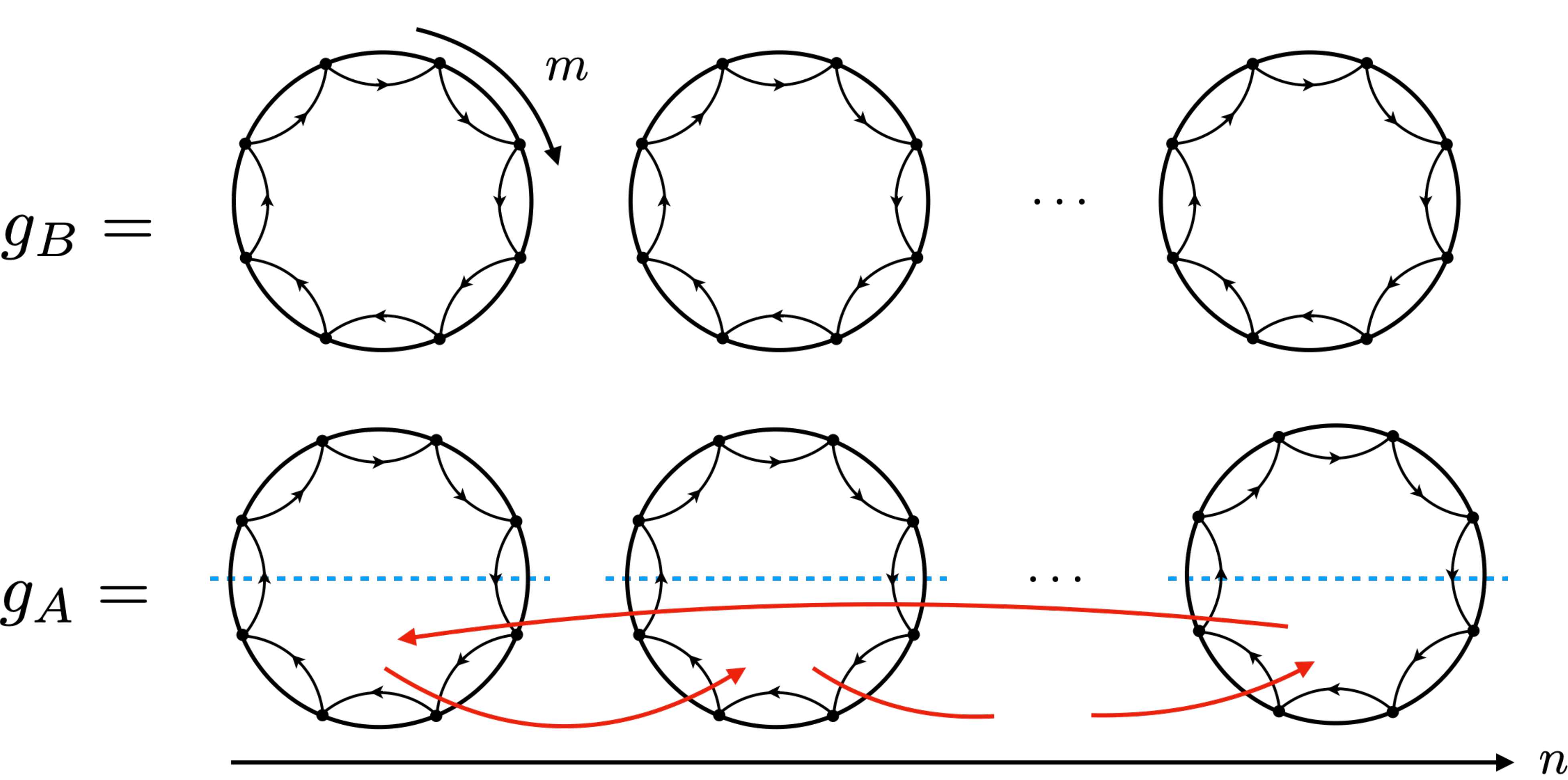}
  \caption{A graphical representation of the elements $g_A$ and $g_B$. Each circle represents $m$ replicas of the original tensor and each circle is replicated $n$ times. A cycle of the permutation is represented by a directed closed loop. The element $g_B$ is a product of individual cyclic permutations $\tau_m$ on each circle. The element $g_A$ is additionally conjugated by the twist operator $\gamma_\tau$ (indicated by red arrows) whose action can be thought of as slicing each circle in half and shifting the lower half in an $n$-cyclic order.}
  \label{fig:gAB}
\end{figure}

Further, for an RTN state $\ket{\psi}$, \Eqref{eq:replica} can be computed on average using Haar integrals to reduce $Z_{mn}$ to a sum over the permutation group $S_{mn}$ \cite{Akers:2021pvd}, i.e.,
\begin{equation}\label{eq:Zmn}
	\overline{Z_{m,n}} = \sum_{\substack{g_x\in S_{mn}\\
                  x\in \{V\backslash\partial\}}} \exp\left[-\sum_{\{x,y\}\in E} (\ln \chi_{xy}) \, d(g_x,g_y)\right],
\end{equation}
where $\partial$ represents the boundary vertices. The permutation group elements on $\partial$ are fixed as boundary conditions to be $g_A$ on subregion A, $g_B$ on subregion $B$ and $e$ (the identity element) on subregion $C$. The exponent is a sum over the Cayley distances\footnote{The Cayley distance $d(g_x,g_y)$ is a metric on $S_{mn}$ measuring the number of two-element swaps required to go from $g_x$ to $g_y$.} $d(g_x,g_y)$ weighted by the bond dimensions of each edge, denoted $\chi_{xy}$.

We can now specialize to the 2TN model, which is depicted in \figref{fig:2TN}. It models a four-boundary wormhole with horizon areas $\mathcal{A}_i$ which are related to the external bond dimensions $\chi_i=\exp[\frac{\mathcal{A}_i}{4G}]$, where $i\in\{A,B,C_1,C_2\}$. Moreover, the internal bond dimension $\chi$ is related to the area $\mathcal{A}_W$ of the cross section surface $\gamma_W$ in a similar fashion. For convenience, we denote $\chi_{C_1}\chi_{C_2}=\chi_C$, which is related to the total horizon area for region $C$.

When $\min(\mathcal{A}_{A},\mathcal{A}_{B})<\mathcal{A}_W$, the analysis is similar to that done for the single random tensor which was studied in detail in Ref.~\cite{Akers:2021pvd}. Thus, here we will only be interested in the situation where the minimal entanglement wedge cross section $EW(A:B)$ is indeed given by $\mathcal{A}_W$ in the connected phase. We will ensure this to be the case by picking $\mathcal{A}_W$ to be parametrically smaller than $\mathcal{A}_{A/B/C}$. 

For this problem, there are two limits of interest that we can discuss:
\begin{itemize}
    \item The Temperley-Lieb (TL) limit: $q_A=\frac{\chi_A}{\chi_{C_1}}$,$q_B=\frac{\chi_B}{\chi_{C_2}}$, $\chi$ held finite while taking all external bond dimensions $\chi_{A/B/C_1/C_2}\rightarrow \infty$.
    \item The saddle point limit: $\frac{\log \chi_i}{\log \chi_{C}}$ is held finite while taking all bond dimensions $\chi_i\rightarrow \infty$.
\end{itemize}

\begin{figure}[t]
    \centering
    \includegraphics[width=.9\textwidth]{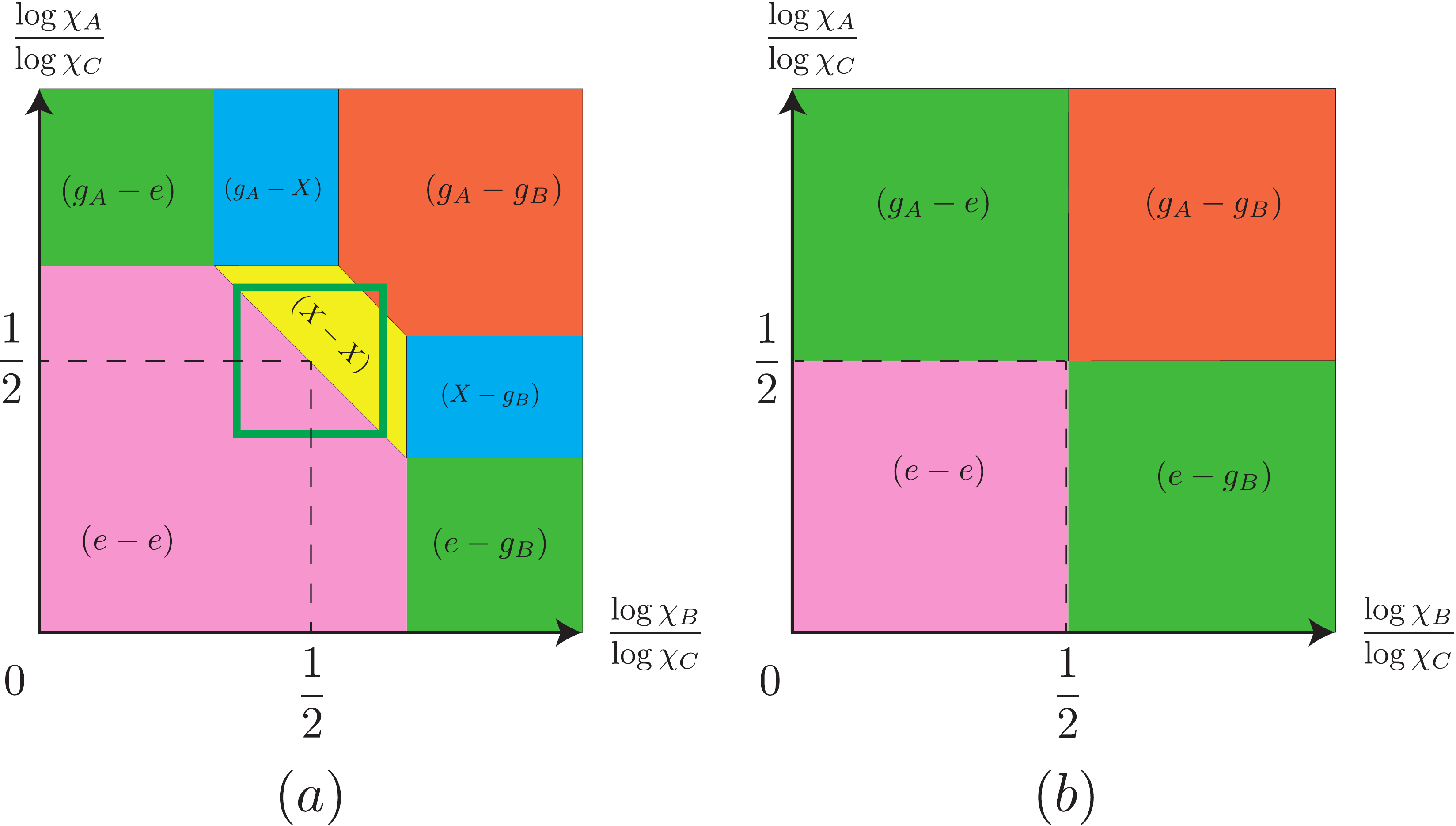}
    \caption{(a): An illustration of the phase diagram in the saddle point limit for 2TN for generic values of $m \geq 2, n\geq 1,\chi$ as a function of $\frac{\log \chi_A}{\log \chi_{C}}$ and $\frac{\log \chi_B}{\log \chi_{C}}$. The dominant elements in each phase are shown in a tuple as $(g_1,g_2)$. The green square roughly indicates the region of the phase diagram accessed by the TL limit. (b): As we take the limit $\frac{\log \chi}{\log \chi_{C}}\rightarrow 0$, the yellow and blue domains shrink in the phase diagram.
    }
    \label{fig:2TNphase}
\end{figure}

While our calculations in this paper will be in the TL limit, we briefly make a detour to discuss the phase diagram for the 2TN model in the saddle point limit. The saddle point limit is useful since it can be shown that at every point in phase space, one only needs to optimize over the set $\{e,g_A,g_B,X\}$\footnote{$X$ is the unique element that lies on the commmon geodesic $\Gamma(g_A,e)\cap \Gamma(g_B,e)$ while being closest to $g_A$ and $g_B$. Please refer to Ref.~\cite{Akers:2021pvd} for a detailed discussion on element $X$.} \cite{newpaper}.

For simplicity, we consider the phase diagram in the case $\chi_{C_1}=\chi_{C_2}=\sqrt{\chi_C}$. Doing so, we obtain the phase diagram shown in Fig.~(\ref{fig:2TNphase}a) at generic values of $m\geq 2$ and $n\geq 1$. The first thing to note is that the phase diagram is convex as expected from the fact that the contributions of each possible permutation tuple is linear in $\log \chi_{A/B}$. 

The TL limit approximately arises when we take the limit $\frac{\log \chi}{\log \chi_{C}}\rightarrow 0$ and zoom into a small, restricted region around $\left(\frac{1}{2},\frac{1}{2}\right)$. This follows from the fact that the ratios $q_A$ and $q_B$ are held finite. In the TL limit, the domains involving the $X$ element shrink and we obtain the simpler phase diagram shown in Fig.~(\ref{fig:2TNphase}b). 

\begin{figure}[t]
  \centering
  \includegraphics[width=.7\textwidth]{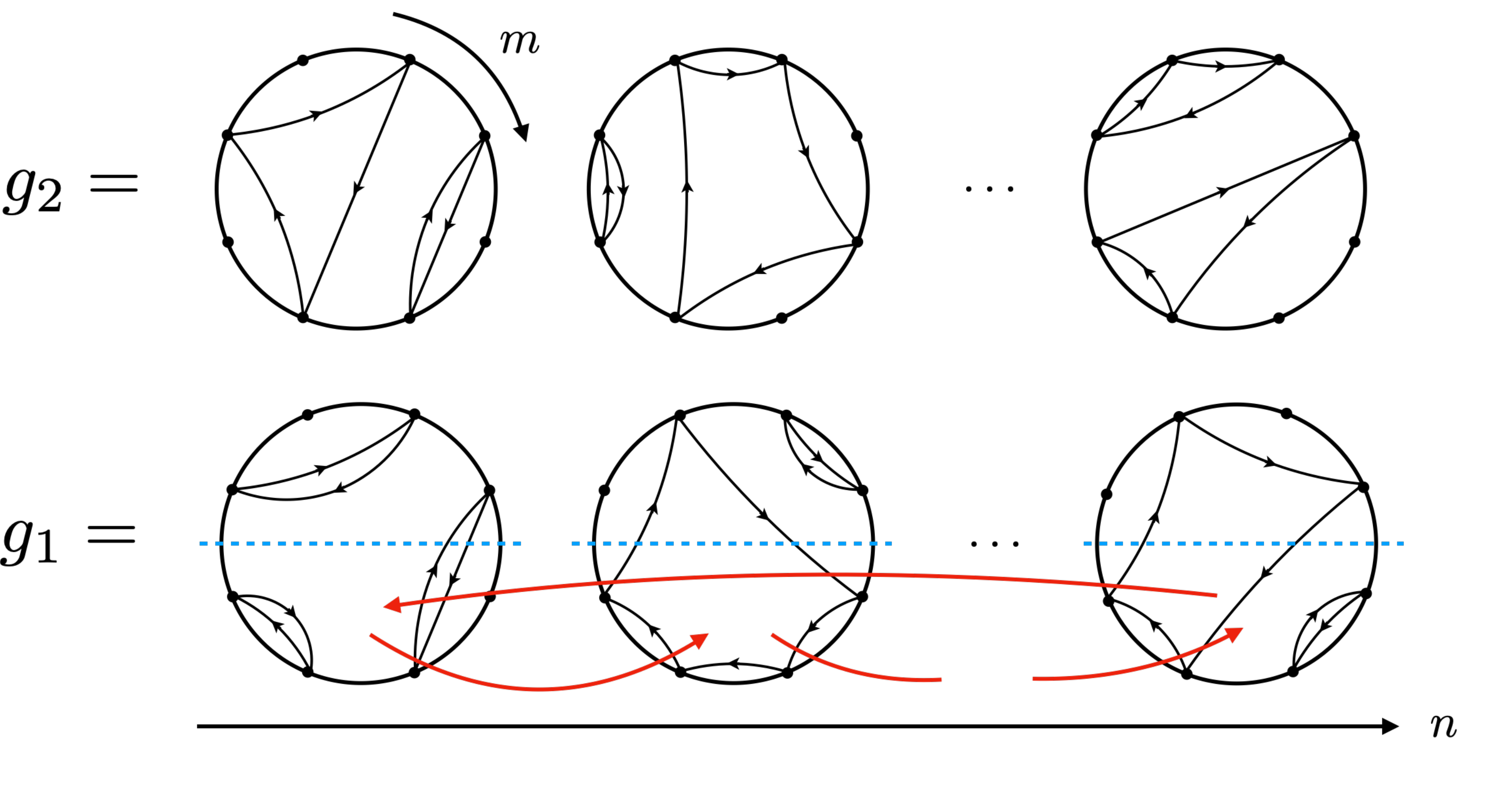}
  \caption{A graphical representation of the elements lying on the geodesics $\Gamma(g_{A},e)$ and $\Gamma(g_{B},e)$. An element $g_2 \in \Gamma(g_{B},e)$ consists of products of non-crossing permutations acting within each circle; whereas an element $g_1 \in \Gamma(g_{A},e)$ is similar to $g_2$, but is additionally conjugated by the twist operator $\gamma_\tau$.}
  \label{fig:g1g2}
\end{figure}

Returning to the calculation in \Eqref{eq:Zmn} for the 2TN model, we have
\begin{equation}\label{eq:partn}
  \overline{Z_{m,n}} = \frac{1}{(\chi_A\chi_B\chi_{C_1}\chi_{C_2}\chi)^{mn}}\sum_{g_1,g_2\in S_{mn}}\chi_A^{\#(g_1g_A^{-1})}\chi_B^{\#(g_2g_B^{-1})}\chi_{C_1}^{\#(g_1)}\chi_{C_2}^{\#(g_2)}\chi^{\#(g_1g_2^{-1})},
\end{equation}
where $\#(g)$ counts the number of cycles in permutation $g$ (including trivial ones) and we have used the relation $d(g,h)=n m -\#(g h^{-1})$. In the TL limit, the elements $g_{1/2}$ are then constrained to lie on the geodesics between $g_{A/B}$ and $e$, labelled by $\Gamma(g_{A/B},e)$, since the contributions of all other elements are infinitely suppressed. The relevant geodesic condition reads
\begin{equation}
  \#(g_1g_A^{-1})+\#(g_1) = n(m+1),  \quad   \#(g_2g_B^{-1})+\#(g_2) = n(m+1),
\end{equation}
and the permutation elements that satisfy these conditions can be parameterized as:
\begin{equation}\label{eq:par}
     g_1 = \gamma_\tau \left( \prod^n_{i=1}  h_i\right)\gamma_\tau^{-1}, \quad g_2 = \prod^n_{i=1}k_i,
\end{equation}
where $h_i, k_i$ are permutations that act in a non-crossing fashion on the elements $m(i-1)+1,\cdots,mi$ and trivially on all other elements. $\gamma_{\tau}$ can be thought of as a ``twist operator'' that acts on the lower $m/2$ elements in each circle, cyclically permuting them as shown in \figref{fig:g1g2}. We refer the reader to Ref.~\cite{Akers:2021pvd} for more details on the geodesic elements and non-crossing permutations.


\subsection{Resolvent via the Temperley-Lieb Algebra} 
\label{sub:TL}

Now, restricting the sum to the permutations on the relevant geodesics, we can write $Z_{mn}$ using the parametrization of \Eqref{eq:par} as
\begin{equation}\label{eq:Zmn_2TN}
  \overline{Z_{m,n}} =  \left(\frac{\chi_A\chi_B}{\chi_C^{m}\chi^m}\right)^n \sum_{\{h_i,k_i\}} q_A^{-\sum_i\#(h_i)} q_B^{-\sum_i\#(k_i)}\chi^{\#(g_1g_2^{-1})}.
\end{equation}
To proceed, we must find a way to express $\#(g_1g_2^{-1})$ solely in terms of $h_i$ and $k_i$. To attack this problem we introduce the \textit{Temperley-Lieb algebra}, denoted by $\text{TL}_m$. 

The Temperley-Lieb (TL) algebra is an abstract algebra with basis vectors consisting of diagrams of non-crossing strands between $2m$ points arranged on two vertical lines as follows:
\begin{equation}
  \begin{matrix}
    \includegraphics[width=.55\textwidth]{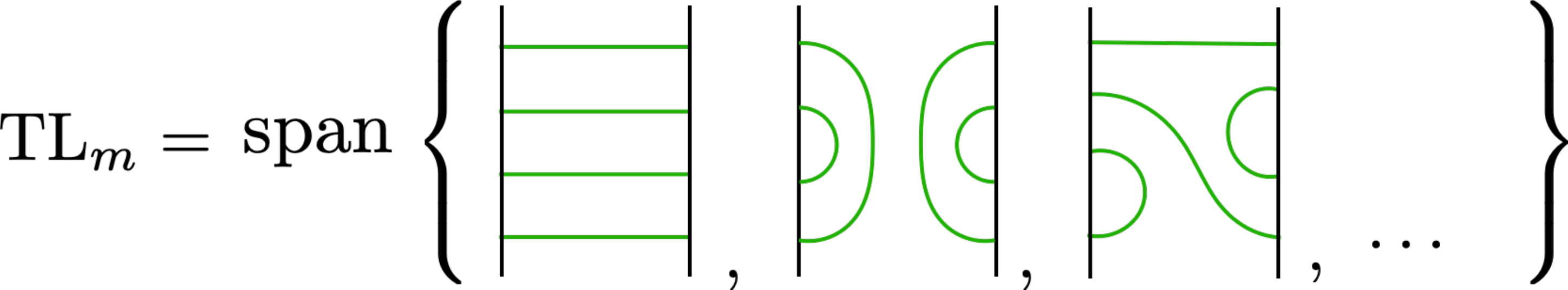},
  \end{matrix}
\end{equation}
where each point necessarily has a strand emerging from it. This vector space is further endowed with a bilinear product given by concatenating two diagrams and replacing each closed loop by a power of $\chi$. As an illustrative example, consider the following:
\begin{equation}
  \begin{matrix}
     \includegraphics[width=.55\textwidth]{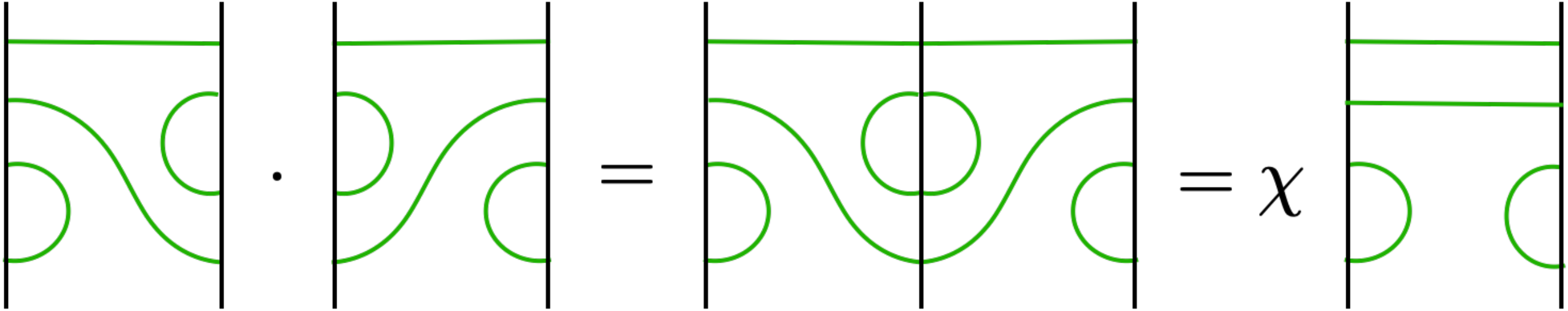}   
  \end{matrix}
\end{equation}
In our problem, the parameter $\chi$ in the algebra is chosen to be the same as the internal bond dimension in the 2TN model. 
For the interested reader, we present a short review of the TL algebra in Appendix~\ref{sec:TLalgebra}. 

For our purposes, the important point is that there is a natural one-to-one correspondence between the set of non-crossing permutations $\text{NC}_m$ and the elements of $\text{TL}_m$ as shown pictorially in \figref{fig:NC-TL}. We denote the element in $\text{TL}_m$ corresponding to $h \in \text{NC}_m$ as $D(h)$.
There is also a natural trace $\text{Tr}_{\text{TL}_m}$ on $\text{TL}_m$, defined diagrammatically for basis elements by closing the strands on opposing ends and assigning the value $\chi^{\#\text{loops}}$, e.g.
\begin{equation}
  \begin{matrix}
     \includegraphics[scale=.35]{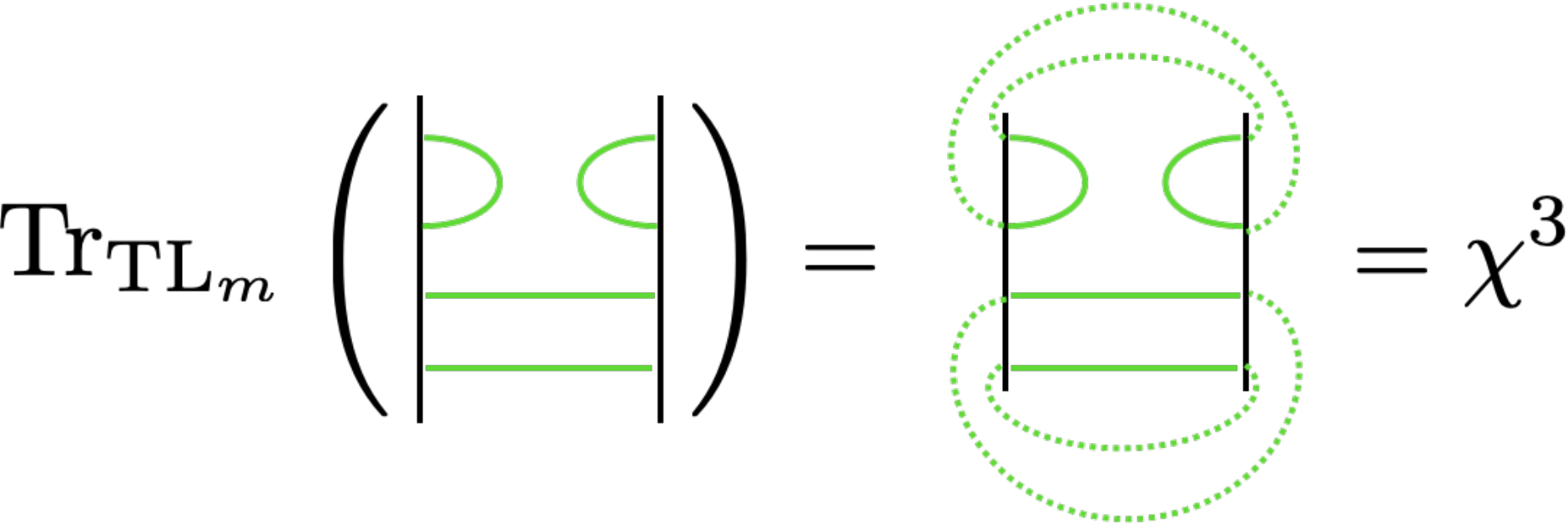}   
  \end{matrix}
\end{equation}
The trace is then linearly extended to the full algebra.

\begin{figure}[t]
  \centering
  \includegraphics[width=.6\textwidth]{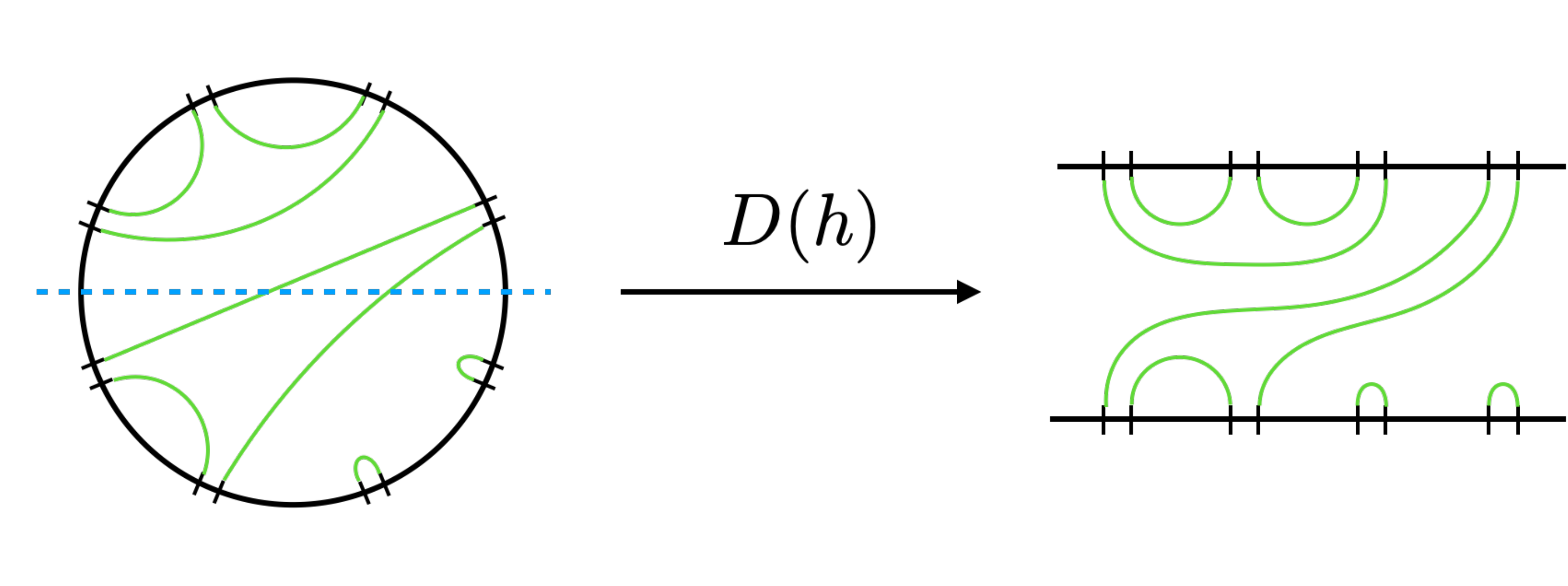}
  \caption{A non-crossing permutation is represented in double line notation, where there is an incoming and outgoing line at each vertex. The mapping from $h\in \text{NC}_m$ to $D(h)\in\text{TL}_m$ can be thought of as slicing the circle in half and straightening the two boundary arcs into lines while preserving the connections.}
  \label{fig:NC-TL}
\end{figure}

We now claim that
\begin{equation}\label{eq:bondTL}
  \chi^{\#(g_1g_2^{-1})} = \text{Tr}_{\text{TL}_m}\left[ D(h_1)^TD(k_1)D(h_2)^TD(k_2)\cdots D(h_n)^TD(k_n) \right],
\end{equation}
where $D(h)^T$ denotes the transpose of $D(h)$, obtained by flipping the diagram across its central axis between two boundary lines. Instead of providing a formal proof, we demonstrate the above statement via pictures. For example, for $(m,n)=(4,2)$, we have
\begin{equation}
  \#(g_1g_2^{-1}) = \#\left(\gamma_\tau h_1h_2\gamma_\tau^{-1}k_2^{-1}k_1^{-1}\right).
\end{equation}
Then, as a sample configuration consider:
\begin{equation}
  \begin{matrix}
    \includegraphics[width=.55\textwidth]{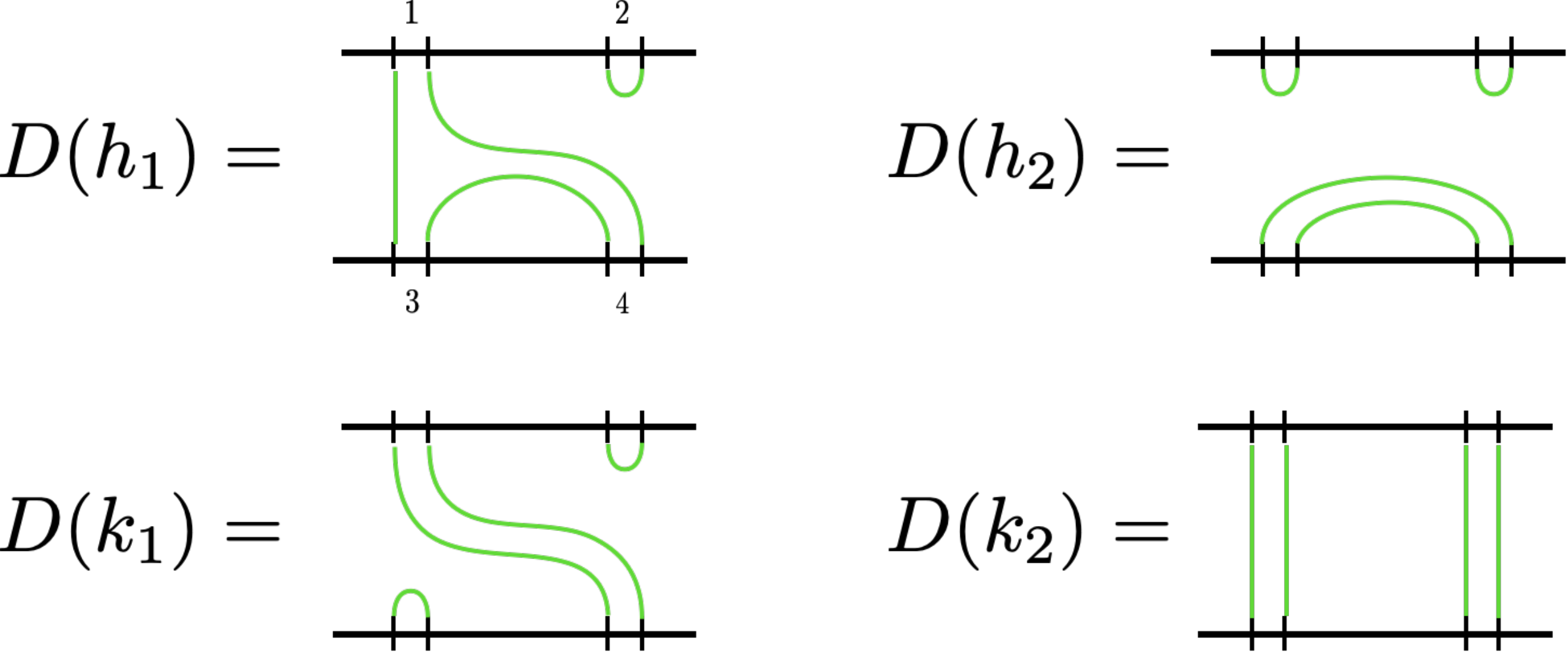}  .
  \end{matrix}
\end{equation}
For the above configuration, we have
\begin{equation}
  \begin{matrix}
    \includegraphics[width=.5\textwidth]{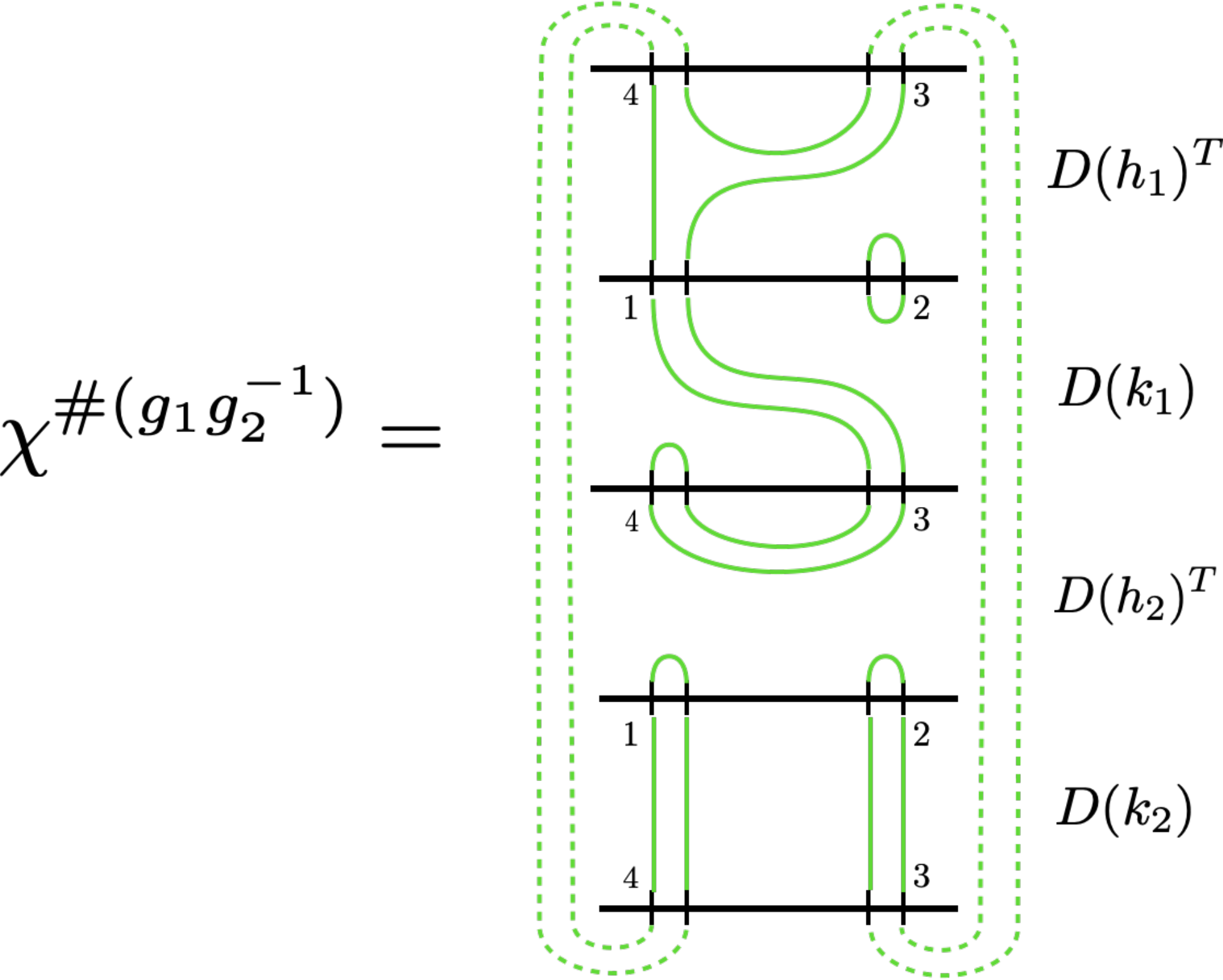},
  \end{matrix}
\end{equation}
as the reader can easily verify. This generalizes in a straightforward manner to arbitrary $(m,n)$.

Using this result we can then do the sum in \Eqref{eq:Zmn_2TN}:
\begin{equation}\label{eq:Zmn_TL}
  \overline{Z_{m,n}} = \left(\frac{\chi_A\chi_B}{\chi_C^{m}\chi^m}\right)^n \text{Tr}_{\text{TL}_m}
  \left[\left(\sum_{h\in \text{NC}_m}D(h)q_A^{-\#(h)}\right)\left(\sum_{k\in \text{NC}_m}D(k)q_B^{-\#(k)}\right)\right]^n
\end{equation}
where we have used the fact that the $h$ sum is invariant under $D(h)\to D(h)^T$ to drop the transpose. 

We will now use \Eqref{eq:Zmn_TL} to obtain the reflected entanglement spectrum via the resolvent trick \cite{Akers:2021pvd,Akers:2022max}. The resolvent for the reflected density matrix $\rho^{(m)}_{AA^*}$ is defined as
\begin{equation}\label{eq:res}
  R_m(\lambda) = \sum^\infty_{n=0}\frac{\overline{Z_{m,n}}}{\lambda^{1+n}}
\end{equation}
where $Z_{m,0} = \chi^2_A$ and $Z_{m,1} = \text{Tr}\rho^m_{AB}$ (which can also be checked from \Eqref{eq:Zmn_TL}). Plugging \Eqref{eq:Zmn_TL} into \Eqref{eq:res} we obtain a formal expression for the resolvent:
\begin{equation}
  \label{eq:TL_resolvent}
  R_m(\lambda) =  \text{Tr}_{\text{TL}_m}
  \left[ \lambda-\frac{\chi_A\chi_B}{\chi_C^{m}\chi^m}\left(\sum_{h\in \text{NC}_m}D(h)q_A^{-\#(h)}\right)\left(\sum_{k\in \text{NC}_m}D(k)q_B^{-\#(k)}\right)\right]^{-1}.
\end{equation}
Once we evaluate the resolvent, one can extract the eigenvalue spectrum $D_m(\lambda)$ of $\rho^{(m)}_{AA^*}$ from it using:
\begin{align}
	D_m(\lambda) &= -\frac{1}{\pi}\lim_{\epsilon \rightarrow 0} \text{Im} \,R_m(\lambda+i\epsilon).
\end{align}
From the spectrum, one can obtain all the $(m,n)$-R\'enyi reflected entropies as well as the reflected entropy after analytically continuing to $m=1$.

Now, in order to evaluate \Eqref{eq:TL_resolvent} we will pick a representation of $\text{TL}_m$. However, an arbitrary representation will not do the job for us. In particular, we must pick a representation such that the trace function on this representation correctly reproduces $\text{Tr}_{\text{TL}_m}$ defined above.
Such a representation will in general be decomposed into a direct sum of irreducible representations (irreps).
Thus, we expect the reflected entanglement spectrum to be grouped into different ``sectors'' labeled by these irreps. This is precisely the form of spectrum we will find for the 2TN model. 

There are many ways of classifying the irreps of the TL algebra. Here, we will make use of the \textit{standard module} \cite{Ridout:2012gg}. 
It is defined by considering vector spaces with basis vectors called \textit{link states}, made by cutting the basis elements of the $\text{TL}_m$ algebra into half. Cutting a diagram in half will always expose an even number of ``defects'' (since $m$ is even in our case), which we label by $2k$ with $0\le k\le m/2$. We call such a diagram with $2k$ defects a \emph{$(m,k)$-link state} and the set of all $(m,k)$-link states $\mathfrak{B}^{(m)}_k$.
The vector space spanned by $\mathfrak{B}^{(m)}_k$ is denoted by $\mathcal{V}^{(m)}_k$. 

There is a natural action of TL diagrams on $(m,k)$-link states given by concatenating and replacing closed loops with powers of $\chi$. Such a concatenation may result in a number of disconnected strands that decrease the number of defects. We will further require that the action maps the link state to zero whenever there are any disconnected strands after concatenation. For example:
\begin{equation}
  \begin{matrix}
     \includegraphics[width=.5\textwidth]{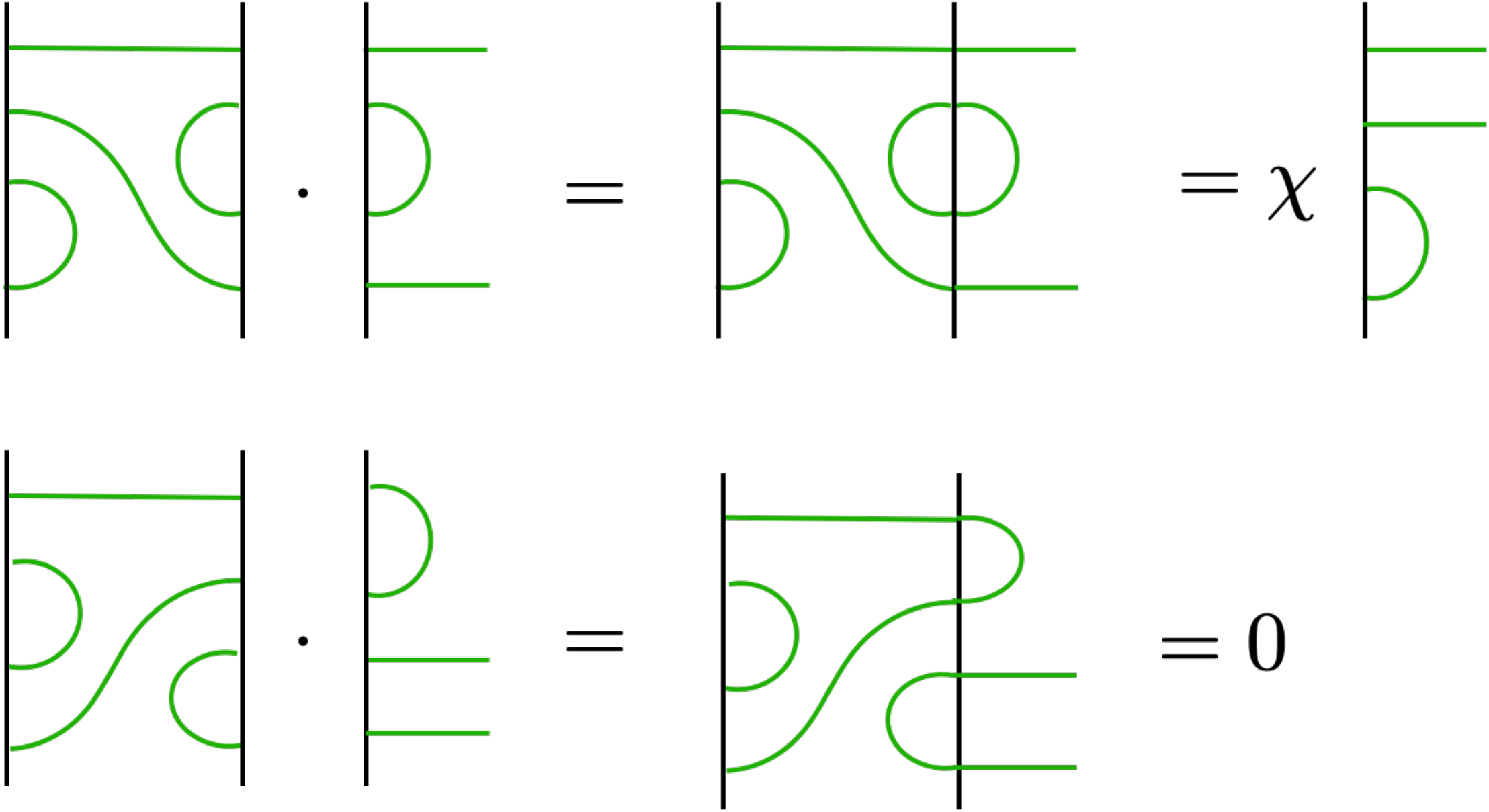} 
  \end{matrix}
\end{equation}
This action then defines a representation of $\text{TL}_m$ on $\mathcal{V}^{(m)}_k$, called the \emph{standard module}.
The usefulness of the standard modules comes from the fact that they classify all the finite dimensional irreducible representations of $\text{TL}_m$.\footnote{There are certain values of $\chi$ where the standard modules are reducible. However those values are discrete and only exist for $\chi < 2$, thus, irrelevant for our problem of interest.}
We will denote $\pi^{(m)}_k(t)$ to be the matrix representation of $t\in\text{TL}_m$ associated with $\mathcal{V}^{(m)}_k$. For more details on the standard module, we refer the reader to Appendix~\ref{sec:TLalgebra}. 

Our goal is to build a finite dimensional representation
\begin{equation}
  \mathcal{V}^{(m)} = \bigoplus_{k=0,1,\cdots,m/2} d^{(m)}_k\mathcal{V}^{(m)}_k,
\end{equation}
where the coefficients $d^{(m)}_k$ represent the number of times the irrep $\mathcal{V}^{(m)}_k$ appears in $\mathcal{V}^{(m)}$. 
These coefficients are uniquely determined by demanding that the trace of $\mathcal{V}^{(m)}$ agree with the trace in the TL algebra $\text{Tr}_{\text{TL}_m}$. 
We will see that $d^{(m)}_k$ is in fact independent of $m$ and thus, the superscript will be omitted from now on. The trace on $\text{TL}_m$ then decomposes into the sum of matrix traces in each submodule, i.e.,
\begin{equation}
  \text{Tr}_{\text{TL}_m} (t) = \sum_{k=0,1,\cdots,m/2} d_k \, \text{tr} (\pi^{(m)}_k(t)), \quad \forall t\in \text{TL}_m.
\end{equation}

Finally, in order to compute the resolvent, we are left with computing the matrix
\begin{equation}
  M^{(m)}_k(q) = \sum_{h\in \text{NC}_m} q^{-\#(h)}\pi^{(m)}_k (D(h)).
\end{equation}
for all $k$.
From \Eqref{eq:TL_resolvent}, we see that the spectrum of the product $M^{(m)}_k(q_A)M^{(m)}_k(q_B)$ determines the full spectrum of the reflected density matrix. In more detail, the resolvent is given by
\begin{equation}
	R_m(\lambda)=\sum_{k,i} \frac{d_k}{\lambda-\frac{\chi_A \chi_B}{(\chi \chi_C)^m}\lambda_{M^{(m)}_k,i}},
\end{equation}
where $\lambda_{M^{(m)}_k,i}$ are the eigenvalues of $M^{(m)}_k(q_A)M^{(m)}_k(q_B)$. Thus, we obtain the spectrum:
\begin{equation}\label{eq:spectrum}
	D(\lambda)=\sum_{k,i} d_k\, \delta\left(\lambda-\frac{\chi_A \chi_B}{(\chi \chi_C)^m}\lambda_{M^{(m)}_k,i}\right)=\sum_{k,i} d_k\, \delta\left(\lambda-\lambda^{(m)}_{k,i}\right),
\end{equation}
where $\lambda^{(m)}_{k,i}\equiv\frac{\chi_A \chi_B}{(\chi \chi_C)^m}\lambda_{M^{(m)}_k,i}$.
The spectrum takes the simple form of a sum over poles where the coefficients $d_k$ determine the degeneracy, while the eigenvalues are given by $\chi_A\chi_B/(\chi\chi_C)^m$ times the eigenvalues of $M^{(m)}_k(q_A)M^{(m)}_k(q_B)$.

We emphasize that, so far, our result is valid for large $\chi_{A/B/C} \gg 1$ with the ratios $q_A=\chi_A/\chi_{C_1}, q_B=\chi_B/\chi_{C_2}$ held fixed. The internal bond dimension is kept finite and thus, $\chi_{A/B/C} \gg \chi$. To match the gravitational saddles, we will eventually take the limit $\chi\gg 1$, but the more general result \Eqref{eq:spectrum} is valid for \emph{all} values of $\chi$.



\section{Reflected Entropy in 2TN} 
\label{sec:2TN}

Having set up the formalism for computing the reflected spectrum via the TL machinery, we will now compute and analyze the spectrum in this section. First, we demonstrate the formalism by performing the analysis at finite $\chi$ in \secref{sub:finite}, focusing on $m=2,4$ as illustrative examples. We will then analyze the large $\chi$ limit in \secref{sub:large}. Working in this limit enables us to obtain the spectrum as an analytic function of $m$ and thus, continue to $m=1$. With the spectrum at hand, we compute the reflected entropy and discuss the interpretations in terms of superselection sectors in \secref{sub:sr}. We then analyze the leading effect of finite external bond dimensions in \secref{sub:corrections}. We will see that they shift the locations and give rise to a width for each of the peaks in the spectrum. Finally, we demonstrate consistency of our analysis with numerics in \secref{sub:num}.

\subsection{Finite $\chi$} 
\label{sub:finite}
The recipe outlined in \secref{sub:TL} allows one to compute the spectrum exactly for arbitrary even integer $m$.
In this subsection, we will work out the detailed spectrum for $m=2$ and $m=4$ following this recipe.
The purpose of this analysis is twofold:
First, it serves as a proof of principle for analyzing the spectrum of arbitrary even integer $m$.
Note that in practice, such an analysis is not always feasible since the complexity of the calculation increases exponentially as $m$ increases. Second, since we pose no restriction on $\chi$ here rather than $\chi_{A/B/C}\gg \chi$, we expect our result to hold even when $\chi$ is small. This particular small $\chi$ regime allows us to make non-trivial predictions for the \emph{R\'enyi} reflected spectrum at even integer $m$. Such predictions can be checked numerically up to high accuracy, as opposed to the $\chi\to \infty$ results we will obtain in \secref{sub:large} and \secref{sub:sr}, where our numerics are limited by finite $\chi$ effects. 

\begin{itemize}
    \item $m=2$ \\
    Let us begin with a trivial case.
    The diagrammatic basis of the $\text{TL}_2$ algebra is given by
\begin{equation}
  \begin{matrix}
    \includegraphics[scale=.3]{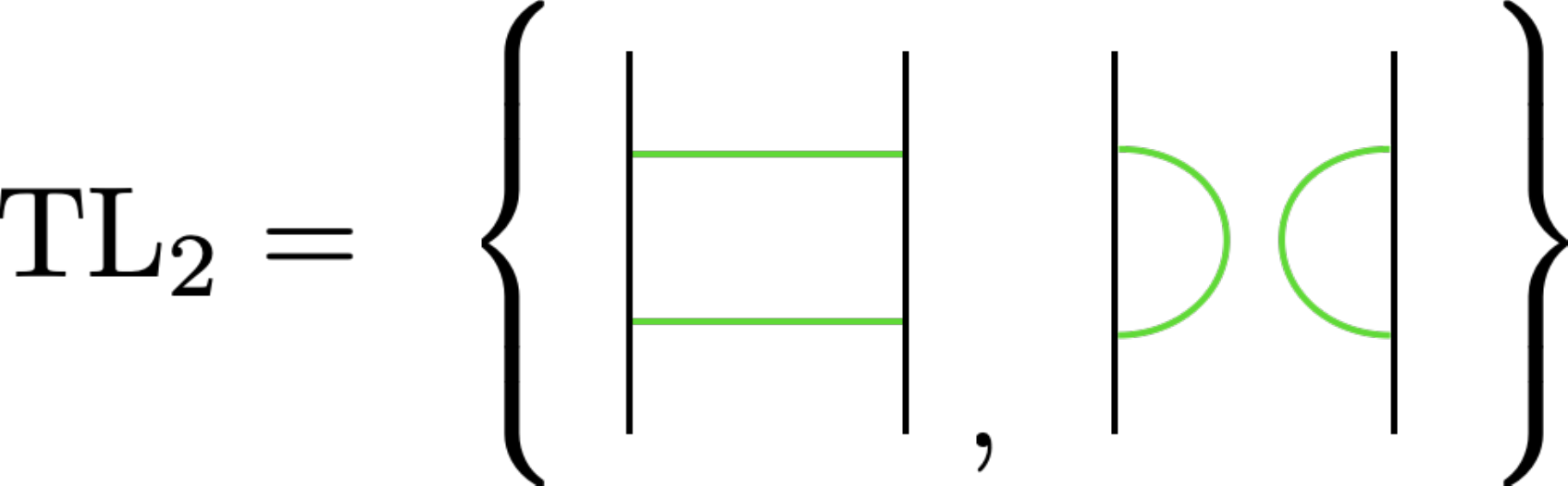}
  \end{matrix}
\end{equation}
which we will refer to as $D(\tau_2)$ and $D(e)$ following the notation of $\text{NC}_2$. There are two standard modules $\mathcal{V}^{(2)}_0$ and $\mathcal{V}^{(2)}_1$, each being one-dimensional:
\begin{equation}
  \begin{matrix}
    \includegraphics[scale=.3]{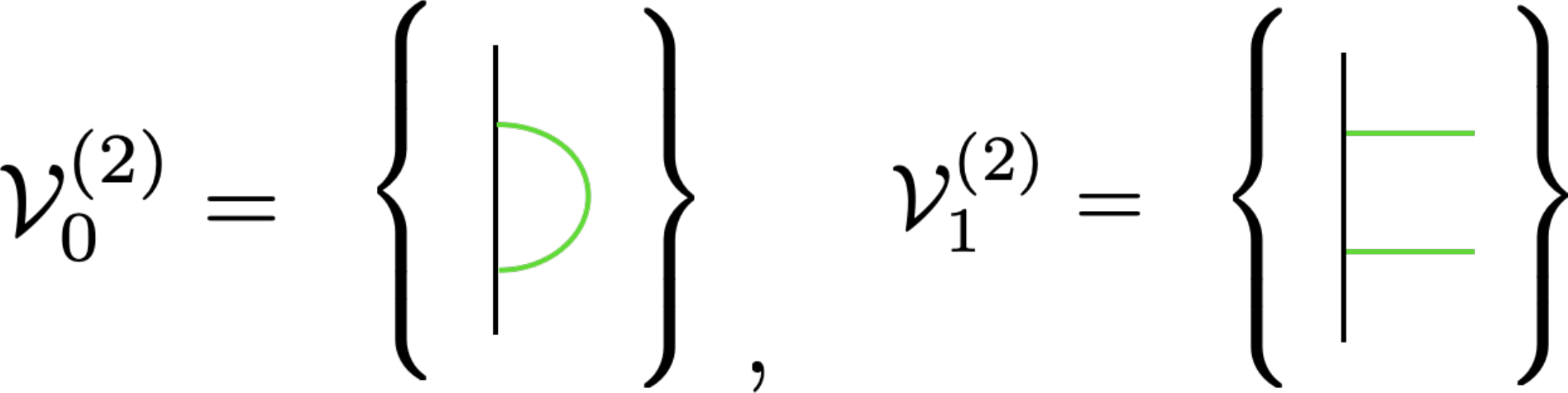},
  \end{matrix}
\end{equation}
where we remind the reader that the modules are labelled by half the number of open lines. Using the diagrammatic action described previously, we compute
\begin{align}
  \begin{split}
    \pi^{(2)}_0(D(\tau_2)) = 1,\quad   \pi^{(2)}_0(D(e)) = \chi; \\
  \pi^{(2)}_1(D(\tau_2)) = 1, \quad \pi^{(2)}_1(D(e)) = 0.
  \end{split}
\end{align}
To reproduce the TL trace, it is easy to check that we must pick $d_0=1$ and $d_1 = \chi^2-1$. Thus, the matrices $M^{(2)}_k(q)$ (which are simply 1-dimensional here) are:
\begin{align}
  M^{(2)}_0(q) = \chi q^{-2} + q^{-1}, \quad M^{(2)}_1(q)  = q^{-1}  
\end{align}
Using \Eqref{eq:spectrum}, we find the spectrum to be a sum over two poles, which have degeneracies and eigenvalues summarized in Table~\ref{tab:m=2}.
\begin{table}[H]
\centering
\begin{tabular}{|c|c|c|}
\hline
    sector & eigenvalue $\lambda^{(m)}_k$ & multiplicity $d_k$ \\
    \hline
     $k=0$ & $(q_A^{-1}+\chi^{-1})(q_B^{-1}+\chi^{-1})/\chi_C$ & 1 \\
     $k=1$ & $1/(\chi^2\chi_C)$ & $\chi^2-1$ \\
     \hline
\end{tabular}
\caption{The list of eigenvalues and their degeneracies for the $m=2$ reflected spectrum.}
\label{tab:m=2}
\end{table}

    \item $m=4$ \\
Moving on to a slightly more involved example, consider the $\text{TL}_4$ algebra which has $14$ basis diagrams:
\begin{equation}
  \begin{matrix}
    \includegraphics[scale=.3]{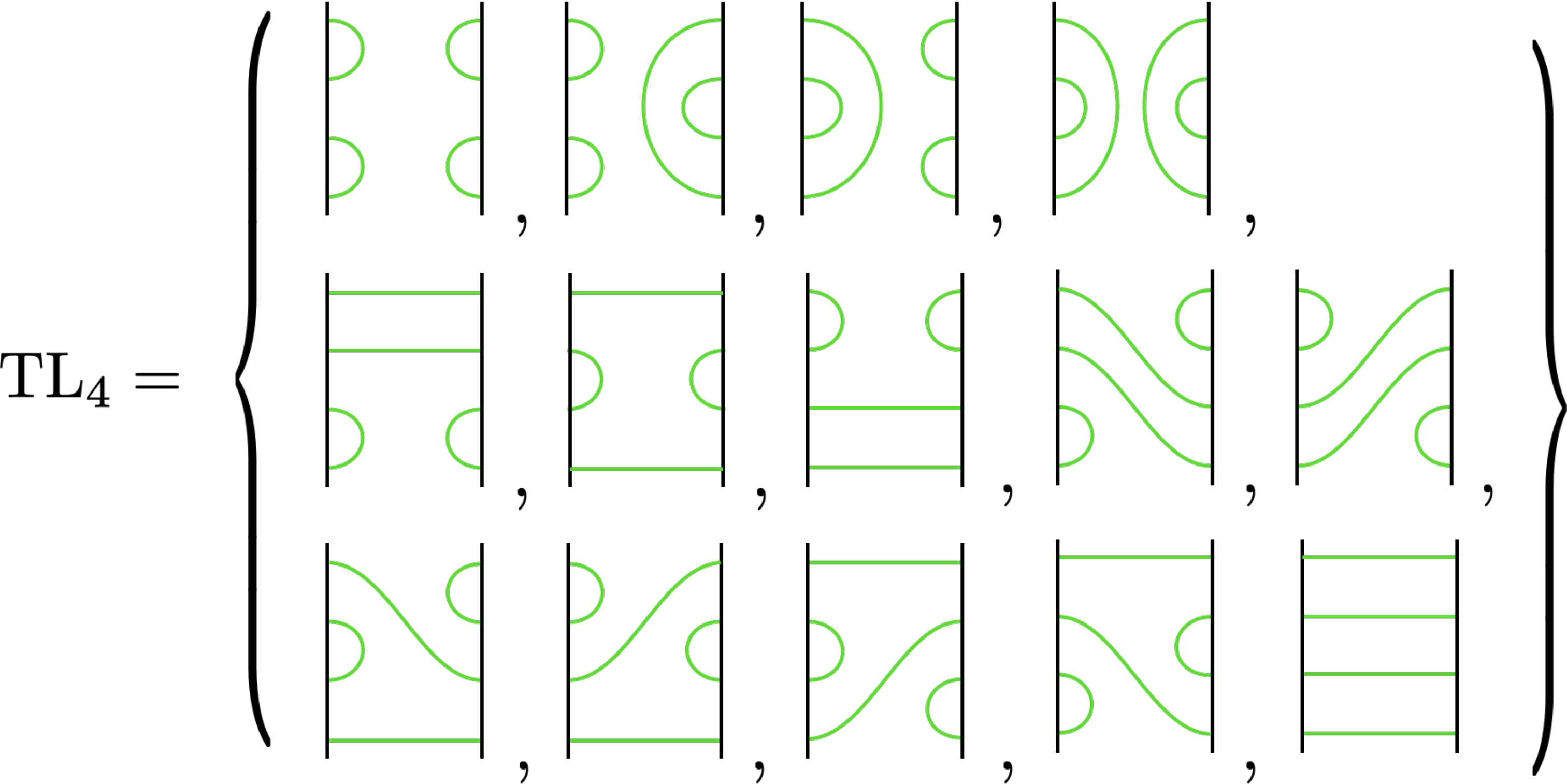}
  \end{matrix}
\end{equation}
There are three standard modules now: $\mathcal{V}^{(4)}_0$, $\mathcal{V}^{(4)}_1$, $\mathcal{V}^{(4)}_2$ with dimensions $2,3,1$ respectively:
\begin{equation}
  \begin{matrix}
    \includegraphics[scale=.3]{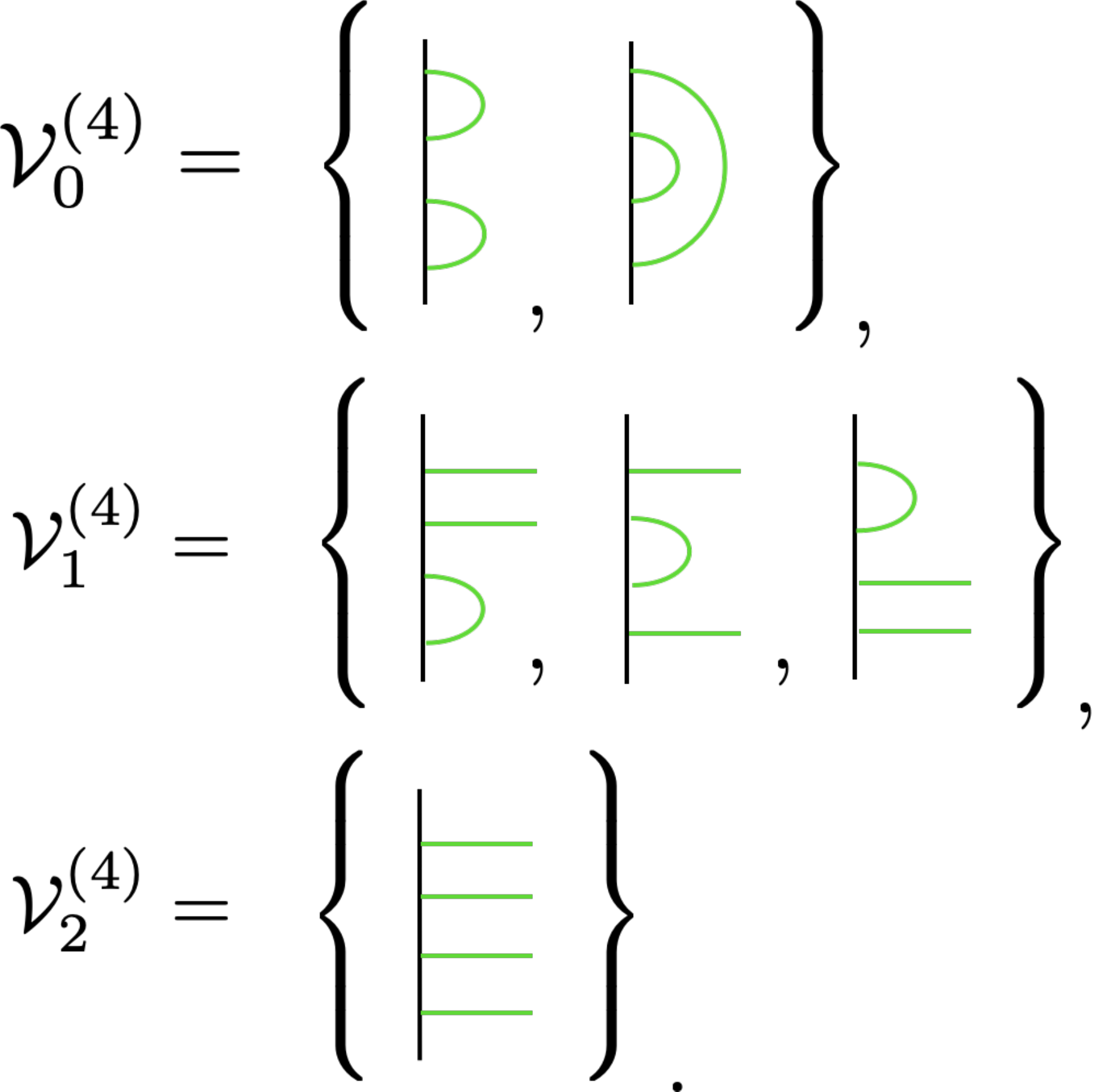}
  \end{matrix}
\end{equation}
Note that any two of the $(m,k)$-link states with same $k$ can be ``paired up'' to form a unique TL diagram with $2k$ crossing connections.
In the case at hand the total number of diagrams that can be formed this way are $2^2+3^2+1^2=14$, i.e. such pairings generate the entire set of $\text{TL}_4$ basis diagrams.
This pairing up action will be important in our large $\chi$ analysis in the next subsection.

Similar to the previous $m=2$ calculation, the module is determined by the diagrammatic action of TL diagrams on link states. The matrix $M^{(4)}_k$ is then determined by a weighted sum over the module representation of these diagrams. 
Since it is conceptually straightforward, albeit tedious, to carry out the analysis, we skip the details and present the results. They read:
\begin{align}
  \label{eq:m4_matrix}
  M^{(4)}_0(q) &=
                 \begin{pmatrix}
                   \chi^2q^{-4}+5\chi q^{-3} + 3q^{-2} & \chi^2q^{-3}+3\chi q^{-2}+q^{-1} \\
                   \chi^2q^{-3}+\chi(q^{-4}+2q^{-2})+4q^{-3} & \chi^2q^{-2}+\chi(q^{-3}+q^{-1})+3q^{-2}
                 \end{pmatrix}\\
  M^{(4)}_1(q) &=
                 \begin{pmatrix}
                   \chi q^{-3}+2q^{-2} & \chi q^{-2}+2q^{-3} & \chi q^{-3}+q^{-2} \\
                   \chi q^{-2}+q^{-1}& \chi q^{-1}+3q^{-2} & \chi q^{-2}+q^{-1}\\
                   \chi q^{-3}+q^{-2} & \chi q^{-2}+2q^{-3} & \chi q^{-3}+2q^{-2}
                 \end{pmatrix}\\
    M^{(4)}_2(q) &= q^{-2}.
\end{align}
Reproducing the TL trace requires us to pick $d_0=1, d_1 = \chi^2-1$ and $d_2=\chi^4-3\chi^2+1$. 
The eigenvalues and their corresponding degeneracies are given in Table~\ref{tab:m=4}.
\begin{table}[H]
\centering
\begin{tabular}{|c|c|c|}
\hline
    sector & eigenvalue $\lambda^{(m)}_k$ & multiplicity $d_k$ \\
    \hline
     $k=0$ & $\ell_+$ & 1 \\
     $k=0$ & $\ell_-$ & 1 \\
     $k=1$ & $(5q^{-1}\chi^{-1}+(2q^{-2}+1))^2/(\chi_C^3\chi^2)$ & $\chi^2-1$ \\
     $k=1,2$ & $1/(q^2 \chi_C^3\chi^4)$ & $\chi^4-\chi^2-1$ \\
     \hline
\end{tabular}
\caption{The list of eigenvalues and their degeneracies for the $m=4$ reflected spectrum. We have set $q_A=q_B=q$ to simplify expressions.}
\label{tab:m=4}
\end{table}
Note that that two eigenvalues of $M^{(4)}_2(q)$ turn out to coincide with the eigenvalue of $M^{(4)}_4(q)$. The $k=0$ eigenvalues $\ell_\pm$ is complicated, arising from the eigenvalues of the matrix in \Eqref{eq:m4_matrix}:
\begin{align}
\begin{split}
   \ell_\pm &= \frac{1}{4q^6\chi_C^3\chi^4}
    \Big((q^2+1)\chi^2+(q^3+6q)\chi+6q^2 \\ 
    &\pm \sqrt{(q^2+1)^2\chi^4+(2q^5+10q^3+12q)\chi^3+(q^6+20q^4+44q^2)\chi^2+(8q^5+52q^3)\chi+16q^4}  
    \Big)^2
\end{split}
\end{align}
At large $\chi$, the expressions simplify and we find
\begin{equation}
  \ell_+ \approx \frac{(q^2+1)^2}{\chi^3_Cq^6}, \quad \ell_- \approx \frac{1}{\chi^3_C\chi^2(q^2+1)^2}.
\end{equation}
We see that $\ell_+$ scales as $O(\chi^0)$, while $\ell_-$ is suppressed by $O(\chi^{-2})$ and comes close to the eigenvalue of the $k=1$ sector.
However, since the multiplicity of $\lambda^{(4)}_1$ scales as $O(\chi^2)$ at large $\chi$, $\lambda_-$ can be ignored (at leading order in $\chi$) when calculating the entropy.
The same can also be said for the sub-leading eigenvalue of the $k=1$ sector (they coincide with that of $k=2$). At large $\chi$ the $k=2$ sector multiplicity scales as $O(\chi^4)$ and it dominates over the subleading poles of the $k=1$ sector.
\end{itemize}

We depict the spectra we found for $m=2$ and $m=4$ in \figref{fig:spec_finite_chi}.
\begin{figure}
    \centering
    \includegraphics[scale=.4]{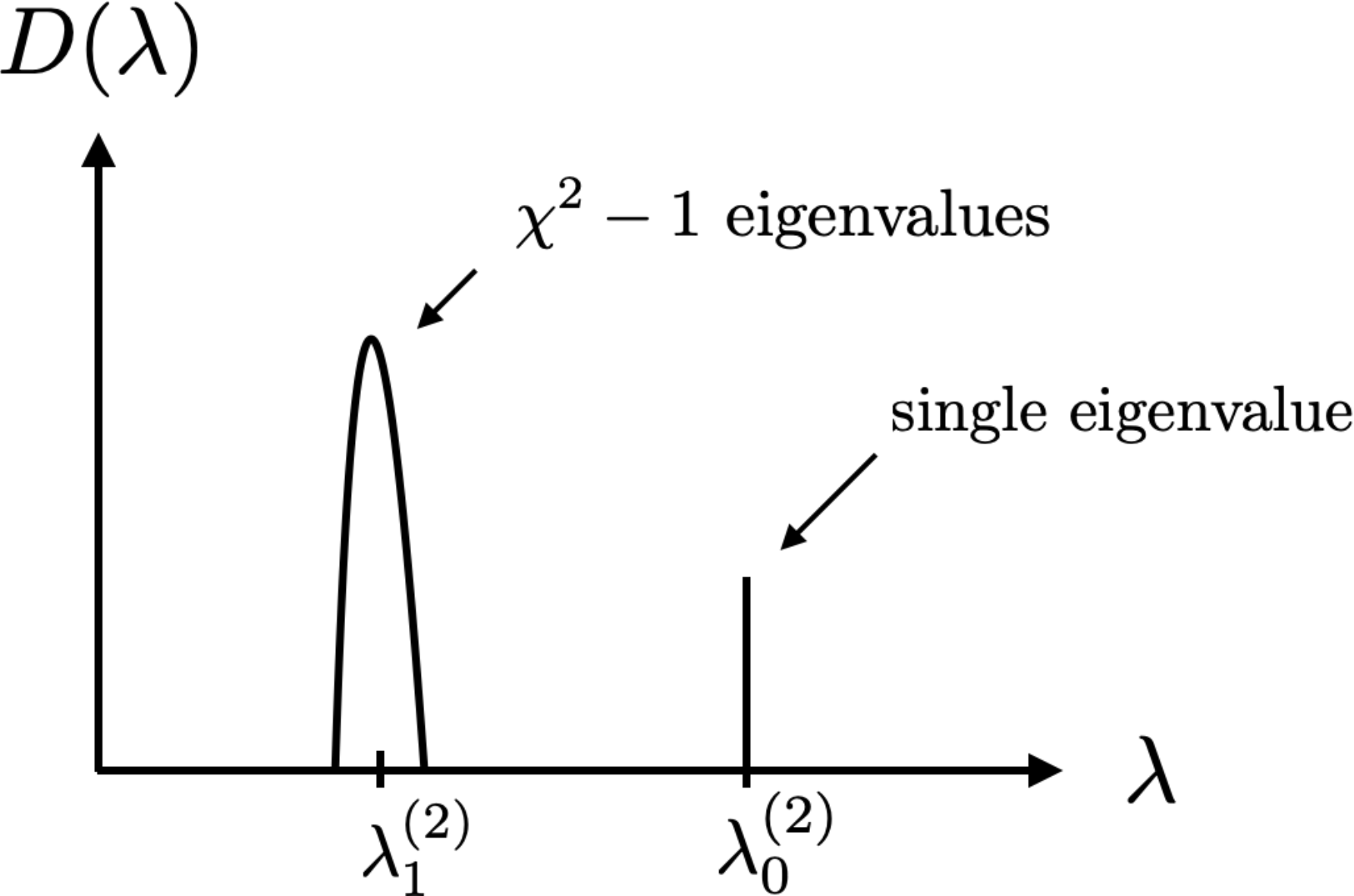}
    \hspace{.2in}
    \includegraphics[scale=.4]{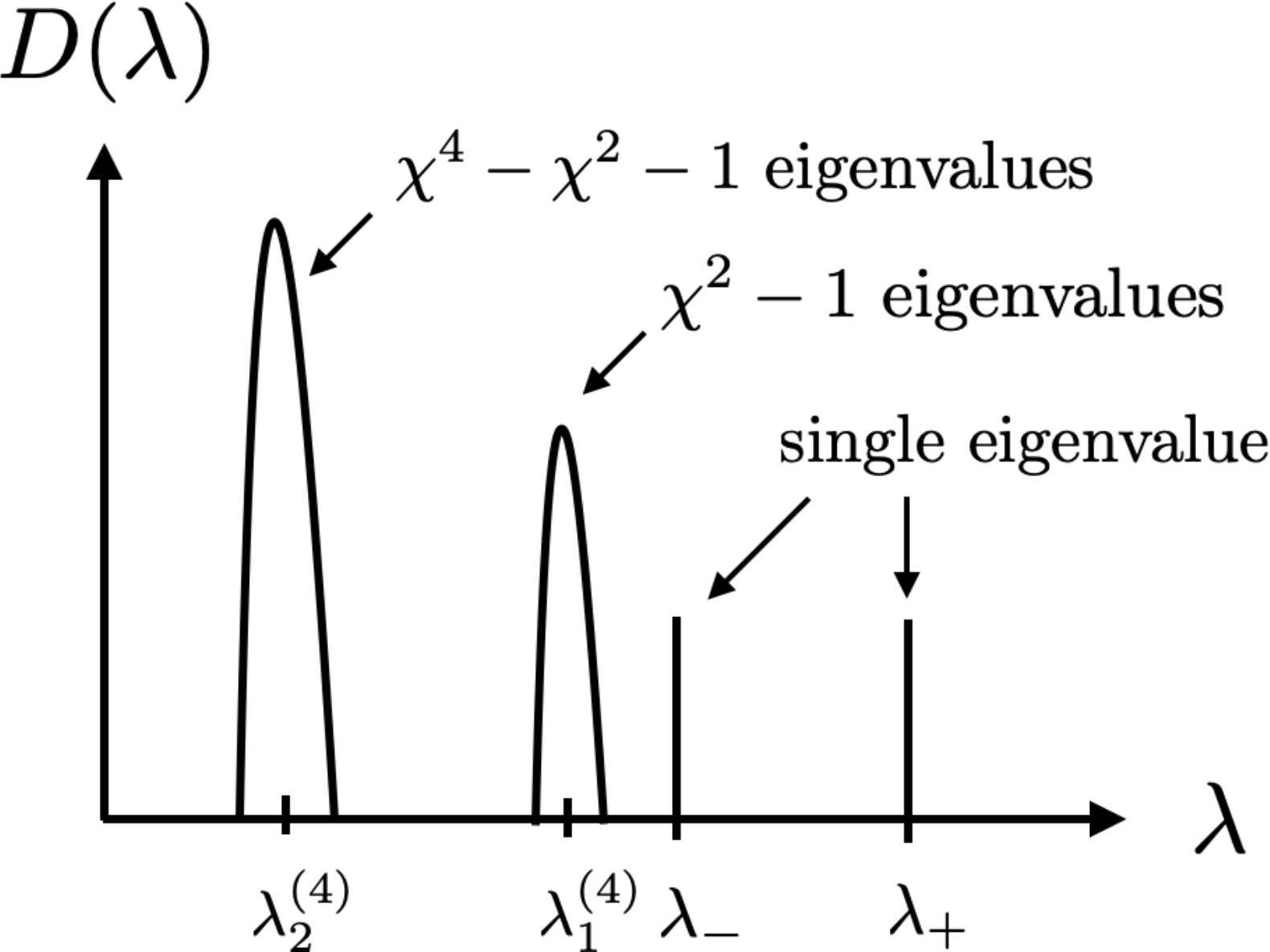}
    \caption{A sketch of the spectrum we obtained for $m=2$ (left) and $m=4$ (right). 
    The eigenvalues of $k>0$ sectors are depicted here as having a small width. This is found by including the finite external bond dimension effects, which we will study in \secref{sub:corrections}}
    \label{fig:spec_finite_chi}
\end{figure}
For higher $m$, similar calculations can be still be carried out, although the dimension of matrices $M^{(m)}_k(q)$ will be large and one has to revert back to numerical methods to find their eigenvalues.
However, the lesson we learnt from studying $m=2,4$ modules still holds:
In general there will be $m/2$ sectors, and the leading eigenvalue in each sector behaves as $O(\chi^{-2k})$ and with multiplicity $\sim \chi^{2k}$.
This hierarchy structure is a general feature for the 2TN spectra and we will see that it persists through analytical continuation $m\to 1$.

\subsection{Large $\chi$ limit} 
\label{sub:large}

We have already seen how one can use the standard module of the TL algebra to extract the reflected spectrum of the 2TN at even integer $m$ for arbitrary values of $\chi$. 
In this subsection we will study this problem in the limit $\chi\to\infty$.
This limit is physically relevant for comparison with holography since it corresponds to the $G\to 0$ limit in gravity. 
In \secref{ssub:even} we obtain the leading $\chi$ behavior of the spectrum in this limit.
Then we will analytically continue the spectrum in $m$ in this limit in \secref{ssub:analytic}.
This information will allow us to work at $m=1$ to understand the reflected entropy, which is the main focus of  \secref{sub:sr}.

\subsubsection{Even integer $m$} 
\label{ssub:even}

We begin with a preliminary statement on the coefficients $d^{(m)}_k$ introduced previously.
\begin{proposition}
  \label{lem:dk}
  The module multiplicity number $d^{(m)}_k$ satisfies the following properties:
  \begin{itemize}
  \item $d^{(m)}_k$  is independent of $m$.
  \item As a function of $\chi$, $d_k$ is determined by
    \begin{equation}
      d_k = [2k+1]_q,
    \end{equation}
    where the q-number $[\cdot]_q$ is defined as
    \begin{equation}
      [k]_q = \frac{q^k-q^{-k}}{q-q^{-1}} = q^{k-1} + q^{k-3} + \cdots + q^{-(k-3)} + q^{-(k-1)}.
    \end{equation}
    $\chi$ and $q$ are related by $\chi = [2]_q = q + q^{-1}$.
  \end{itemize}
\end{proposition}

One can solve for $q$ in terms of $\chi$ to obtain
\begin{equation}
\label{eq:dk}
  d_k = \frac{1}{4^{k}} \sum^{k+1}_{n=1}
  \begin{pmatrix}
    2k+1 \\ 2n-1
  \end{pmatrix}
  \chi^{2(k-n+1)}(\chi^2-4)^{n-1},
\end{equation}
so that $d_k$ is a polynomial in $\chi$. In fact one can show that the coefficients of this polynomial are integers. For instance, the first few values of $d_k$ are
\begin{align}
  \begin{split}
  d_0&=1,\\
  d_1&=\chi^2-1,\\
  d_2&=\chi^4-3\chi^2+1,\\
  d_3&=\chi^6-5\chi^4+6\chi^2-1    
  \end{split}
\end{align}
Although this result holds for all $\chi$, we will mostly just need the large $\chi$ behavior of $d_k$, which is
\begin{equation}
  d_k \approx \chi^{2k} + O(\chi^{2k-2})
\end{equation}
Proving this proposition requires some additional facts about the standard module. 
For this reason we present it in Appendix~\ref{sec:proof}.

Before moving on, let us introduce some useful notation. We define an inner product $\braket{\cdot,\cdot}$ on $\mathcal{V}^{(m)}_k$ as follows: If $x, y$ are two link states, $\Braket{x,y}$ is given by flipping $x$ across the vertical axis, matching to $y$ and assigning a power of $\chi$ for every closed loop. Furthermore, we define $\Braket{x,y}$ to be nonzero only if every defect in $x$ ends up being connected to a defect in $y$. A few examples should suffice to clarify the definition:
\begin{equation}
  \begin{matrix}
     \includegraphics[scale=.3]{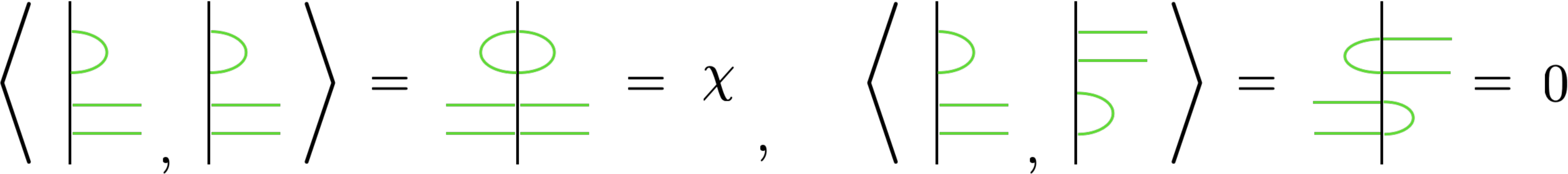}
  \end{matrix}
\end{equation}
Note that if $x,y$ are two link states in $\mathcal{V}^{(m)}_k$, then
\begin{equation}
  \label{eq:approx_ortho}
  \Braket{x,y} = \delta_{x,y} \chi^{m/2-k} + O(\chi^{m/2-k-1}),
\end{equation}
where $\delta_{x,y}=1$ if $x=y$ and $0$ otherwise. In other words, the set of link states form an approximately orthogonal basis for $\mathcal{V}^{(m)}_k$ in the $\chi\to\infty$ limit. This fact will turn out to be useful later in this section, as well as in order to get a gravitational interpretation in \secref{sec:gravity}.

We also define a bilinear map $|\cdot~\cdot|: \mathcal{V}^{(m)}_k\times \mathcal{V}^{(m)}_k\to \text{TL}_m$, by flipping $y$ across the vertical axis and ``pairing up'' with $x$ to form a TL diagram. For instance,
\begin{equation}
  \begin{matrix}
     \includegraphics[scale=.35]{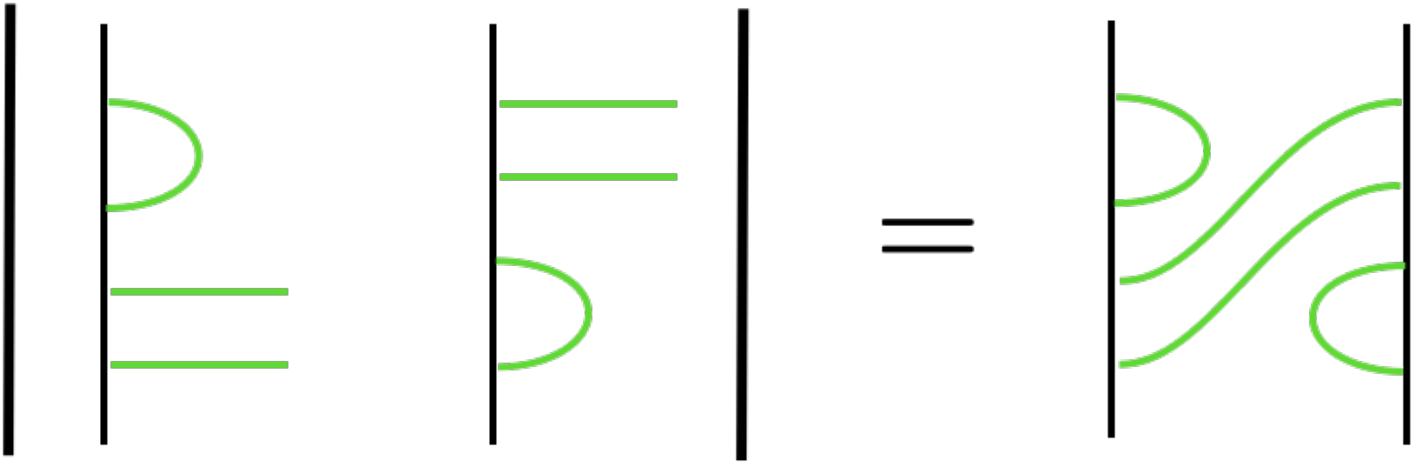}.
  \end{matrix}
\end{equation}
This map has a nice property that for all $x,y,z\in \mathcal{V}^{(m)}_k$ we have
  \begin{equation}
    \label{eq:bilinear}
    |x~y| \,z = \Braket{y,z}x.
  \end{equation}
    The proof of this equality is provided in Appendix~\ref{sec:TLalgebra}.

Returning to our main problem of finding the leading $\chi$ behavior of $M^{(m)}_k$, we ask the following question: For any given link state $v\in \mathcal{B}^{(m)}_k$, which set of diagrams acting on $v$ produce the dominant power of $\chi$? 
It turns out that, to produce the leading power of $\chi$, the right half of the TL diagram must be exactly the mirror image of $v$, since this is the only way to get the maximum number of closed loops.
Moreover, since each closed loop contributes one power of $\chi$, the overall contribution for such a diagram is $O(\chi^{m/2-k})$.
This fact is illustrated in \figref{fig:mirror}.

  \begin{figure}[h]   
    \centering
    \includegraphics[scale=.4]{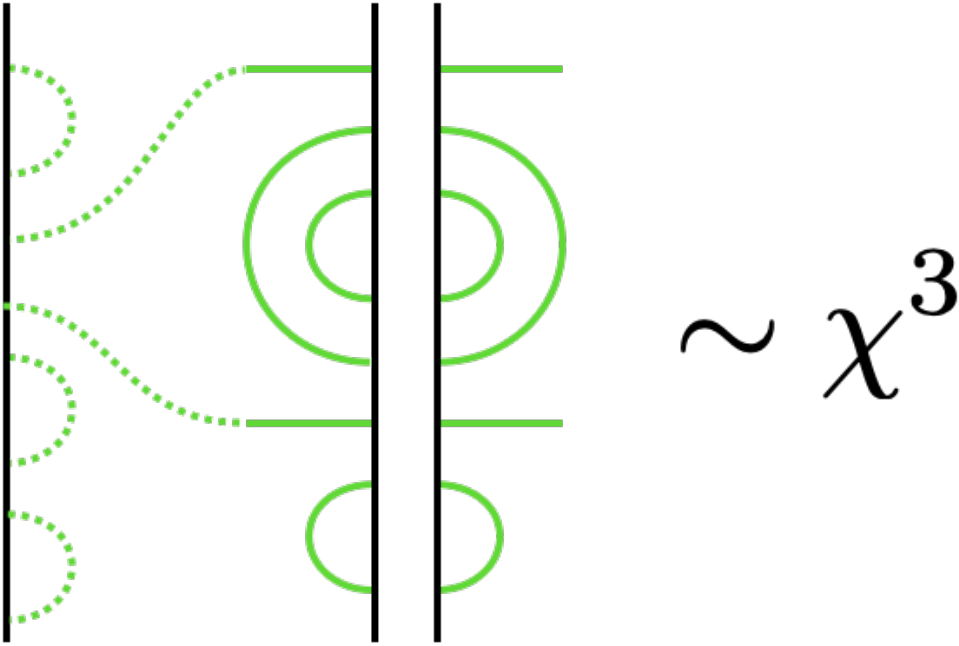}
    \caption{The leading contributions in $M^{(m)}_k$ come from the TL diagrams whose right half is a mirror image of the bases in $\mathcal{V}^{(m)}_k$. The contributing power is related to the number of defects by $\chi^{m/2-k}$. Here $m=8$ and $k=1$.}
    \label{fig:mirror}
  \end{figure}

Using our previous notation, the set of all possibly dominant diagrams for a given vector $y\in \mathcal{V}^{(m)}_k$ is then all elements of the form  $|x~y| \in \text{TL}_m$ where $x\in\mathcal{V}^{(m)}_{k}$. Note that every diagram in $\text{TL}_m$ will be dominant in exactly one module $\mathcal{V}^{(m)}_{k}$. It is easy to write down the sum of such diagrams by making use of the bilinearity of $|\cdot~\cdot|$. For example:
\begin{equation}
    \sum_{t\in \text{TL}_m} \pi_k^{(m)}(t) \approx \pi^{(m)}_k \left( \left|  \sum_{x\in \mathcal{B}^{(m)}_k}x  \quad  \sum_{y\in \mathcal{B}^{(m)}_k}y  \right| \right)+O(\chi^{m/2-k-1}),
\end{equation}
  
To actually compute the matrix $M^{(m)}_k(q)$, we still need to weigh the sum by powers of $q^{-\#(h)}$. Given $h\in \text{NC}_m$, $\#(h)$ can be computed by considering a two-sided concatenation of $D(h)$ with a special link state that we call $e_0 \in \mathcal{V}^{(m)}_0$, given by the relation $D(e) = |e_0 \quad e_0|$ where $e\in NC_m$ is the identity permutation. For example, if $h = (1456)(23) \in \text{NC}_6$, then we have
\begin{equation}
    \begin{matrix}
          \includegraphics[scale=.3]{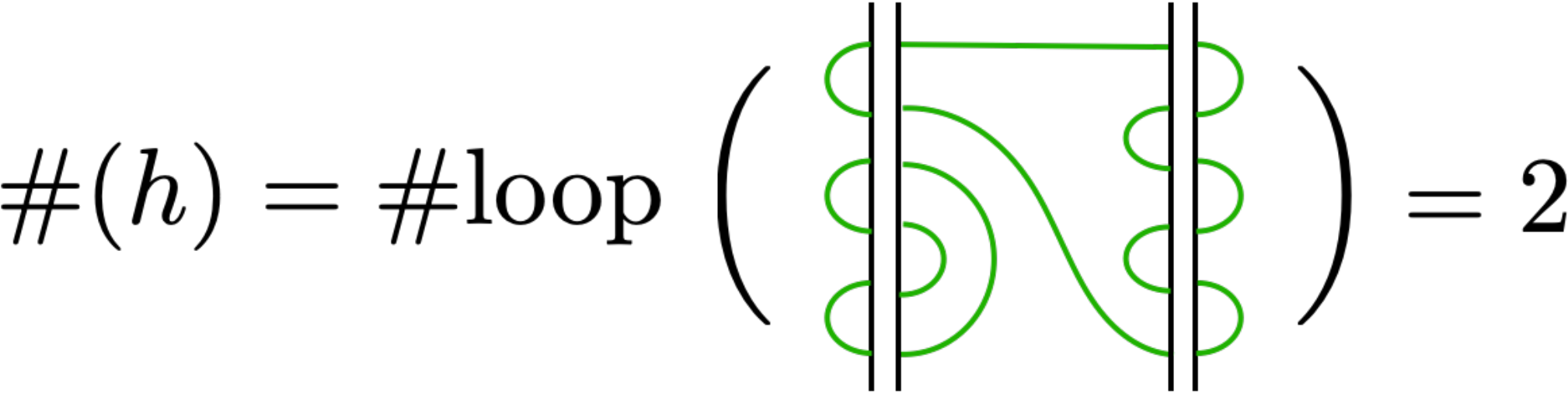}
    \end{matrix}   
\end{equation}
This facilitates the definition of the following linear functional $f_q:\mathcal{V}^{(m)}_k\to \mathbb{Z}[q^{-1}]$:
\begin{definition}
    If $v$ is a link state in $V^{(m)}_k$, the value of $f_q(v)$ is given by concatenating the mirrored reflection of $e_0$ with $v$ and assigning a factor of $q^{-1}$ for every closed loop. The action of $f_q$ on a general vector in $V^{(m)}_k$ follows from the linearity of $f_q$.
\end{definition}
We illustrate the definition with the following example:
\begin{equation}
    \label{eq:f_q}
    \begin{matrix}
       \includegraphics[scale=.3]{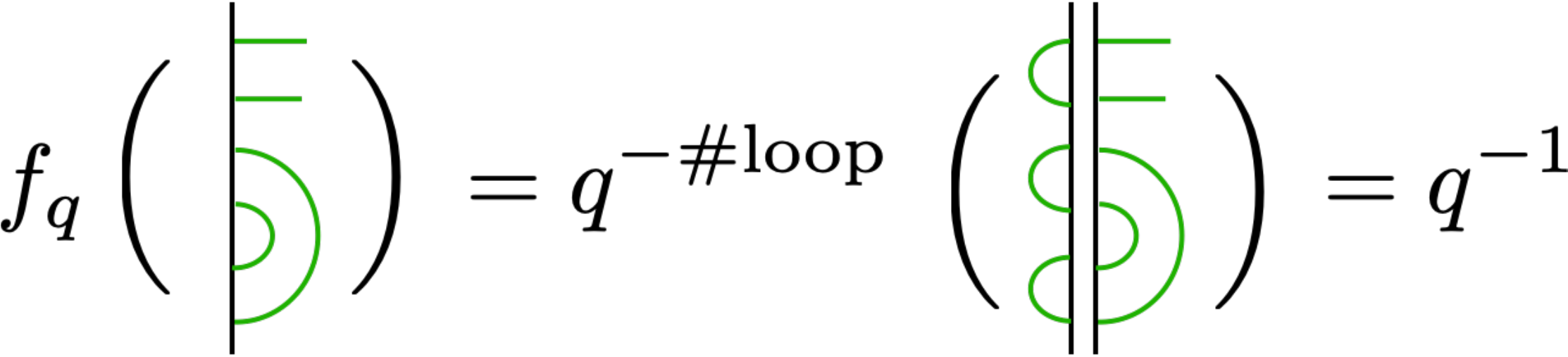}
    \end{matrix}
\end{equation}
where we emphasize that only closed loops are counted, whereas open strands are ignored. Using the linear functional $f_q$, we can then write
\begin{equation}
    q^{-\#(h)} = q^{-k} f_q(x)f_q(y)
\end{equation}
for all $h\in \text{NC}_m$ with the decomposition $D(h) = |x~y|$ and $x,y\in V^{(m)}_k$. It then follows that
\begin{equation}
    M^{(m)}_k(q)=\sum_{h\in \text{NC}_m} q^{-\#(h)}D(h) \approx q^{-k}\left| \sum_{x\in \mathcal{B}^{(m)}_k} f_q(x)x \quad \sum_{y\in \mathcal{B}^{(m)}_k}f_q(y)y \right| + O(\chi^{m/2-k-1})
\end{equation}
when considered as a linear operator acting on $\mathcal{V}^{(m)}_k$. 

For any $v\in V^{(m)}_k$ we can use \Eqref{eq:bilinear} to write
\begin{align}
    \begin{split}
       &\left(\sum_{h\in \text{NC}_m} q_A^{-\#(h)} D(h)\right)\left(\sum_{k\in \text{NC}_m} q_B^{-\#(k)} D(k)\right)v \\
    &\quad \approx (q_{AB})^{-k}\sum_{x,y,z,w\in \mathcal{B}^{(m)}_k}f_{q_A}(x)f_{q_A}(y)f_{q_B}(z)f_{q_B}(w)\Braket{y,z}\Braket{w,v}x,
    \end{split}
\end{align}
where we use the shorthand $q_{AB}\equiv q_Aq_B=\chi_A\chi_B/\chi_C$. It then follows that at large $\chi$, the operator product $M^{(m)}_k(q_A)M^{(m)}_k(q_B)$ is of rank one with eigenvector $\sum_{x\in \mathcal{B}^{(m)}_k} f_{q_A}(x)x$ and eigenvalue
  \begin{align}\label{eq:leadinglambda}
    \lambda_{M^{(m)}_k} \approx  (q_{AB})^{-k}\left( \sum_{x,y\in \mathcal{B}^{(m)}_k}f_{q_A}(x)f_{q_B}(y)\Braket{x,y} \right)^2
    \approx (q_{AB})^{-k}\chi^{m-2k}\left(\sum_{x\in \mathcal{B}^{(m)}_k}f_{q_A}(x)f_{q_B}(x)\right)^2,
  \end{align}
  where we have used \Eqref{eq:approx_ortho} to arrive at the second expression. All other eigenvalues in the same sector are suppressed by $O(\chi^{-2})$
  \footnote{Naively just from \Eqref{eq:leadinglambda} it seems reasonable to assume that the subleading eigenvalue is suppressed by $O(\chi^{-1})$, rather than $O(\chi^{-2})$. We will prove that the first order correction to the subleading eigenvalue vanishes in \secref{sub:corrections}.}.
  
The sum in \Eqref{eq:leadinglambda} is related to the generating function $G(q,r,z)$ of link states. This function can be computed in various ways. We present a diagrammatic derivation in Appendix~\ref{sec:proof} and use the result here without proof.
  \begin{proposition}
    \label{lem:gen_fun}
    The $f_q$-weighted generating function $G(q,r,z)$ for the $(m,k)$-link states is given by
    \begin{equation}
      \label{eq:gen_fun}
      G(q,r,z) = \frac{1-zC(q,z)}{1-z(q+r)-zC(q,z)}
    \end{equation}
    where $C(q,z) = \frac{1}{2z}(1-z(q-1)-\sqrt{(1+z(q-1))^2-4qz})$ is the generating function of $q$-Catalan numbers
    \footnote{For more information on $C(q,z)$ please see the proof of Lemma~\ref{lem:q_catalan} in Appendix~\ref{sec:TLalgebra}.}.
    The argument $r$ counts half the number of defects $k$ and the argument $z$ counts half the number of marked points $m/2$.
  \end{proposition}
$G(q,r,z)$ has the following series expansion:
  \begin{align}
    \begin{split}
       G(q,r,z) &= 1 + (q+r)z \\
                &\quad+ \left((q+q^2)+(1+2q)r+r^2\right)z^2   \\
      &\quad + \left( (q+3q^2+q^3)+(1+5q+3q^2)r+(2+3q)r^2+r^3 \right)z^3 + \cdots
    \end{split}
  \end{align}
where the formal expression $G(q,r,z) = \sum_{\mu,k}g_{\mu,k}(q)z^\mu r^k$ allows us to compute the eigenvalues relevant for the reflected spectrum. 

  \begin{corollary}
    The matrix $M^{(m)}_k(q_A)M^{(m)}_k(q_B)$ has a single nonzero eigenvalue at $\chi\to \infty$ given by
    \begin{align}
      \label{eq:lambda_m}
      \lambda_{M^{(m)}_k} = (q_{AB})^{-k}\chi^{m-2k}(g_{m/2,k}(q_{AB}^{-1}))^2
    \end{align}
    with multiplicity $\chi^k$.
  \end{corollary}
To conclude, for each even $m$ we identify $m/2$ eigenvalues given by \Eqref{eq:lambda_m}, labeled by $k\in\{0,\cdots,m/2\}$. At large $\chi$ the multiplicity of these eigenvalues goes as $d_k\approx \chi^{2k}$. All the other peaks found in \secref{sub:finite} vanish in the limit.


\subsubsection{Analytic continuation}
\label{ssub:analytic}
Here, we will perform the analytic continuation away from even integer $m$ to obtain the reflected spectrum at $m=1$. This will be accomplished by analytically continuing the generating function coefficient to $g_{1/2,k}$. We have seen that for even integer $m$, only a finite number of eigenvalues are present, i.e., only the sectors labelled by even integer $k\in\{0,\cdots,m/2\}$. However when taking the $m\to1$ limit, we will see that all integer values of $k$ contribute to the spectrum, forming an infinite tower of eigenvalues.

We start with \Eqref{eq:gen_fun}. In particular, we need an analytic form for the coefficients in the series
\footnote{We use the integer variable $\mu$ as the $z$ exponent of the generating function $G(q,r,z)$, as opposed to the even integer valued $m$. The analytic continuation of $m\to 1$ is equivalent to continuing $\mu\to 1/2$.}
\begin{equation}
  G(q,r,z) = \sum_{\mu,k=0}^{\infty} g_{\mu,k}(q)z^\mu r^k
\end{equation}
The $r$ expansion is easy to perform:
\begin{align}
  G(q,r,z) = \sum_{k=0}^\infty \frac{(1-zC(q,z))z^k}{(1-zq-zC(q,z))^{k+1}}r^k = \sum_{k=0}^\infty g_k(q,z)r^k
\end{align}
and thus, we identify
\begin{align}
  \label{eq:g_k_simple}
  \begin{split}
    g_k(q,z) &= \frac{(1-zC(q,z))z^k}{(1-zq-zC(q,z))^{k+1}}  \\
  &= C(q,z)(C(q,z)-1)^kq^{-k}
  \end{split}
\end{align}
where we have made use of \Eqref{eq:Cz_identity} to arrive at the final expression.
We then still need to expand this in terms of $z$ and get a closed form for the coefficients. We will use a contour integral trick to pick out the appropriate coefficient, i.e.,
\begin{equation}
  g_{\mu,k} = \frac{1}{2\pi i}\oint \frac{dz}{z^{\mu+1}} g_k(q,z) = \frac{1}{2\pi i}\oint \frac{dz}{z^{\mu+1}}C(q,z)(C(q,z)-1)^kq^{-k}
\end{equation}
Where the contour is chosen to enclose a neighborhood of $z=0$.
The function $C(q,z)$ has a branch cut running between
\begin{equation}
  z_\pm = \frac{1}{(1\mp\sqrt{q})^2},
\end{equation}
and there are no other poles or branch cuts on the complex plane
\footnote{One must be careful about choosing the square root branch cut in $C(q,z)$. In particular if we want the function to behave nicely at $z\to \infty$ we have to use $C(q,z)=(1-z(q-1)+|q-1|\sqrt{z-z_+}\sqrt{z-z_-})/(2z)$. This new $C(q,z)$ is equal to the original definition in Proposition~\ref{lem:gen_fun} near the $z=0$ neighborhood and hence they have the same Taylor expansion.}.
We can then deform the contour to enclose the branch cut running between $(z_-,z_+)$.

Now we would like to perform the analytic continuation for $\mu\to1/2$. This introduces an extra branch cut emanating from $z=0$, which does not affect our choice of contour.
Since $C(q,z)$ has the property
\begin{equation}
    \lim_{\text{Im } z\to 0+} \text{Im}(C(q,z)) = -\lim_{\text{Im } z\to 0-} \text{Im} (C(q,z)),
\end{equation}
one can rewrite the integral as twice of the imaginary part of the UHP contour :
\begin{equation}
  g_{1/2,k}(q) = \frac{1}{\pi} \, \text{Im} \int^{z_+}_{z_-} \frac{dz}{z^{3/2}} \, C(q,z)(C(q,z)-1)^kq^{-k}.
\end{equation}
For illustration, we evaluate the integral for the first few values of $k$:
\begin{itemize}
\item $k=0$\\
  We have
  \begin{equation}
    g_{1/2,0}(q) = \frac{1}{\pi} \, \text{Im} \int^{z_+}_{z_-} \frac{dz}{z^{3/2}} \, C(q,z)
  \end{equation}
  This integral results in the analytic continuation of the $q$-Catalan numbers at $\mu=1/2$,
  \begin{equation}
    g_{1/2,0}(q) = C_{1/2}(q).
  \end{equation}
  $C_\mu(q)$ is a piecewise function depending on whether $q>1$:
  \begin{equation}
    C_\mu(q) =
    \begin{cases}
      q \;_2F_1 (1-\mu,-\mu;2;q), \quad &q\le 1,\\
      q^\mu \;_2F_1(1-\mu,-\mu;2,q^{-1}), \quad&q>1
    \end{cases},
  \end{equation}
  which will be important seeds for expressing the coefficients for all the subsequent $k$. The eigenvalue in this sector is 
  \begin{equation}
    \lambda_0 =  q_{AB}C_{1/2}(q_{AB}^{-1})^2
  \end{equation}
  with multiplicity $1$.
  
 \item $k=1$ \\
   The $k=1$ integral reads
   \begin{equation}
     g_{1/2,1}(q)=\frac{1}{\pi} \, \text{Im} \int^{z_+}_{z_-} \frac{dz}{z^{3/2}} \,q^{-1}(C^2(q,z)-C(q,z))
   \end{equation}
   Using \Eqref{eq:Cz_identity} it can be shown that $C(q,z)$ satisfies the quadratic equation
   \begin{equation}
     \label{eq:Cz_quad}
     C^2(q,z)+\left(q-1-\frac{1}{z}\right)C(q,z)+\frac{1}{z}=0,
   \end{equation}
   so that one can write the integrand as a linear functional of $C(q,z)$:
   \begin{align}
     \begin{split}
        g_{1/2,1}(q)&=\frac{1}{\pi} \, \text{Im} \int^{z_+}_{z_-} dz (q^{-1}z^{-5/2}-z^{-3/2})C(q,z) \\
     &=q^{-1}C_{3/2}(q)-C_{1/2}(q)   
     \end{split}
   \end{align}
   The eigenvalue of this sector is given by
   \begin{equation}
     \lambda_1 = \chi^{-2}(q_{AB}C_{3/2}(q_{AB}^{-1})-C_{1/2}(q_{AB}^{-1}))^2
   \end{equation}
   with multiplicity $\chi^2$.
      
 \item $k=2$ \\
   We can use the same trick to reduce the integrand to be linear in $C(q,z)$:
   \begin{align}
     \begin{split}
     g_{1/2,2}(q) &=\frac{1}{\pi} \, \text{Im} \int^{z_+}_{z_-} \frac{dz}{z^{3/2}} \,q^{-2}C(q,z)(C(q,z)-1)^2 \\
                  &=   \frac{1}{\pi} \, \text{Im} \int^{z_+}_{z_-} dz \,(q^{-2}z^{-7/2}-z^{-5/2}(2q^{-1}+q^{-2})+z^{3/2})C(q,z) \\
                  &= q^{-2}C_{5/2}(q)-(2q^{-1}+q^{-2})C_{3/2}(q)+C_{1/2}(q)       
     \end{split}
   \end{align}
   The eigenvalue of this sector is
   \begin{equation}
     \lambda_2 = q^{-1}_{AB}\chi^{-4}\left(q_{AB}^2C_{5/2}(q_{AB}^{-1})-(2q_{AB}+q_{AB}^2)C_{3/2}(q_{AB}^{-1})+C_{1/2}(q_{AB}^{-1})\right)^2
   \end{equation}
   with multiplicity $\chi^4$.
\end{itemize}

Thus, we see a general pattern emerging. The integrand of $g_{1/2,k}$ is a polynomial of $C^k(q,z)$ and we can always reduce it to some linear functional $C(q,z)$ by repeated uses of \Eqref{eq:Cz_quad}. The outcome can in turn be expressed in terms of (half-integer valued) $q$-Catalan numbers. Therefore, we only need to find out how the reduction works for general $k$.
\begin{proposition}
  \label{lem:gen_fun_2}
  The generating function $G(q,r,z)$ in Proposition~\ref{lem:gen_fun} can be alternatively resummed as
  \begin{equation}
    \label{eq:gen_fun_2}
    G(q,r,z) = \frac{(r+q)C(q,z)-r/z}{r^2+(q+1-1/z)r+q}
  \end{equation}
\end{proposition}
We present the proof of this statement in Appendix~\ref{sec:proof}. This result allows us to write down an explicit form of $g_k(q,z)$ as a linear functional of $C(q,z)$:
\begin{align}
  \begin{split}
      G(q,r,z) = \, &C(q,z)  \bigg[ 1+(q^{-1}z^{-1}-1)r +\left(q^{-2}z^{-2}-(q^{-2}+2q^{-1})z^{-1}+1\right)r^2 \\
                                &\quad +\left(q^{-3}z^{-3}-(2q^{-3}+3q^{-2})z^{-2}+(q^{-3}+2q^{-2}+3q^{-1})z^{-1}-1\right)r^3+\cdots\bigg]\\
    &+ \text{other terms...}
  \end{split}
\end{align}
The detailed expression of the remaining terms not proportional to $C(q,z)$ is irrelevant for determining the contour integral. From this expansion, one can read off the coefficients $g_{1/2,k}$ as linear combinations of q-Catalan numbers easily: A negative power of $z^{-n}$ in the expansion coefficients becomes $C_{1/2+n}(q)$ after performing the contour integral.


\subsection{Reflected spectrum and the effective description}
\label{sub:sr}
In previous subsections, we have seen exact calculations of spectra for even integer $m$ and arbitrary $\chi$, as well as how the large $\chi$ limit allows us to extract analytic behavior of the eigenvalues in individual sectors. 
With these results in hand, we are finally ready to tackle the problem of analyzing the reflected entropy for the 2TN model.

In the limit $\chi\to \infty$, the leading eigenvalue in each $k$ sector is 
\begin{align}
\label{eq:pk}
  \lambda_k = \chi^{-2k}q_{AB}^{1-k}(g_{1/2,k}(q_{AB}))^2 \equiv \chi^{-2k}p_k
\end{align}
with multiplicity $d_k\approx\chi^{2k}$.
The numbers $p_k$ are independent of $\chi$ and satisfy $\sum_kp_k=1$ from the normalization condition. 
We plot the first few $p_k$'s as a function of $q_{AB}$ in \figref{fig:pk}.
The spectrum in this limit is a sum of the eigenvalue over all $k$ sectors: 
\begin{equation}
    D(\lambda) = \sum_{k=0}^\infty \chi^{2k}\delta\left(\lambda-\lambda_k\right)
\end{equation}
We present a sketch of the full reflected spectrum in \figref{fig:sr_spectrum}.

\begin{figure}[t]
  \centering
  \includegraphics[width=.7\textwidth]{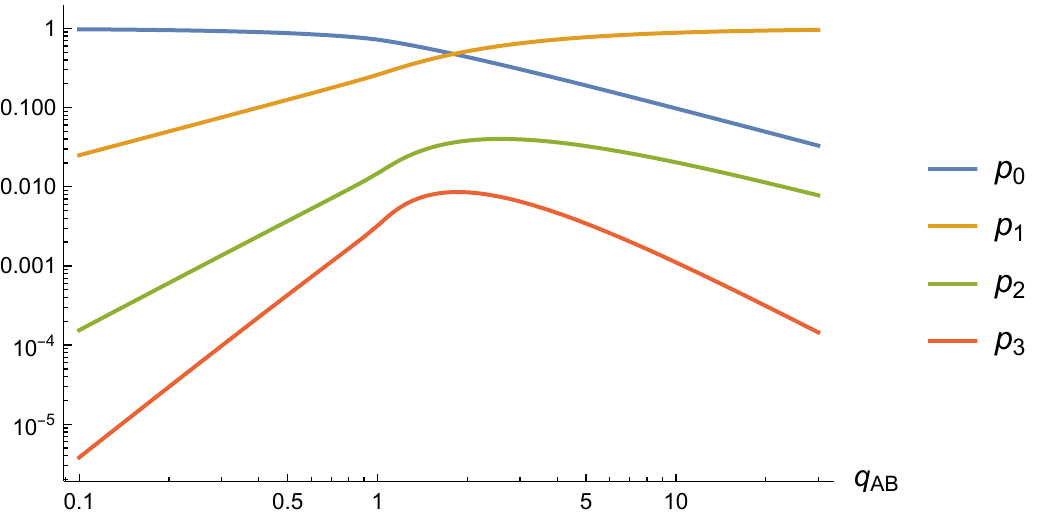}
  \caption{The plot showing the first few $p_k$'s as a function of $q_{AB}$.
    It is clear that at any point in parameter space, the dominant contribution comes from either $p_0$ (if $q_{AB}<1$) or $p_1$ (if $q_{AB}>1$).}
    \label{fig:pk}
\end{figure}
\begin{figure}[t]
  \centering
  \includegraphics[scale=.4]{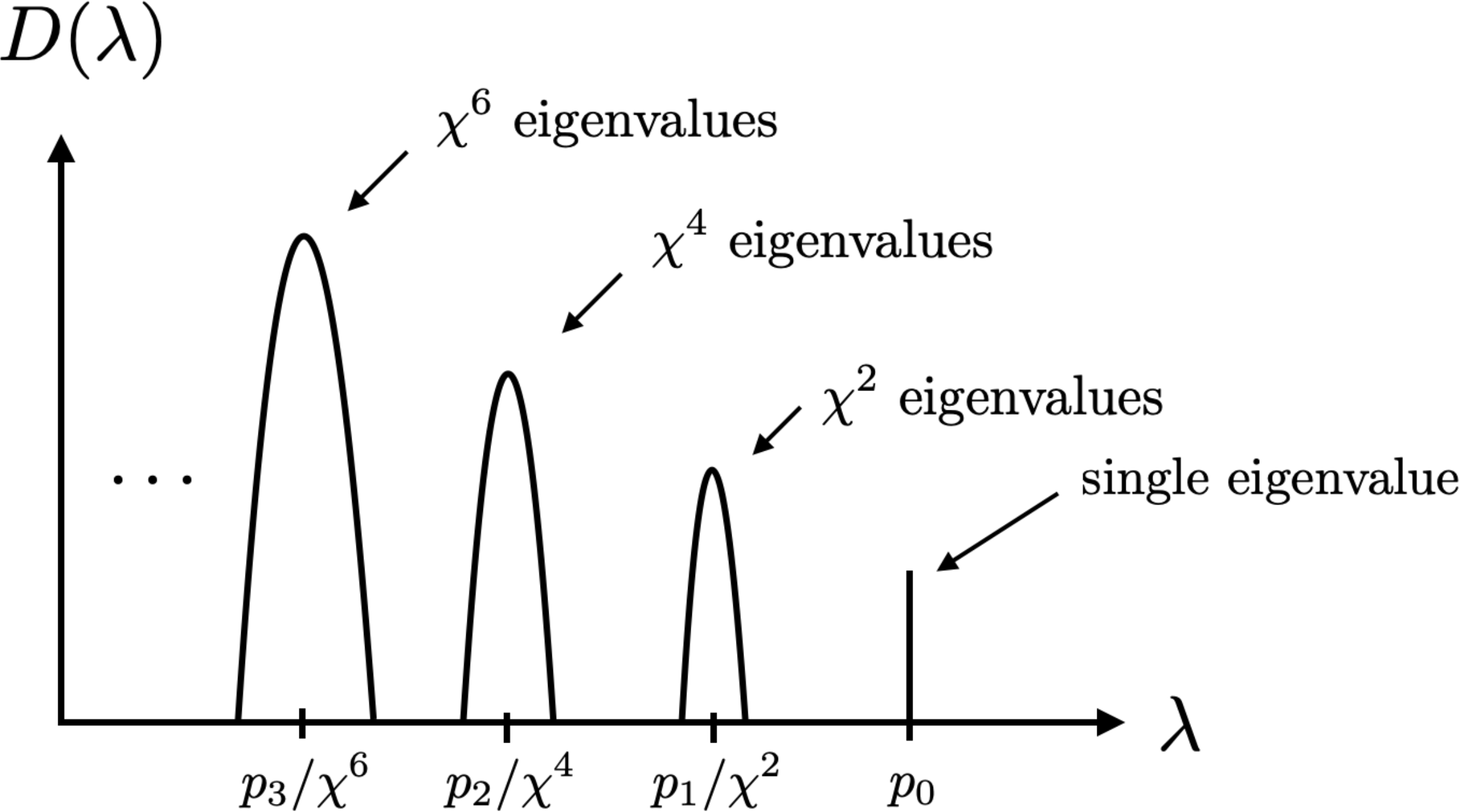}
  \caption{The sketch of the 2TN reflected spectrum at $m=1$. It features an infinite tower of eigenvalues labeled by the sector number $k$. Note that the eigenvalues of the $k>0$ sectors are drawn with a width here. This effect comes from taking the external bond dimensions $\chi_{A/B/C}$ large but finite, which we will investigate in detail in \secref{sub:corrections}.}
  \label{fig:sr_spectrum}
\end{figure}

The reflected entropy obtained from this spectrum is given by
\begin{align}
\label{eq:sr}
  \begin{split}
     S_R(q_{AB}) &= -\sum_{k=0}^{\infty} \chi^{2k}\lambda_k\ln \lambda_k \\
  &= -\sum_{k=0}^{\infty} p_k\ln p_k + \sum_{k=0}^{\infty} p_k\left(2k \ln \chi\right)   
  \end{split}
\end{align}
\Eqref{eq:sr} is the main result of this paper.
There are two different contributions to the reflected entropy.
There is a term from the classical Shannon entropy for the probabilities $p_k$, plus an infinite sum over terms proportional to $2k\ln\chi$, weighted by $p_k$. This should be compared with the result, \Eqref{eq:refgrav} obtained from the gravitational path integral in \secref{sec:gravity}. We can now interpret each $k$ sector as a state with a effective tensor network configuration with different EW cross-sections.

Based on the gravitational calculation (\secref{sec:gravity}), the results on single random tensor \cite{Akers:2021pvd}, and the analysis on finite external bond dimensions effects (\secref{sub:corrections}), we suggest that the effective TN states are built as follows:
Consider the natural doubling procedure \cite{Marolf:2019zoo} for the canonical purification, by duplicating the 2TN state and gluing the two copies across the boundary of the bulk extremal surfaces ($C_1$ and $C_2$ in our case). Call such a state $\ket{\psi_1}$.
Then, construct a series of wave functions $\ket{\psi_k}$ by further replicating $\ket{\psi_1}$ $k$ times and gluing across the $AA^*BB^*$ bonds.
By construction $\ket{\psi_k}$ will have an entropy of $2k\ln\chi$ and we have the follow effective description
\begin{equation}
\label{eq:eff_desc}
    \ket{\sqrt{\rho_{AB}}} = \sqrt{p_0}\ket{\psi_0} + \sum^\infty_{k=1}\sqrt{p_k}\ket{\psi_k},
\end{equation}
where $\ket{\psi_0}$ is a factorized state across $AA^*$ and $BB^*$.
The states $\ket{\psi_i}$ are approximately orthogonal at large $\chi$, i.e. $\braket{\psi_i|\psi_j}\sim \delta_{ij}+O(\chi^{-1})$.
Calculating the entropy of \Eqref{eq:eff_desc} gives precisely \Eqref{eq:sr} at $\chi\to \infty$.
We give a diagrammatic illustration of \Eqref{eq:eff_desc} in \figref{fig:interpret}.
\begin{figure}[t]
    \centering
    \includegraphics[width=\textwidth]{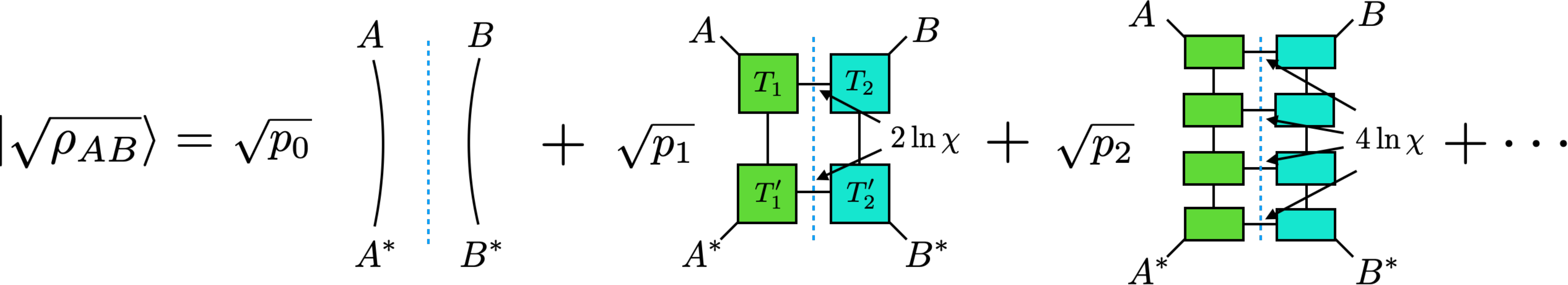}
    \caption{An illustration of \Eqref{eq:eff_desc}. The canonical purification $\ket{\rho_{AB}}$ is effectively described by a superposition of tensor network states. For $k=0$ we have the factorized state with zero cross-section. To form the $k>0$ sectors we start from a simple network state made from two copies of the 2TN state.
    Then for each $k$ we glue together $k$ such states, resulting in a TN with total cross-sectional area $2k\ln\chi$, as indicated by the number of bonds cut by blue dashed lines. This figure should be compared to \figref{fig:genus}, where each effective TN corresponds to a genus-$(2k-1)$ bulk solution of the replica boundary problem.}
    \label{fig:interpret}
\end{figure}

\Eqref{eq:eff_desc} should be compared to the states prepared by a gravitational path integral in the sense of \Eqref{eq:super}.
There we have a superposition of bulk solutions of EW cross-section $2k\ln \chi$, whose reduced density matrices $\rho_{k,AA*}$ have approximately orthogonal support.
One should think of the integer index $k$ as capturing the \emph{topology} of the effective description. By gluing together $k$ copies of the state $\ket{\psi_1}$ we have created a bulk solution with genus $2k-1$.
From the entanglement structure we also identify emergent superselection sectors labeled by the topological index $k\in \mathbb{Z}_{\ge0}$.
The area operator is:
\begin{equation}
    \mathcal{L}_{AA^*} = \sum_{k} 2k\ln\chi\, \Pi_k,
\end{equation}
where $\Pi_k$ is the projection operator down to the orthogonal subspace of $\ket{\psi_k}$.
Note that in real scenarios with finite bond dimensions, the aforementioned superselection sector is only approximate as supports of the density matrices are not exactly orthogonal.

As discussed in Ref.~\cite{Harlow:2016vwg} there is a connection between the area law and quantum error correcting codes. Taking this seriously here we see that the effective description looks like an emergent error correcting code, with only a central degree of freedom. Hence it is perhaps better thought of as a classical error correcting code. Presumably including bulk degrees of freedom in the original network, before canonically purifying, will give rise to a genuinely quantum version of this code.

Returning to the calculation of the reflected entropy, the phase transition of $S_R(A:B)$ is controlled by the list of classical probabilities $p_k$ as a function of $q_{AB}$.
They can be shown to have the following asymptotic behaviors
\begin{align}
\label{eq:pk_asym}
(p_0,p_1,p_{k>1}) \sim
\begin{cases}
    \left(q_{AB}^{-1}, ~
    1-\dfrac{5}{4}q_{AB}^{-1}, ~
    \dfrac{\Gamma(k-1/2)^2}{\pi\Gamma(k)^2}q^{1-k}_{AB}\right), &\quad q_{AB} \gg 1, \\
    \left( 1-\dfrac{1}{4}q_{AB},\dfrac{1}{4}q_{AB},\dfrac{\Gamma(k-1/2)^2}{\pi (2k)^2\Gamma(k)^2}q_{AB}^{k}\right),&\quad q_{AB} \ll 1.
\end{cases}
\end{align}
Away from phase transition, i.e., for $q_{AB}\ll1$ ($q_{AB}\gg1$), one can argue that either $k=0$ ($k=1$) is the dominant sector. 
The probabilities $p_k$ of higher $k$ sectors are always suppressed by factors of $q$ ($q^{-1}$), as shown in \Eqref{eq:pk_asym} and \figref{fig:pk}.
This matches our expectation, as in the limit $q_{AB}\ll 1$ the state $\ket{\psi_0}$ is dominant and we have $S_R(A:B)\sim O(1)$; whereas as $q_{AB}\gg 1$ the state $\ket{\psi_1}$ is dominant and we get the classical result $S_R(A:B) = 2\ln \chi = 2 EW(A:B)$.  
Near the vicinity of phase transition, i.e. $q_{AB}\approx 1$, all the $k$ sectors become important in determining the entropy, as one can check numerically that the probabilities $p_k$ for $k\geq 2$ are of comparable order of magnitude, although numerically smaller than $p_{0,1}$.

The $(1,n)$-R\'enyi reflected entropies are given by the sum
\begin{equation}\label{eq:renyiS}
    S^{(1,n)}_R = \frac{1}{1-n} \ln \left( \sum_k\chi^{2k(1-n)} p_k^n \right).
\end{equation}
In reality, only the first two terms $k=0,1$ can ever dominate the sum as long as $n>1$. Therefore we found that $S_R^{(1,n)}$ has three different approximate behavior depending on the value of $q_{AB}$:
\begin{equation}
S_R^{(1,n)} \approx 
    \begin{cases}
        \frac{1}{1-n}\ln p_0 \approx 0, \quad & q_{AB}\ll 1,\\
        \frac{1}{1-n}\ln p_0 \approx \frac{n}{n-1}\ln q_{AB}, \quad & 1\ll q_{AB}\ll \chi^{(2-2/n)},\\
        2\ln \chi, \quad & q_{AB}\gg \chi^{(2-2/n)},
    \end{cases}
\end{equation}
which matches exactly the expected results from the phase diagram in \figref{fig:2TNphase} in the parameter regime where our analysis is valid. As we take $n\to 1$, the middle regime vanishes and we get back the single sharp phase transition of reflected entropy. 

Interestingly, for $n<1$, the $k>2$ saddles can dominate. For sufficiently large $\chi$, the term $\chi^{2k(1-n)}$ in \Eqref{eq:renyiS} becomes increasingly important for larger $k$. For even integer $m$, this implies that the highest sector, i.e., $k=\frac{m}{2}$, dominates. For $m$ away from the even integers, this leads to a runoff to arbitrarily high $k$ which leads to the entropy being infinite in our approximation. In practice, such calculations would receive large corrections from the finiteness of the external bond dimensions $\chi_{A/B/C}$ since there is a constraint on the rank of the reflected spectrum arising from $\min(\chi_A^2,\chi_B^2)$.



\subsection{Corrections to the spectrum}
\label{sub:corrections}
In this subsection, we study the effect of having finite bond dimensions $\chi$ and $\chi_{A,B,C_1,C_2}$.
This is motivated by comparing to situations such as the four boundary wormholes with large but finite horizon areas and EW cross-section.
Moreover, this allows us to make better comparison with numerics we obtained in \secref{sub:num}.

Taking the internal/external bond dimensions finite alter the large $\chi$ spectrum in independent ways.
In short, the leading effect of finite \textit{internal} bond dimension $\chi$ is to shift the location of the poles in each sector.
In contrast, the leading effect of finite \textit{external} bond dimension $\chi_{A,B,C_1,C_2}$ is to spread out each pole into a narrow mound.
In the following we will examine how these effects work together to create a consistent spectrum that matches our numerical data well and deepens our understanding of the effective description of the 2TN as a sum over superselection sectors. 
However, the calculation of the finite internal bond dimension $\chi$ corrections is rather technical and involved. For this purpose we present the full analysis in Appendix~\ref{sec:finite_chi} for interested readers.

\subsubsection{Eigenvalue shifts}
\label{sub:shifts}
Here we give a qualitative summary of the effect of finite internal bond dimension $\chi$.
We identify two phenomena as we take $\chi$ finite:
First, the orthogonality of link state basis $v\in \mathcal{B}^{(m)}_k$, i.e. \Eqref{eq:approx_ortho}, fails.
Second, elements not of the form $\{|x~y|; x,y\in\mathcal{B}^{(m)}_k\}$  in the sum $\sum \pi^{(m)}_k(D(h))$ will start to contribute.
These two effects introduce $O(\chi^{-1})$ corrections to the leading eigenvalue in each sector $\lambda_k$.
The latter effect also splits the degenerate zero eigenvalues in each sector.
However, such corrections only appear at $O(\lambda_k\chi^{-2})$, so they do not affect the entropy at leading order.
Note that $O(\lambda_k\chi^{-2})$ is also the order of the leading eigenvalue $\lambda_{k+1}$ in sector $k+1$.
This is seen in our numerics from the fact that the number of eigenvalues in each mound is not exactly $d_k$, but has corrections suppressed by $O(\chi^{-2})$.
We will refer to this as \textit{sector mixing}, which should be understood as a signature that the wave functions in different superselection sectors acquire a non-zero overlap when $\chi$ is finite.
Both effects are seen in our explicit examples of integer $m$ (sector mixing is only visible in $m\ge4$) in \secref{sub:finite}.

We conclude the brief discussion by giving expressions for the leading order correction to the eigenvalues $\lambda_k$ for $k=0$ and $k=1$.
We chose to do so since the shifts of these two eigenvalues completely characterizes the leading order correction to the reflected entropy for all $q_{AB}$.
Corrections to other sectors are also obtainable via analytic continuation of the related generating function, but their contribution to the entropy is sub-leading.
Please refer to Appendix~\ref{sec:finite_chi} for a detailed treatment.

\begin{align}
\label{eq:lambda0_shift}
        \Delta \lambda_0 = \frac{q_A^{-1}+q_B^{-1}}{\chi}
    &\Bigg[ C_{1/2}(q_{AB}^{-1})\left((1+q_{AB}^{-2})D_{3/2}(q_{AB}^{-1})-2(1+q_{AB}^{-1})D_{5/2}(q_{AB}^{-1})+D_{7/2}(q_{AB}^{-1})\right)\nonumber \\
    & + \left(q_{AB}C_{3/2}(q_{AB}^{-1}) - C_{1/2}(q_{AB}^{-1})\right)^2
    \Bigg],
\\
\label{eq:lambda1_shift}
    \Delta \lambda_1 =
    \frac{q_{A}^{-1}+q^{-1}_B}{\chi^3} &\Bigg[
    3\left(C_{1/2}(q_{AB}^{-1})-(2q_{AB}+q^2_{AB})C_{3/2}(q_{AB}^{-1})+q^{2}_{AB}C_{5/2}(q^{-1}_{AB})\right)^2 \nonumber \\
    & + (-C_{1/2}(q_{AB}^{-1})+q_{AB}C_{3/2}(q_{AB}^{-1}))\nonumber \\
    &\times \big(( -(q^{-1}_{AB}-1)^2D_{1/2}(q^{-1}_{AB})+2q^{-1}_{AB}D_{3/2}(q^{-1}_{AB}) +q_{AB}^{2}D_{5/2}(q_{AB}^{-1}) \nonumber    \\
      &\quad-2(q_{AB}^{2}+q_{AB})D_{7/2}(q^{-1}_{AB})+q_{AB}^{2}D_{9/2}(q_{AB}^{-1}) \big)    \Bigg].
\end{align}
where the function $D_\mu(q)$ is given by
\begin{align}
    D_\mu(q) = 
    \begin{cases}
        \;_2F_1(1-\mu,1-\mu;1;q),\quad & q\le 1, \\
        q^{\mu-1}\;_2F_1(1-\mu,1-\mu;1;q^{-1}), \quad & q>1.
    \end{cases}
\end{align}

\subsubsection{Fluctuations in each sector}
\label{sub:fluc}
So far, the spectrum we have obtained consists of a bunch of poles, which in turn happened because we were working in the limit of large external bond dimensions. The leading effect of taking these dimensions finite will be to spread out each pole into a narrow mound, and the goal in this part will be to find a crude estimate for the width.

Based on the results obtained in the single random tensor \cite{Akers:2021pvd}, we conjecture that such spreading effects arise from including non-trivial permutations that act on the $n$-cycles. This motivates us to consider summing over a more general class of elements:
\begin{equation}
  g_1 = \gamma_p\left( \prod_{i=1}^n h_i \right)\gamma_{p}^{-1} ,  \quad g_2 = \gamma_q\left(\prod_{j=1}^n k_j\right)\gamma_{q}^{-1}
\end{equation}
where as before $h_i,k_j\in \text{NC}_m$ and additionally we pick $p,q\in \text{NC}_n$. $\gamma_p$ stands for the $n$-twist operator associated to $p$ applied to the lower half of the elements. This should be contrasted with \Eqref{eq:par} where we restricted to $\gamma_p=\gamma_\tau$ and $\gamma_q=e$. Note that these elements are only schematic since there is a possible overcounting when $h_i,k_j\in \text{NC}_{m/2}\times \text{NC}_{m/2}$, which will be accounted for later.

We remind the reader that in order to calculate the partition function $Z_{mn}$ in \Eqref{eq:partn}, we need to evaluate $\#(g_1g^{-1}_A)$ and $\#(g_1g^{-1}_B)$. These numbers depend on whether $h_i,k_j\in \text{NC}_{m/2}\times \text{NC}_{m/2}$ (in which case we call these permutations \textit{disconnected}) or not (in which case they are called \textit{connected}). For the connected sector, we furthermore classify $h_i$ and $k_j$ based on the number of crossing connections when viewed as an element of $\text{TL}_m$: We say $h\in \text{NC}_{m,k}$ if $D(h) = |x~y|$ for some $x,y\in  \mathcal{B}^{(m)}_k$.

For now we restrict the sum to the set of $h_i,k_j$ where all of the permutations are in the same subclass $\text{NC}_{m,k}$ but we allow $k$ to vary.
Using the formalism of annular non-crossing permutations (see Appendix~A of Ref.~\cite{Akers:2021pvd} for details) we get
\begin{align}
    \#(g_1g^{-1}_A) &=
    \left\{
    \begin{array}{l}
    \sum_i \#(h_i\tau^{-1}_m), \\
    \sum_i \#(h_i\tau^{-1}_m) - 2(n-\#(p\tau^{-1}_n)),
    \end{array}
  \right. 
  &
  \begin{array}{l}
  h_i\in \text{NC}_{m,0} \\
  h_i \in \text{NC}_{m,k>0}
  \end{array}
  \\
  \#(g_2g^{-1}_B) &=
  \left\{
    \begin{array}{ll}
    \sum_j \#(k_j\tau^{-1}_m), \\
    \sum_j \#(k_j\tau^{-1}_m) - 2(n-\#(q)),
    \end{array}
  \right.
  &
  \begin{array}{l}
    k_j\in \text{NC}_{m,0} \\
    k_j \in \text{NC}_{m,k>0}
  \end{array}
\end{align}
 Note that it is possible to give a formula for an arbitrary mixture of two different values of $k$ but those effects turn out to be subleading in determining the width, so we ignore them for now.
 
The partition function factorizes into different subclasses based on the number of crossings:
\begin{align}
  Z_{mn} \simeq Z^{(0)}_{mn} + Z^{(1)}_{mn} + Z^{(2)}_{mn} + \cdots
\end{align}
where the disconnected sum is
\begin{equation}
  Z^{(0)}_{mn} = \left(\frac{\chi_A\chi_B}{\chi_C^{m}\chi^{m}}\right)^n \sum_{h_i,k_j\in \text{NC}_{m,0}} q_A^{-\sum_i\#(h_i)}q_B^{-\sum_j\#(k_j)}
  \text{Tr}_{\text{TL}_m}\left(D(h_1)D(k_1)\cdots D(h_n)D(k_n)\right)
\end{equation}
The summation over $p,q$ drops out since they all overcount the same set of permutations. Since its form receives no correction from the new $p,q$ twist operators, we conclude that the single eigenvalue $\lambda_0$ remains unchanged in this approximation. We expect that by including the correction from mixed $k$ contributions, $\lambda_0$ may be shifted by a small amount or spread out to a very narrow peak.

Moving on to the connected sectors, there is no overcounting and we get extra corrections from the cycles in $p$ and $q$:
\begin{align}
  \label{eq:Z1_sum}
  Z^{(k)}_{mn} = \left(\frac{1}{\chi_A\chi_B\chi_C^{m}\chi^{m}}\right)^n \sum_{p,q\in \text{NC}_n} \chi_A^{2\#(p\tau^{-1}_n)} \chi_B^{2\#(q)} \sum_{h_i,k_j\in \text{NC}_{m,k}} q_A^{-\sum_i\#(h_i)}q_B^{-\sum_j\#(k_j)} \nonumber \\
  \times  \text{Tr}^{pq^{-1}}_{\text{TL}_m}\left(D(h_1)D(k_1)\cdots D(h_n)D(k_n)\right),
\end{align}
where the trace pattern is determined by the partition pattern in $pq^{-1}$. For example,
\begin{align}
\begin{split}
    &\text{Tr}^{(12)(3)}_{\text{TL}_m}(D(h_1)D(k_1)D(h_2)D(k_2)D(h_3)D(k_3)) \\
  &\quad = \text{Tr}_{\text{TL}_m}(D(h_1)D(k_1)D(h_2)D(k_2)) \times \text{Tr}_{\text{TL}_m}(D(h_3)D(k_3))
\end{split}
\end{align}
The dominant contribution in \Eqref{eq:Z1_sum} comes from $p=\tau_n, ~ q=e$, which gives the single eigenvalue identified in \secref{sub:finite} and \secref{sub:large}. The summation over $p$ and $q$ introduces nontrivial correlations in $Z_n$, which we now focus on analyzing. 

Denoting $pq^{-1} = \prod_i c_i$ to be the individual cycle decomposition of $pq^{-1}$, we can write
\begin{align}
  \begin{split}
  &\sum_{h_i,k_j\in NC_{m,k}} q_A^{-\sum_i\#(h_i)}q_B^{-\sum_j\#(k_j)}    \text{TL}^{pq^{-1}}_{\text{TL}_m}\left(D(h_1)D(k_1)\cdots D(h_n)D(k_n)\right) \\
    =& \prod_{\{c_i\}} \text{Tr}_{\text{TL}_m} \left[  \left(\sum_{h\in \text{NC}_{m,k} }q_A^{-\#(h)}D(h)\right) \left(\sum_{k\in \text{NC}_{m,k} }q_A^{-\#(k)}D(k) \right) \right]^{|c_i|}, 
  \end{split}
\end{align}
where $|c_i|$ is the number of elements in a given cycle. Then, using the fact that the leading order result $Z_{n}\approx \chi^{2k}\lambda_k^n$ in large $\chi$ limit for such sector, we have
\begin{align}
\label{eq:width_Z}
  Z^{(k)}_n  \approx \left(\frac{\lambda_k}{\chi^2_A\chi^2_B}\right)^n\sum_{p,q\in \text{NC}_n}\chi_A^{2\#(p\tau_n^{-1})}\chi_B^{2\#(q)} \chi^{2k\#(pq^{-1})},
\end{align}
Interestingly, the partition function $Z^{(k)}_{n}$ is identical to the partition function of an equivalent tensor network (up to an overall normalization), see \figref{fig:2TN_RT}.
\begin{figure}[t]
  \centering
  \includegraphics[width=.5\textwidth]{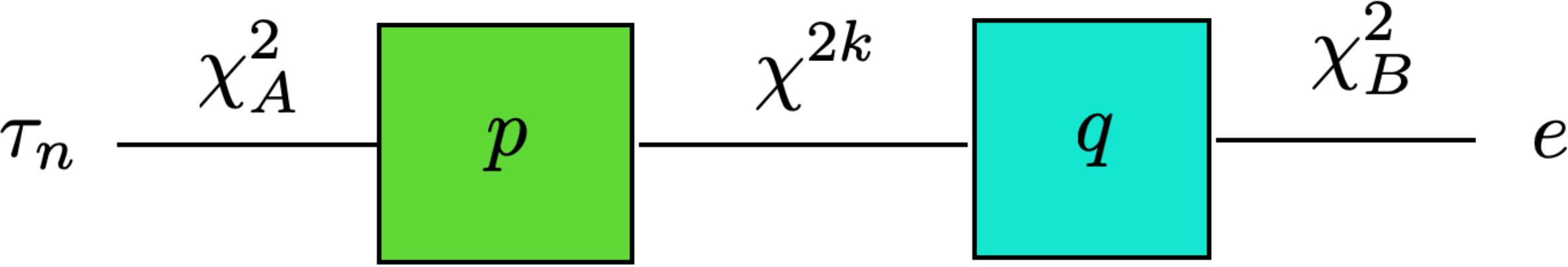}
  \caption{Calculating the partition function in the each $k$ sector is equivalent to calculating the partition function for finding the $n$-th moment of an effective two-tensor model (up to a normalization factor) with two external bonds with bond dimensions $\chi_A^2$ and $\chi_B^2$, along with an internal bond with dimension $\chi^{2k}$.}
  \label{fig:2TN_RT}
\end{figure}
One should view this effective network as an instantiation of the effective description of the $k$th superselection sector, in the sense that our result exposes the effective internal entanglement structure of a given sector.
Note that the while the picture presented here is not the same as the one proposed in \figref{fig:interpret}, they have the same entanglement spectrum up to leading order in $\chi_{A/B}^k$. 
We conjecture that by including the sum over different $k$ sectors in the full partition function, one can restore the hidden internal structure of the effective description.

The normalized spectrum of the tensor network shown in \figref{fig:2TN_RT} is the same as the spectrum of the product of two rectangular Ginibre matrices, which has been worked out in the large bond dimension limit \cite{Collins:2010fsu,2014arXiv1401.7802D,PhysRevE.92.012121} using techniques from free probability theory \cite{Mingo_2017}. The resolvent for this network can be obtained through the following cubic algebraic equation
\begin{equation}
  \label{eq:S_transform}
  \lambda W(\lambda) = (1+W(\lambda))(1+c_A W(\lambda))(1+c_B W(\lambda))
\end{equation}
where $c_A = \chi^{2k}/\chi_A^2$, $c_B = \chi^{2k}/\chi_B^2$ and $W(\lambda)$ is related to the resolvent by $\lambda R(\lambda) = (1+W(\lambda))$.
Working in the limit $c_A,c_B \ll 1$, we expand \Eqref{eq:S_transform} to first order in $c_{A/B}$:
\begin{equation}
  \label{eq:S_transform_2}
  \lambda W(\lambda) \approx (1+W(\lambda))(1+(c_A+c_B)W(\lambda)),
\end{equation}
whose solution gives a Marchenko-Pastur distribution \cite{Marchenko_1967} with parameter $c=c_A+c_B = \chi^{2k}(1/\chi_A^2+1/\chi_B^2)$. Putting back the correct normalization factors, we obtain an approximate spectrum for the $k$-th sector
\begin{align}
\label{eq:width}
  D_k(\lambda) = \frac{1}{2\pi\lambda\lambda_k(\chi_A^{-2}+\chi_B^{-2})} \sqrt{(\lambda-\lambda_{k-})(\lambda_{k+}-\lambda)},\quad 
  \lambda_{k\pm} = \lambda_k\left(1\pm \chi^k\sqrt{(\chi_A^{-2}+\chi_B^{-2})}\right)^2
\end{align}
i.e. the single eigenvalue $\lambda_k$ in each sector spreads out into a narrow peak with width $\sim 4\lambda_k\chi^k\sqrt{(\chi_A^{-2}+\chi_B^{-2})}$.
Note that the approximation in \Eqref{eq:S_transform_2} fails when $\chi^{k}$ and $\chi_{A/B}$ are of comparable size. 
This is the case in numerics when the number of eigenvalues in a higher $k$ sector approaches the finite rank constraint $\min(\chi_A^2,\chi_B^2)$. In this case one should use the full solution to the cubic equation \Eqref{eq:S_transform} instead.

Note that it is rather straightforward to adapt our calculation for arbitrary finite $\chi$ and finite integer $m$.
First, one needs to replace the number of eigenvalues in each sector by $\chi^{2k}\to d_k$ in \Eqref{eq:width_Z} to account for the finite $\chi$ effects.
Further, our previous analysis in \secref{sub:finite} shows that there are subleading eigenvalues in each sector that are suppressed by $O(\chi^{-2})$ compared to the leading one.
Repeating the calculation, we find that \emph{every} eigenvalue in a sector will obtain a width, not just the leading one.
Finally, for even integer $m$ the normalization term $(Z_{m,1})^n$ in \Eqref{eq:replica} can no longer be ignored.
Far from the EW phase transition, $Z_{m,1}$ is sharply peaked around $\min(\chi_A\chi_B,\chi_C)^{m-1}$ and it merely restores the correct normalization for the spectrum.
However around the transition, $Z_{m,1}$ has a large variance and it introduces extra fluctuations to the spectrum by spreading each mound further out in addition to the $Z_{m,n}$ effects computed here.
While this is indeed a concern for our computation, we expect our result in this subsection to hold well when our system is far from the EW phase transition at $q_{AB}\sim 1$. 

\subsection{Numerical results}
\label{sub:num}
Throughout this section, we saw that the use of TL algebra is extremely powerful and enables us to extract various analytical properties for the 2TN reflected spectrum.
To recapitulate, working in the limit $\chi_{A,B,C}\gg \chi$, it allows us to obtain $\chi$-exact spectra for even integer $m$ (\secref{sub:finite}), leading (\secref{sub:large}) and sub-leading (\secref{sub:shifts}, Appendix~\ref{sec:finite_chi}) contributions to $m\to 1$ spectra, reflected entropy (\secref{sub:sr}), as well as fluctuation effects (\secref{sub:fluc}). In this subsection we corroborate these analytical predictions by comparing to numerical results.
All the numerical results presented here are obtained directly from simulating the 2TN state by contracting two random tensors along the internal bond and calculating its reflected spectrum using exact diagonalization.

First, in \figref{fig:num_even_m} we present the histogram of the reflected spectra for $m=2$ and $m=4$. 
The analytic predictions for the locations of the eigenvalues come from Table~\ref{tab:m=2}, \ref{tab:m=4}, and the spreading within each sector is given by a modification of \Eqref{eq:width} to accommodate for finite-$\chi$ effects.
The bond dimensions used in these plots are $\{\chi_A,\chi_B,\chi_{C_1},\chi_{C_2},\chi\}=\{40,40,80,80,3\}$. 
We remind our reader that the analytic results presented here are exact for arbitrary $\chi$ and when the system is sufficiently far from phase transition. \figref{fig:num_even_m} serves as an exceptional confirmation of our formalism in the regime of small internal bond dimension $\chi$.

Moving on to the analytic continuation $m\to 1$ which is relevant to the canonical purification and reflected entropy, we present a similar histogram of the spectra for $m=1$ in \figref{fig:num_m=1}. 
The bond dimensions used here are $\{\chi_A,\chi_B,\chi_{C_1},\chi_{C_2},\chi\}=\{16,16,40,40,4\}$ (top) and $\{30,30,70,70,4\}$ (bottom).
As opposed to the case of even integer $m$, here we only have analytic control of the spectra up to first order corrections in $\chi$. The leading order analytic results of the eigenvalues are given by \Eqref{eq:pk}, and the first order corrections for $\lambda_0$ and $\lambda_1$ are given by \Eqref{eq:lambda0_shift} and \Eqref{eq:lambda1_shift}.
The corrections to higher $\lambda_k$'s are obtained via direct numerical contour integration to extract the relevant generating function coefficients, i.e. \Eqref{eq:corr_dl1} and \Eqref{eq:corr_dl2}.
The number of sectors that show up in the numerics abides the \emph{rank constraint}, namely that the total number of eigenvalues cannot exceed the dimension of the matrix. Therefore, we expect to see higher $k$ sectors materializing as $\chi_A\chi_B$ increases, which is indeed the case here: As we increase $\chi_A\chi_B$, we see an emergent fourth peak in the bottom plot compared to the top plot.

We also plot the transition of classical sector probability $p_k$ in \figref{fig:num_pk_sr} as functions of $q_{AB}$, and likewise the normalized reflected entropy $S_R/2EW$ in \figref{fig:SR}. 
The bond dimensions used in these two plots are $\{\chi_A,\chi_B,\chi\}=\{25,25,5\}$ and we vary $\chi_{C_1}=\chi_{C_2}$ to obtain different values of $q_{AB}$.
The colored dots indicate the numerical results. We present two different analytical predictions here: Solid lines are the ones that includes first order corrections (\Eqref{eq:lambda0_shift}, \Eqref{eq:lambda1_shift}, etc.), whereas dashed lines are the leading result (\Eqref{eq:pk} for $p_k$ and \Eqref{eq:sr} for $S_R$), which is valid in the limit $\chi\to\infty$.
It is evident that the first order corrections captures the non-trivial effects of small internal bond dimension surprisingly well.

\begin{figure}[h]
    \centering
    \includegraphics[width=.9\textwidth]{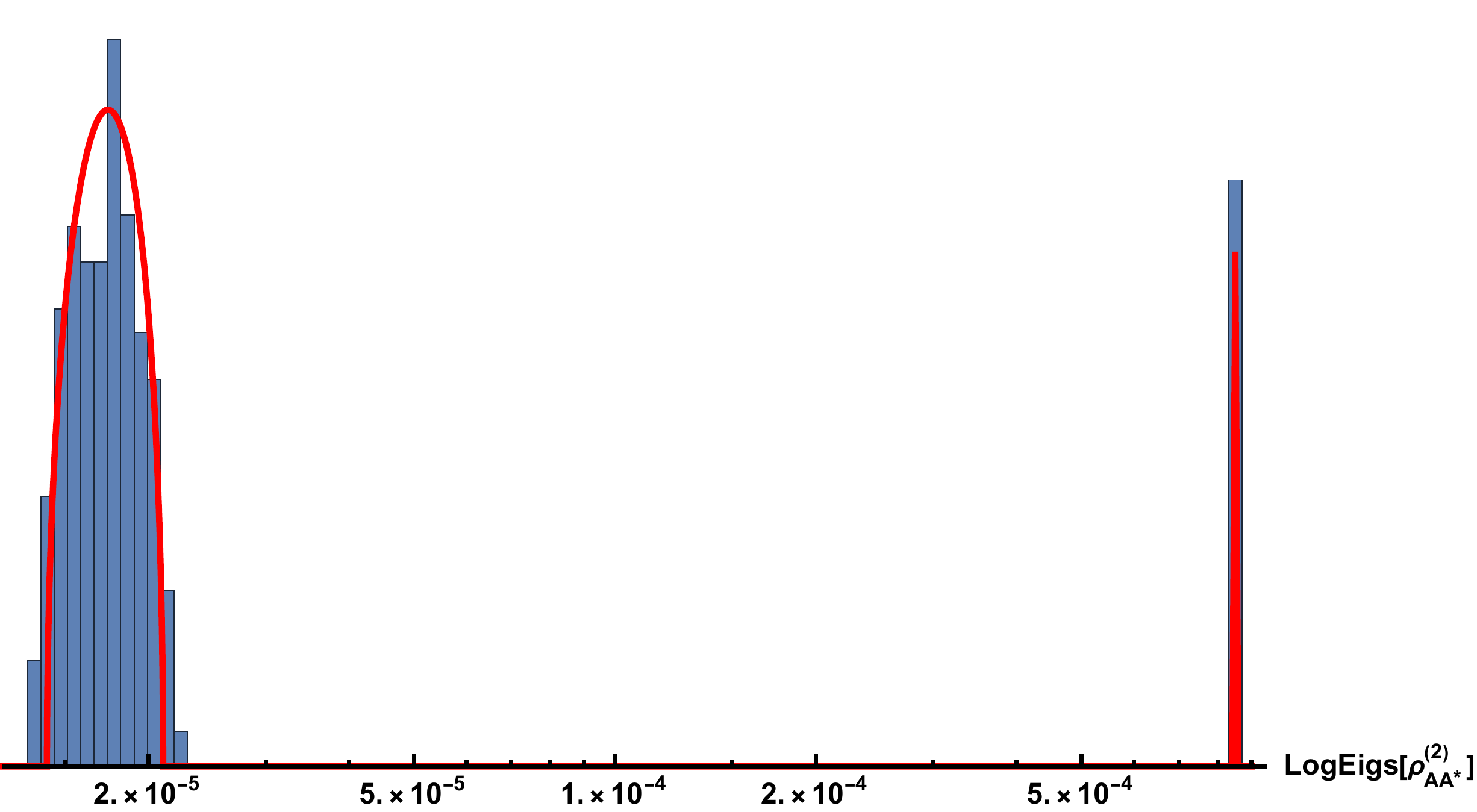}
    \includegraphics[width=.9\textwidth]{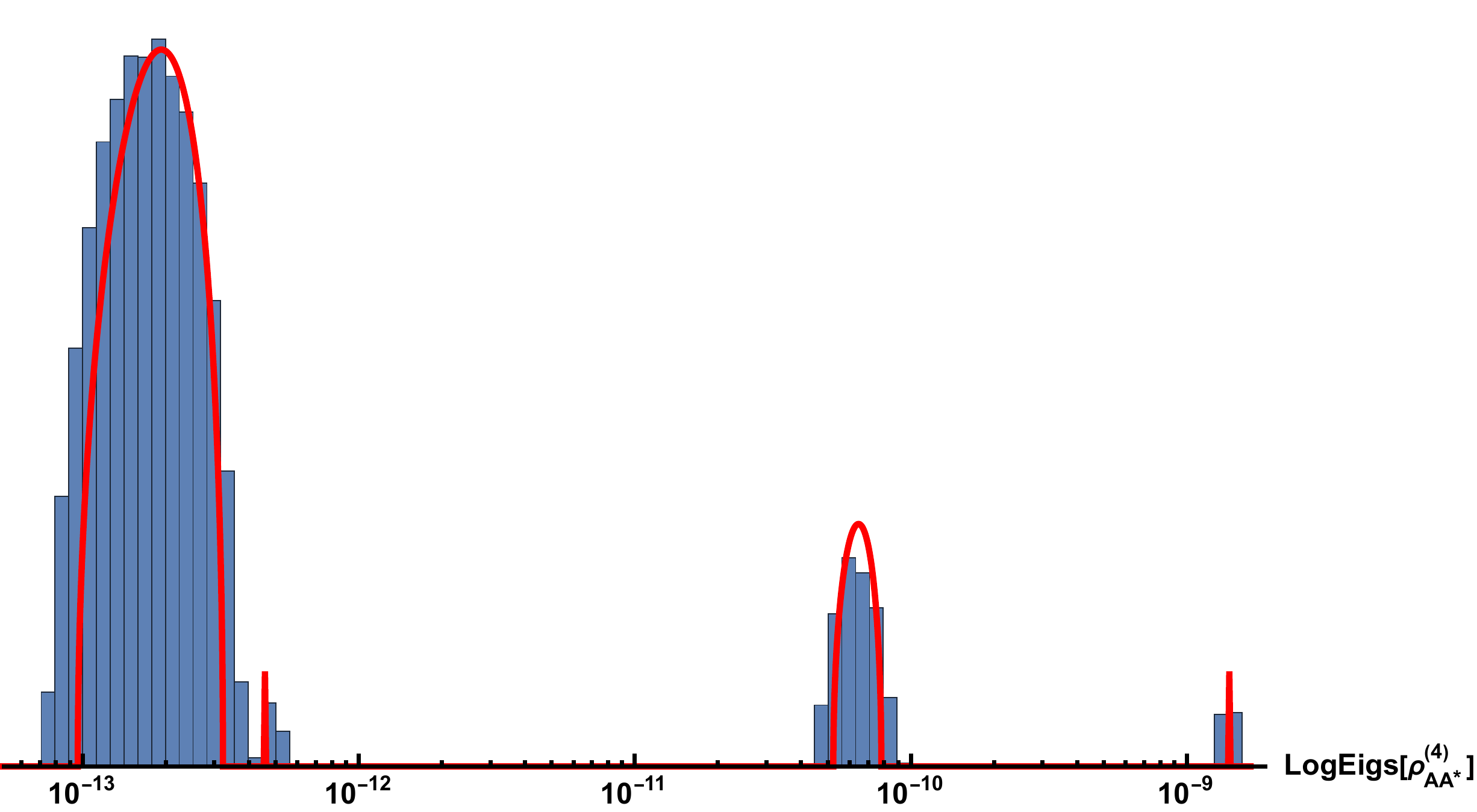}
    \caption{Plots of spectra of the (unnormalized) density matrix $\rho^{(m)}_{AB}\equiv\text{Tr}_{BB^*}\ket{\rho^{m/2}_{AB}}\bra{\rho^{m/2}_{AB}}$ for $m=2$ (top) and $m=4$ (bottom), with the red lines being the analytic predictions. Note that in the case of $m=4$ there is an additional eigenvalue that lies close to the leftmost peak. This is a finite $\chi$ effect that is not visible in our $m\to 1$ analytics. Nevertheless our results for even integer $m$ are $\chi$-exact which is well confirmed by these two plots.}
    \label{fig:num_even_m}
\end{figure}

\begin{figure}[h]
    \centering
    \includegraphics[width=.9\textwidth]{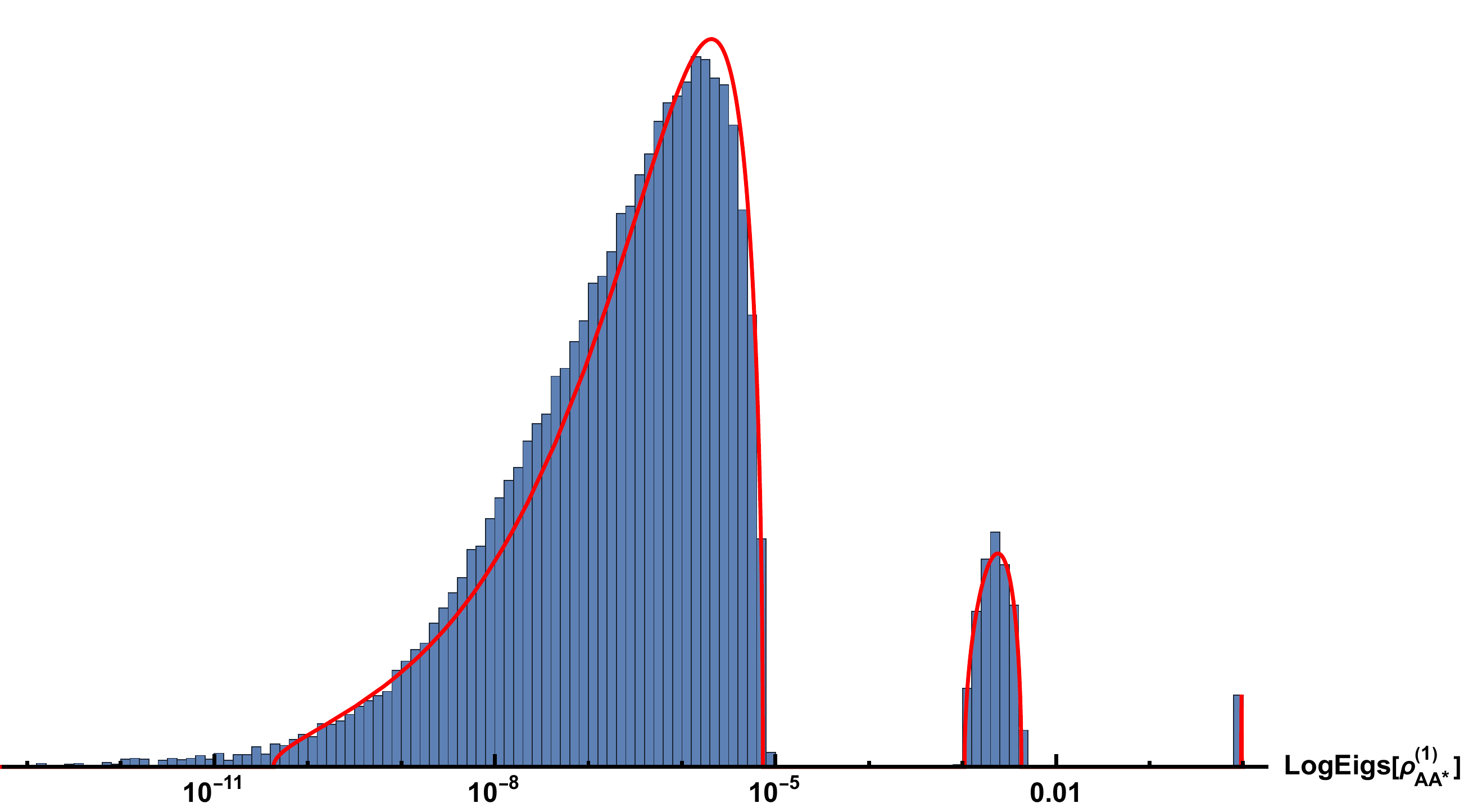}
    \includegraphics[width=.9\textwidth]{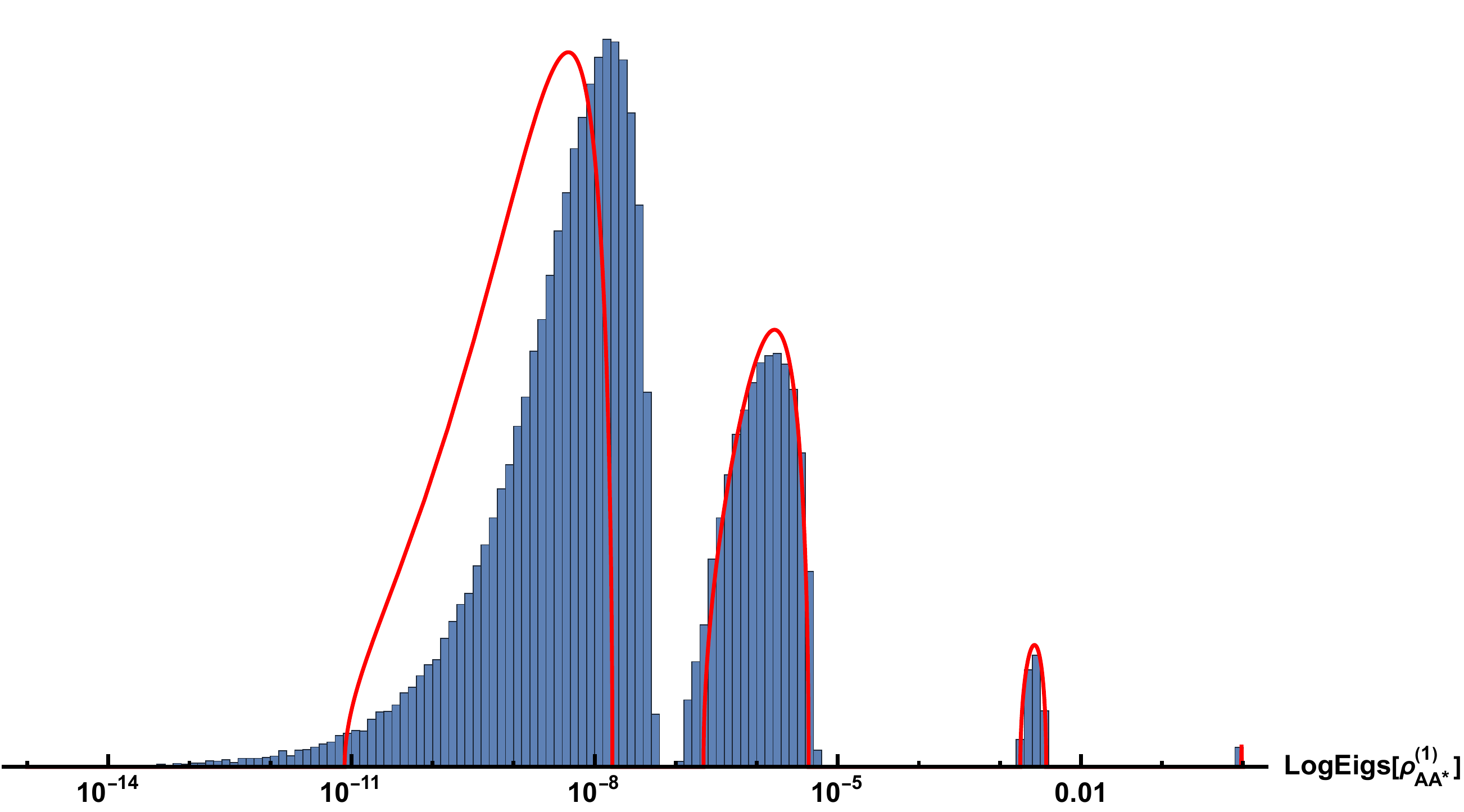}
    \caption{Plots of spectra of the $m=1$ reflected density matrix $\rho^{(1)}_{AB}\equiv\text{Tr}_{BB^*}\ket{\sqrt{\rho_{AB}}}\bra{\sqrt{\rho_{AB}}}$ for $\{\chi_A,\chi_B,\chi_{C_1},\chi_{C_2},\chi\}=\{16,16,40,40,4\}$ (top) and $\{30,30,70,70,4\}$ (bottom). The analytic predictions are shown as red lines.  The mismatch of the $k=3$ peak in the bottom plot is due to the fact that we can only work numerically in low bond dimension. As we scale up the bond dimensions in numerics, the agreement becomes better. We expect that it can also be resolved by improving the analytics, including subleading corrections to the eigenvalues beyond first order.}
    \label{fig:num_m=1}
\end{figure}

\begin{figure}[h]
    \centering
    \includegraphics[scale=.9]{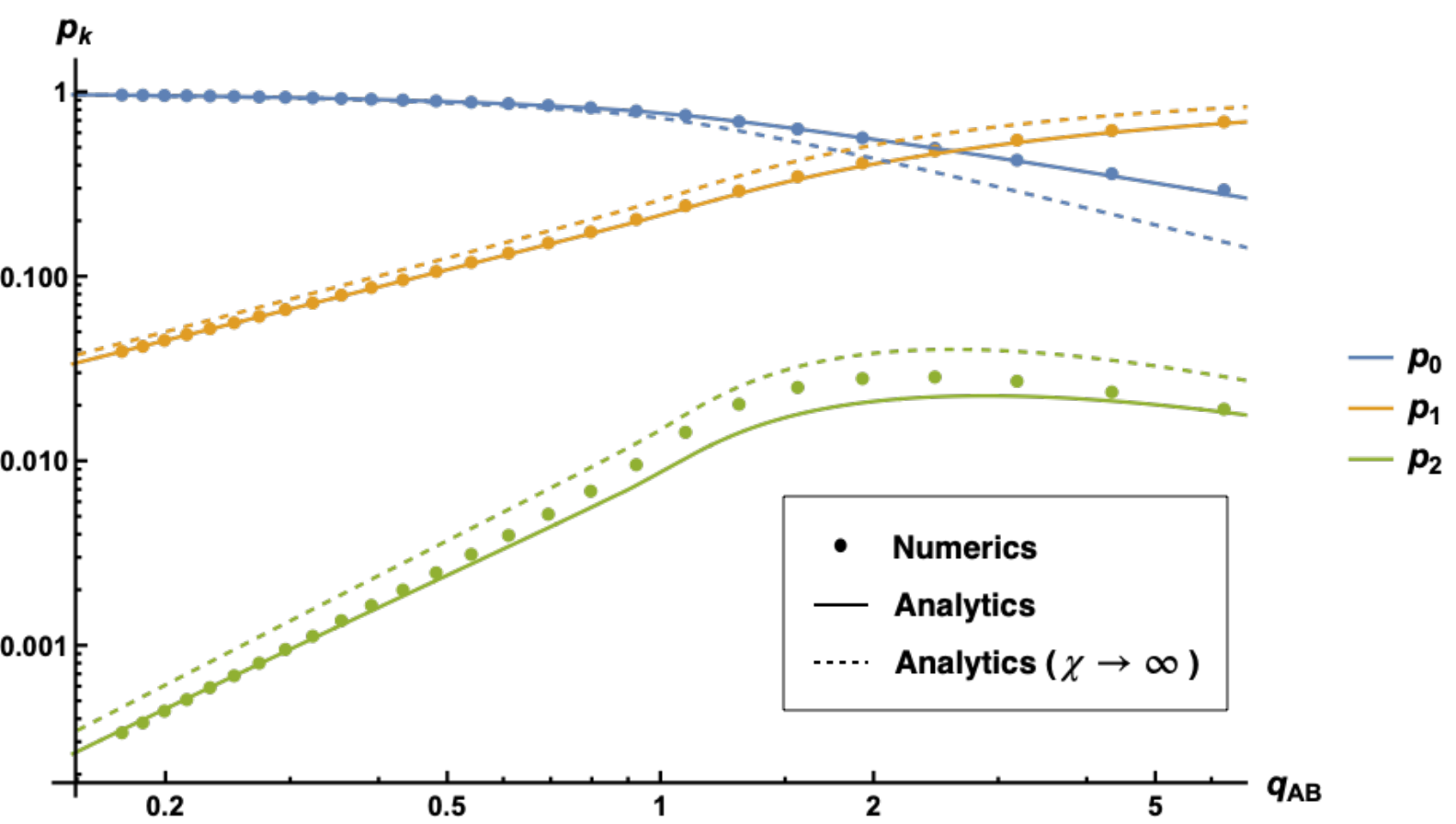}
    \caption{Plots of the sector probability $p_k$. The analytic result (solid lines) is obtained from the product of sector multiplicity $d_k$ and eigenvalues $\lambda_k$. The $d_k$ used here are the exact expressions in \Eqref{eq:dk} and the $\lambda_k$ used here includes the first order shifts. We also incorporated the exact expression of $p_k$ at $\chi\to\infty$ (\Eqref{eq:pk}) as dashed lines.}
    \label{fig:num_pk_sr}
\end{figure}
\begin{figure}[h]
    \centering
    \includegraphics[scale=.9]{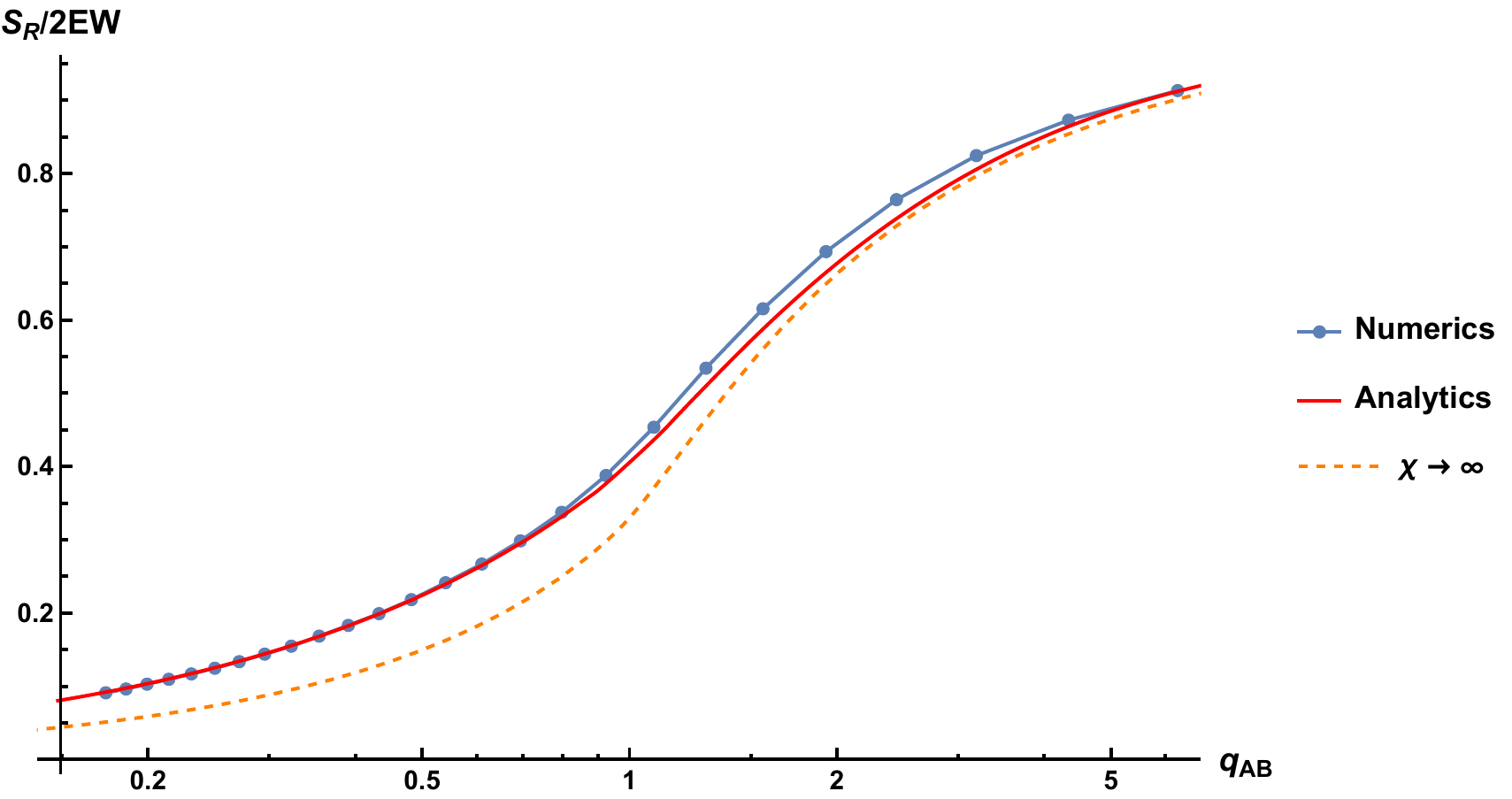}
    \caption{Plots of the reflected entropy, normalized by twice the EW cross-section $2\ln \chi$. The analytic prediction (red) is obtained by summing over the eigenvalues in $k=0,1,2$ sectors (these are the only available ones due to the rank constraint), including the first order corrections and the contribution from spreading (\Eqref{eq:width}). We have also included the form of $S_R$ in the limit $\chi\to\infty$ shown as the orange dashed line for comparison.}
    \label{fig:SR}
\end{figure}



\section{Discussion} 
\label{sec:disc}

In this paper we have continued the study of
canonical purification and reflected entropy for random tensor networks. The picture developed in Ref.~\cite{Akers:2021pvd} for the tensor network version of the gravitational gluing construction persists to more complicated tensor networks. In particular we have found a detailed match between gravitational saddles that contribute to the canonical purification with higher genus equal time surfaces and certain saddles in the statistical mechanics model governing random tensor network calculations. Area fluctuations, probed using reflected entropy, are represented by the different topological sectors. 

For the model at hand the topological sectors are governed mathematically by the Temperley-Lieb algebra. The representation theory of the TL algebra then gives a nice picture of the emergent superselection sectors that can be vividly seen in the numerics. The higher genus equal time surface arise from cutting open the TL diagrams. The genus, and hence topological index, is determined by the number of strands that are cut.

In the rest of the discussion section we summarize some possible avenues for future work and some intriguing speculative connections to the theory of emergent non-trivial von Neumann algebras.

\subsection{General RTNs and Multiboundary Wormholes}

Our analysis in this paper was focused on the 2TN model, where we performed a concrete calculation using the TL algebra, and matched the results to those obtained from the gravitational path integral in \secref{sec:gravity}. There is in fact a more general connection between RTNs and multiboundary wormholes \cite{Balasubramanian:2014hda}. In this section, we would like to make some comments on the presence of similar saddles for general multiboundary wormholes.

Consider an arbitrary RTN built out of constituent random tensors with three legs, an example is shown in \figref{fig:4TN}. Such an RTN models a multiboundary wormhole where the tensors correspond to the constituent pair-of-pants decomposition of the spatial geometry. In general, there are multiple ways to decompose a hyperbolic geometry into pairs-of-pants. However, the RTN has fixed-area surfaces which pick a preferred pair-of-pants decomposition. The network is then interpreted as a coarse-grained descriptions of the given geometry. In Ref.~\cite{Balasubramanian:2014hda}, it was argued that this model captures the entanglement entropies of such multiboundary wormholes accurately. Our results here show that the same is true for the canonical purification and the reflected entropy, thus enlarging the scope of the random state model discussed in Ref.~\cite{Balasubramanian:2014hda}.

Firstly, note that the construction of geometries contributing to the canonical purification of a multiboundary wormhole is similar to that discussed in \secref{sec:gravity}. Since all the horizons have fixed areas, the saddles contributing to the path integral in \figref{fig:gluing} are identical, except that we now have more ways to glue together the fixed-area saddles. In terms of the Cauchy surface obtained on the $\mathbb{Z}_2$-symmetric slice, the saddles can be classified by picking a particular choice of entanglement wedge and then gluing together multiple copies of the respective bulk regions.

\begin{figure}[t]
  \centering
  \raisebox{-0.5\height}{\includegraphics[scale=.33]{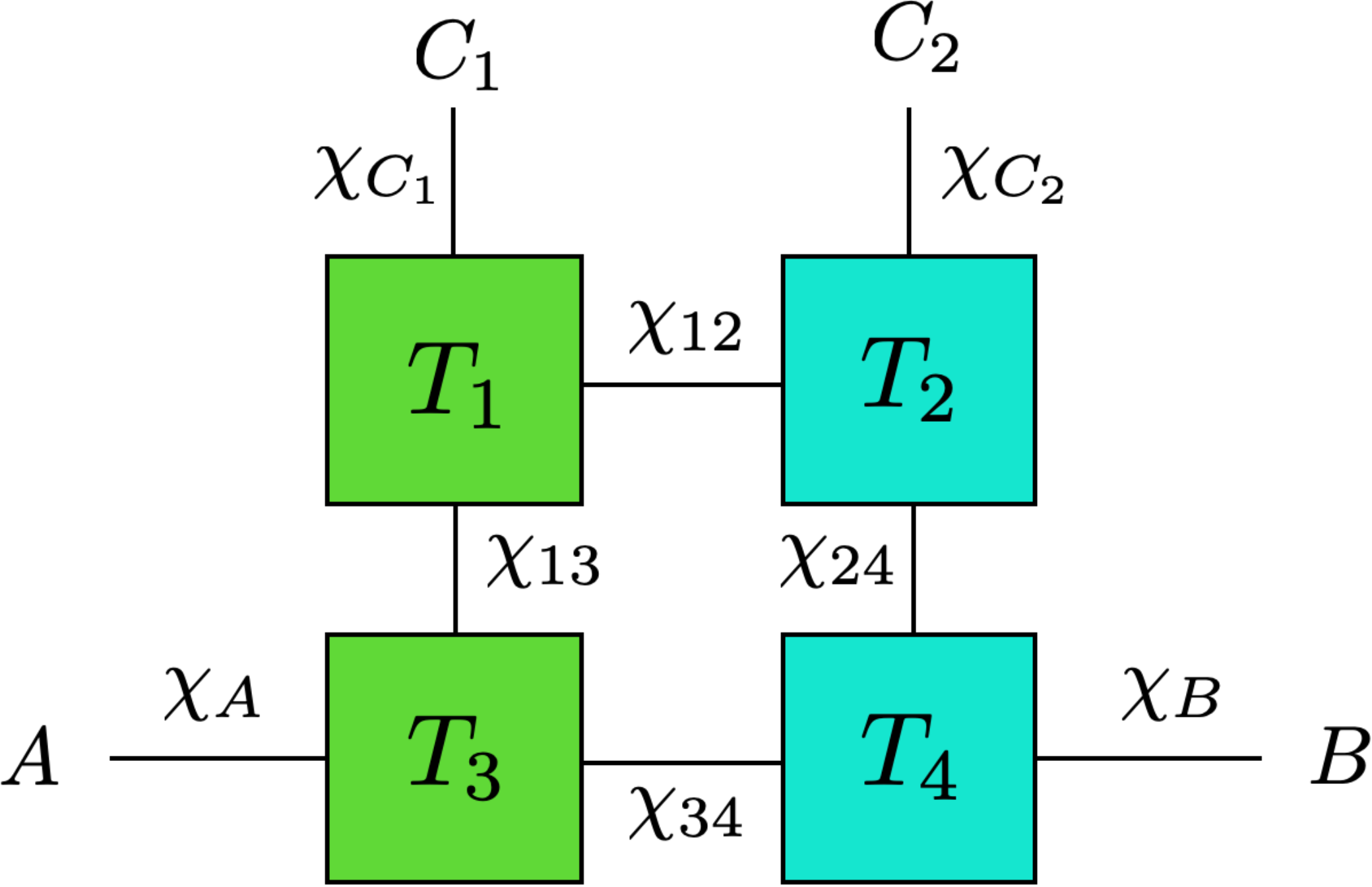}}
  \hspace{.2in}
  \raisebox{-0.53\height}{\includegraphics[scale=.33]{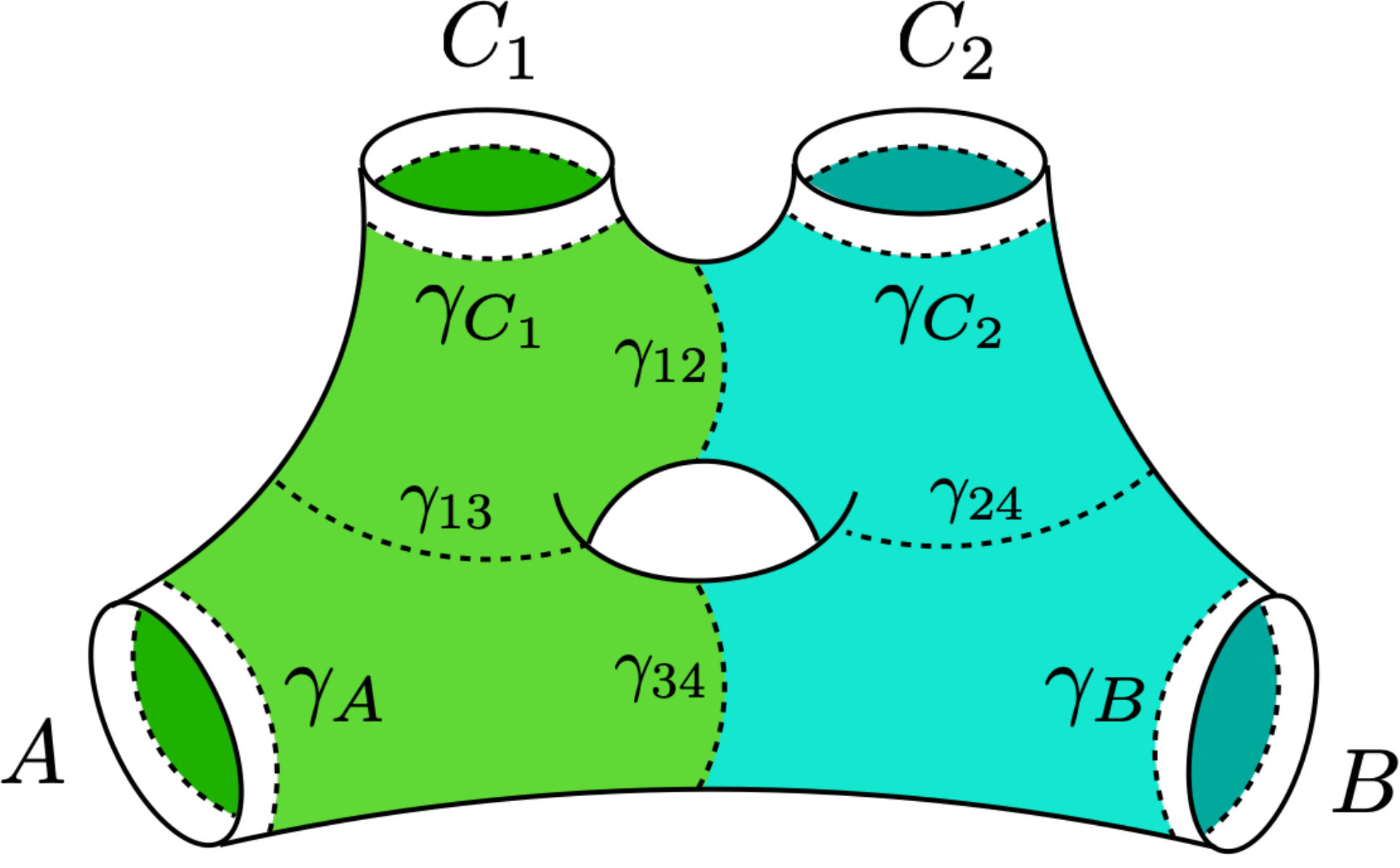}}
  \caption{(left): The 4TN network consists of four internal vertices connected by internal bonds $\chi_{ij}$. (right): It models a four-boundary wormhole with a handle.}
  \label{fig:4TN}
\end{figure}

In general, the TL algebra techniques can be applied to more general RTNs as well. For example, consider the 4TN network shown in \figref{fig:4TN}. The computation of the average partition function $\overline{Z_{mn}}$ in this network includes a sum over independent permutations $g_{i}$ on each of the tensors. As before, we can take the limit of external bond dimensions to be large. In addition, if we also take the internal vertical bonds to be large, it is easy to see that we can then restrict the sum to the non-crossing permutations $g_{1/3}\in \Gamma(g_A,e)$ and $g_{2/4}\in \Gamma(g_B,e)$ since the other permutations are suppressed. For each internal horizontal bond $\chi_{ij}$, there is a domain wall cost of the form $\chi_{ij}^{\#(g_i g_j^{-1})}$. As in \Eqref{eq:bondTL}, this can be done by introducing a TL algebra labelled by the bond dimension. Thus, $\overline{Z_{mn}}$ is computed by a product of TL traces, one for each horizontal bond. However, the analysis is complicated by the fact that the traces are coupled to each other. It would be interesting to analyze this using free probability theory techniques \cite{Cheng:2022ori,Wang:2022ots} in the future.

\subsection{Junctions for the Cross Section}

Our results in this paper rigorously apply to situations where the entanglement wedge cross section is spacelike separated from all the external RT surfaces. This guarantees that we can simultaneously fix their areas and thus, model the holographic state by an RTN. Our TL algebra techniques can then be applied to such situations. However, we now argue that the TL algebra analysis can be useful for more generic situations in holography where all the surfaces are not necessarily spacelike separated. 

For example, consider a three-boundary wormhole with horizons $\gamma_i$ for $i=A,B,C$, each of the subregions being one asymptotic boundary. As discussed in Ref.~\cite{Akers:2021pvd}, when the minimal entanglement wedge cross section is given by one of the external horizons, this situation is modelled by a single random tensor with three boundary legs, one for each subregion. However, in general there is a non-trivial cross section surface $\gamma_W$ as shown in \figref{fig:interval}, which can be important for the reflected entropy. In fact, we show in Appendix~\ref{sec:multi} that this is true for a large region in parameter space. 

In such a situation, the TL algebra analysis cannot be applied directly since there is a codimension-3 junction where the surface $\gamma_W$ meets $\gamma_C$. One can check that the areas of $\gamma_C$ and $\gamma_W$ can indeed be simultaneously fixed. This is manifest from the fact that the area of $\gamma_W$ generates a kink transform \cite{Bousso:2020yxi,Kaplan:2022orm,xi}, which preserves the minimal entanglement wedge cross section. Nevertheless, the RTN in fact fixes the areas of $\gamma_{C_1}$ and $\gamma_{C_2}$ as well, since they correspond to bonds in the network. It can be checked in various models \cite{Wu} that these areas do not commute with $\gamma_W$.

Nevertheless, the saddles contributing to the gravitational path integral for $\tr(\rho_{AB}^m)$ discussed in \secref{sec:gravity} are valid even if we don't fix the area of $\gamma_W$. Thus, for the state $\ket{\rho_{AB}^{m/2}}$, we still obtain the same geometries $\Sigma_k$, labelled by their genus. While, we cannot then use the TL algebra technology to compute the R\'enyi reflected entropies, we can gain some mileage from directly assuming that the RT formula can still be applied to each of these geometries. Since this is a superposition over a small number of geometries (not exponentially large in the entropy), we can use the expectations from Ref.~\cite{Almheiri:2016blp} to argue that the (non-linear) entropy is given by the expectation value of the minimal area, a linear operator defined on the gravitational phase space. This expectation is clearly borne out in the situation where we could fix the area of $\gamma_W$ as seen from \Eqref{eq:avg}. With this assumption, even in this situation, we expect to have an analog of \Eqref{eq:avg}. It would be interesting to generalize our TL algebra construction to include a non-flat spectrum which allows fluctuations for $\gamma_W$ and make this heuristic argument more rigorous.

\begin{figure}[t]
  \centering
  \includegraphics[scale=.35]{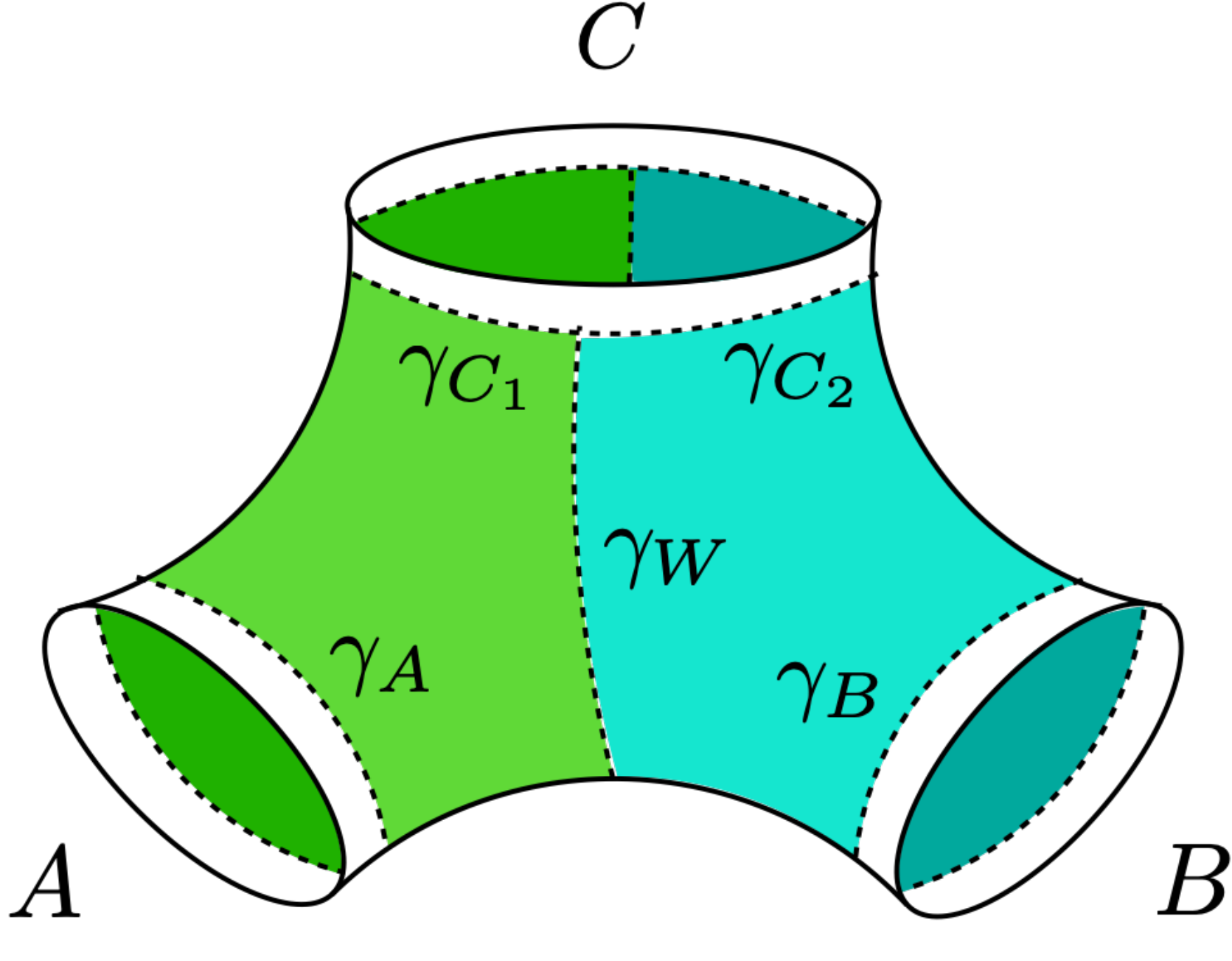}
  \caption{A three-boundary wormhole with horizons labeled as $\gamma_{A/B/C}$ and the non-trivial EW cross-section  $\gamma_W$. While one cannot fix the individual area of $\gamma_W$ and $\gamma_{C_1/C_2}$ simultaneously, the area of $\gamma_W$ and the sum $\gamma_C=\gamma_{C_1}+\gamma_{C_2}$ can be fixed simultaneously.}
  \label{fig:interval}
\end{figure}

With the above heuristic understanding, our 2TN results also apply to a rather generic situation that arises for subregions in a CFT (see \figref{fig:EWCS}). In such a situation, the areas of the extremal surfaces or equivalently the external bond dimensions are IR divergent, whereas the cross section is finite. The external areas can be regularized by allowing for a small splitting between the regions. Having done so, the density matrix $\rho_{AB}$ can in fact be defined rigorously by the split construction described in Ref.~\cite{Dutta:2019gen}. The regularization removes subtleties associated with fixing an IR divergent area. With this understanding, the analysis is similar to the heuristic argument described above for the three-boundary wormhole and thus, our results are rather generic.

\subsection{Emergent von Neumann algebras}

Temperley-Lieb (TL) algebras play an important role in the study of type-II$_1$ subfactors, initiated by Jones. A subfactor is a subalgebra 
 $\mathcal{B} \subset \mathcal{A}$ of von Neumann algebras both of which do not have centers. In this case the TL algebra emerges from a sequence of further subalgebras $ \mathcal{B} \subset \mathcal{A} \subset \mathcal{A}_1 \subset \mathcal{A}_2 \ldots $ constructed using the Jones basic construction \cite{Jones1983}. Then the TL operators are generated by a sequence of projection operators that arise from the basic construction.  
The TL$_m$ algebra for even $m$ then appears in the relative commutants $\mathcal{A}_{i+m} \cap \mathcal{A}_{i}'$ of two of the algebras in this sequence. The algebra generated by all of these relative commutants would reproduce a type-II$_1$ algebra. We might then speculate that our model gives one more avenue \cite{Leutheusser:2021qhd,Witten:2021unn,Chandrasekaran:2022cip} through which non-trivial von Neumann algebras can arise from gravitational like theories. 

Comparing to the emergent TL algebras in our tensor network model the relevant inclusion involves decreasing the number of handles on the fixed time slice.  To really see an emergent type-II$_1$ algebra we would need to move to high topological index. In particular $m \rightarrow \infty$ would be interpreted as giving rise to such a non-trivial von Neumann algebra.
Interestingly by continuing $m$ away from an even integer we have a mechanism by which arbitrarily high topological index can arise, unfortunately however these states become increasingly suppressed at high $k$ by the sector probabilities. Near the phase transition all sectors are excited, however to really claim an emergent type-II$_1$ algebra we would need to somehow project to a high $k$ sector so that $p_k$ is peaked around $k_0$ which is then approaching $\infty$. If we managed this then the leading reflected entropy $\approx k_0 \ln \chi$ is diverging with $k_0$, a necessary condition for an emergent type-II$_1$ von Neumann algebra. Perhaps this can be achieved by perturbing the original density matrix by some $\chi_{A,B}$ dependent operator that excites the topological mode $k_0(\chi_{A,B})$, with $k_0 \rightarrow \infty$ as we send $\chi_{A,B} \rightarrow \infty$ while holding fixed $\chi$. We certainly need to send the external bond dimensions to $\infty$ so that the rank condition does not come into play. A natural guess for such a perturbation is to apply $\rho_{AB}^{is}$ (half sided) modular flow to the original state $\left|\rho_{AB}^{1/2} \right>$. 
In fact one can show, using similar techniques to Section~\ref{sec:2TN}, that under modular flow the probabilities $p_k$ get modified to
\begin{equation}
    p_k(s) = q^{1-k}_{AB} \, g_{1/2+is,k}(q_{AB}) g_{1/2-is,k}(q_{AB})
\end{equation}
at leading order in $\chi$.
This change of $p_k(s)$ is seemingly consistent with the above requirements. We have confirmed this numerically, see \figref{fig:pk_modular}. We would then need to pick $s(\chi_A,\chi_B)$ diverging and one important question is can we control such a computation while taking $s$ large? 
This would be an interesting avenue to pursue since it would give rise to a computable model of topological gravitational like fluctuations. It is also important to understand the exact nature of the putative emergent type-II$_1$ von Neumann algebra other than simply via entropy computations. 

\begin{figure}[t]
    \centering
    \includegraphics[width=.7\textwidth]{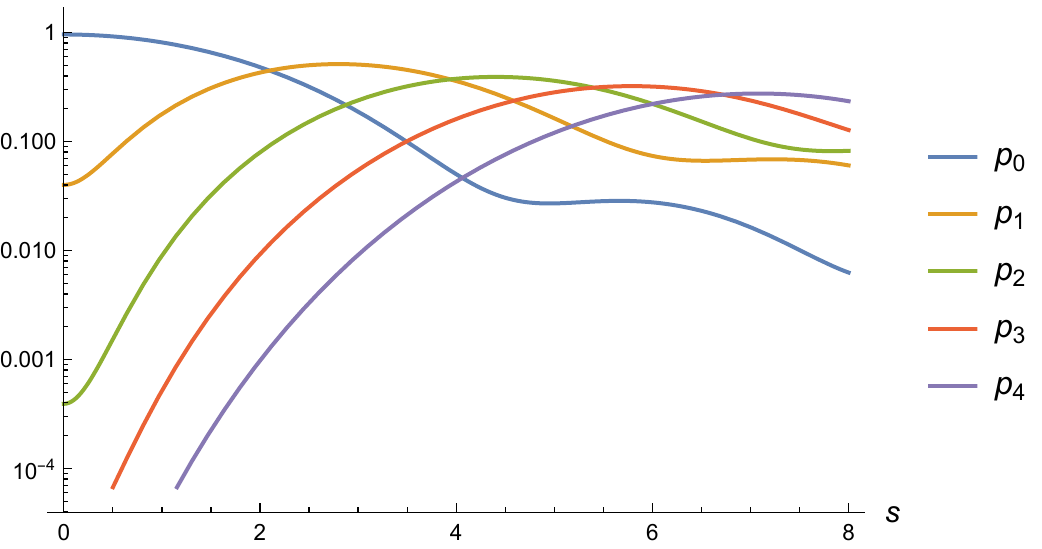}
    \includegraphics[width=.7\textwidth]{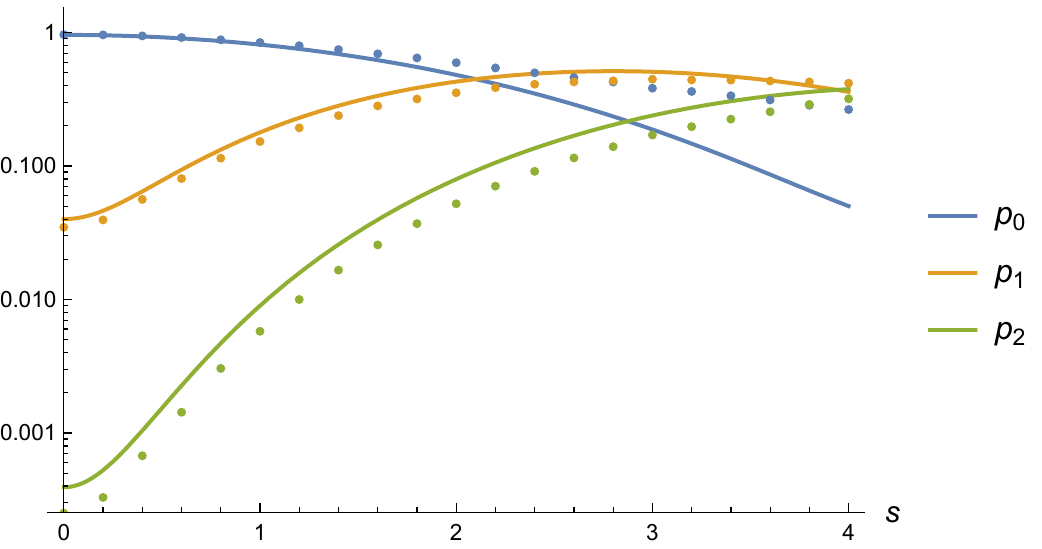}
    \caption{(Top) Sketch of the ($\chi$-leading) modular flowed sector probability $p_k$ as a function of modular parameter $s$. One can see the shift of dominance to higher $k$ sectors as $s$ increases.
    (Bottom) The numerical result of $p_k$ for $k=0,1,2$, shown as colored dots. The bond dimensions used in the numerics are $\{\chi_A,\chi_B,\chi_{C_1},\chi_{C_2},\chi\}=\{16,16,40,40,4\}$. One sees a good match with the analytical prediction at small $s$. To match the behavior at larger $s$ one would need to access higher $k$ sectors in the numerics.}
    \label{fig:pk_modular}
\end{figure}

There are also some superficial similarities to Ref.~\cite{Lin:2022rbf} that are worth exploring further. Ref.~\cite{Lin:2022rbf} gives a Hilbert space interpretation of the double scaled SYK model, for which a diagrammatic solution was given in Ref.~\cite{Berkooz:2018jqr}. These are the chord diagrams and the Hilbert space description is in terms of cutting open these chord diagrams, similar to the Temperley-Lieb Hilbert spaces. In this case an emergent type-$II_1$ von Neumann algebra naturally arises.


\acknowledgments


We would like to thank Xi Dong and Fikret Ceyhan for useful discussions. SL would especially like to thank Shiliang Gao and Jingwei Xu for their help regarding various mathematical aspects in this note. PR is supported in part by a grant from the Simons Foundation, and by funds from UCSB. 
CA is supported by the Simons foundation as a member of the It from Qubit collaboration, the NSF grant no. PHY-2011905, and the John Templeton Foundation via the Black Hole Initiative.
This material is based upon work supported by the Air Force Office of Scientific Research under award number FA9550-19-1-0360. This work benefited from the Gravitational Holography program at the KITP and we would like to thank the KITP, supported in part by the National Science Foundation under Grant No. NSF PHY-1748958.

\appendix 

\section{Multiboundary wormholes}

\label{sec:multi}

\begin{figure}[t]
  \centering
  \includegraphics[width=.6\textwidth]{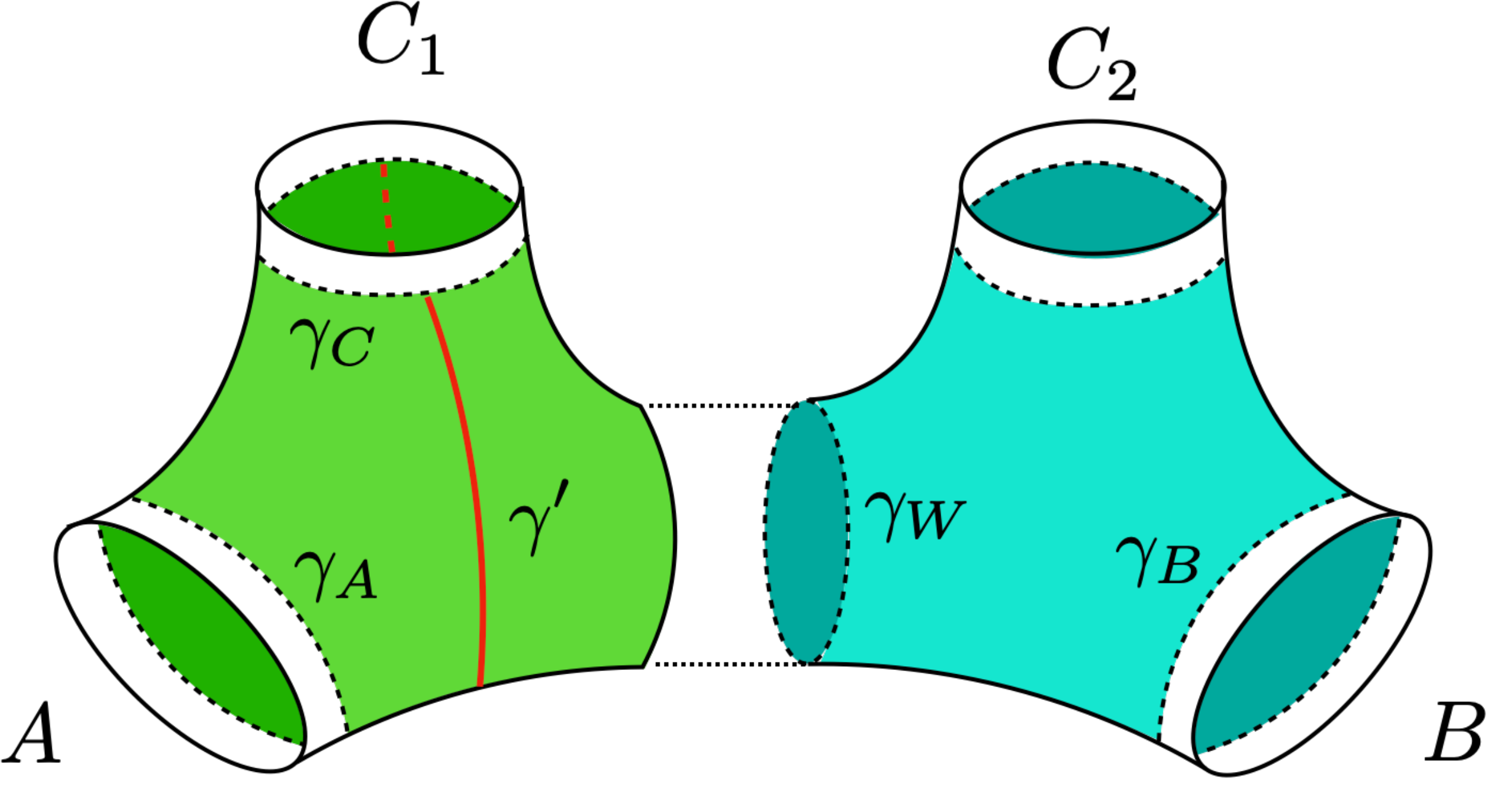}
  \caption{A pair of pants decomposition of a four-boundary wormhole. The constituent three-boundary wormholes have horizons $\gamma_i$ with area $L_i$. The surface $\gamma_W$ was treated as the minimal entanglement wedge cross section in our analysis. However, another possible non-trivial entanglement wedge cross section surface $\gamma'$ exists with an area $L'$ that could be lesser in regions of parameter space.}
  \label{fig:pants}
\end{figure}

In this appendix, our primary goal is to analyze the extremal surfaces in a multiboundary wormhole. In particular, we are interested in understanding the competition between different candidates for the minimal entanglement wedge cross section.

The four boundary wormhole we are interested in can be decomposed into two pair-of-pants geometries (see \figref{fig:pants}). Without loss of generality, consider one of the constituent geometries with extremal surfaces $\gamma_i$ that have area $L_i$ for $i=A,C_1,W$, corresponding to each of the asymptotic boundary regions as shown in \figref{fig:pants}. In this geometry, there are three potential candidates for the minimal entanglement wedge cross section for the region $AB$, i.e., $\gamma_A$ and $\gamma_W$, which have obvious analogs in the 2TN model, as well as a non-trivial surface $\gamma'$, whose area $L'$ is completely fixed in terms of the moduli $L_i$. In order for the tensor network to faithfully model this wormhole's minimal cross section, we need to ensure that $L_W<L'$.

To compute $L'$ as a function of $L_i$, we can use the symmetries of the problem and identities from hyperbolic geometry (see Ref.~\cite{Hayden:2021gno} and references therein). In particular, an identity satisfied by right-angled pentagons in hyperbolic space is
\begin{equation}
	\sinh(a)\sinh(b)=\cosh(c),
\end{equation}
where $a$ and $b$ are the lengths of adjacent sides of the pentagon and $c$ is the unique side not adjacent to either of them. Note that all lengths are measured in units of the AdS scale $l$. Suppose $\gamma'$ splits $\gamma_{C_1}$ into portions of length $x$ and $L_{C_1}-x$, then we have
\begin{align}
\begin{split}
	\sinh(x)\sinh\left(\frac{L'}{2}\right) &= \cosh\left(\frac{L_W}{2}\right),\\
	\sinh\left(\frac{L_{C_1}}{2}-x\right)\sinh\left(\frac{L'}{2}\right) &= \cosh\left(\frac{L_A}{2}\right).
 \end{split}
\end{align}
Moreover, we are interested in a limit where $L_A,L_{C_1}\gg L_W$. Solving the above set of equations in this limit, we obtain
\begin{equation}
L'= 2 \sinh ^{-1}\left(\sqrt{\exp\left(L_A-L_{C_1}\right)+2 \exp\left(\frac{L_A-L_{C_1}}{2}\right)\cosh\left(\frac{L_W}{2}\right)}\right).
\end{equation}
We can now compare $L'$ to $L_W$ either by plotting the function as in \figref{fig:plotL}, or by comparing them in various limits. When $L_A> L_{C_1}$, we find that $L'\gg L_W$ which means we can ignore the surface $\gamma'$ for the cross section. On the other hand, for $L_A<L_{C_1}$, we find that $L'\ll L_W$ and thus, it is very important to consider $\gamma'$. This means that the 2TN approximation fails as we go far into the disconnected phase. However, for our calculations we were mostly interested in the region near phase transitions. Thus for simplicity, consider the case $L_A=L_{C_1}$. Then, for small $L_W$, we have $L'=2\sinh^{-1}(\sqrt{3})\gg L_W$. For large $L_W$, we instead have $L'=\frac{L_W}{2}\ll L_W$. As seen from \figref{fig:plotL}, there is a transition at an $O(1)$ value of $L_W$. Identical arguments can be applied to the other constituent pair-of-pants geometry. Thus, as long as we tune the moduli to such a regime, we find a wormhole whose minimal entanglement wedge cross section is indeed given by $\gamma_W$ and is modelled well by the 2TN. 

\begin{figure}[t]
  \centering
  \includegraphics[width=.6\textwidth]{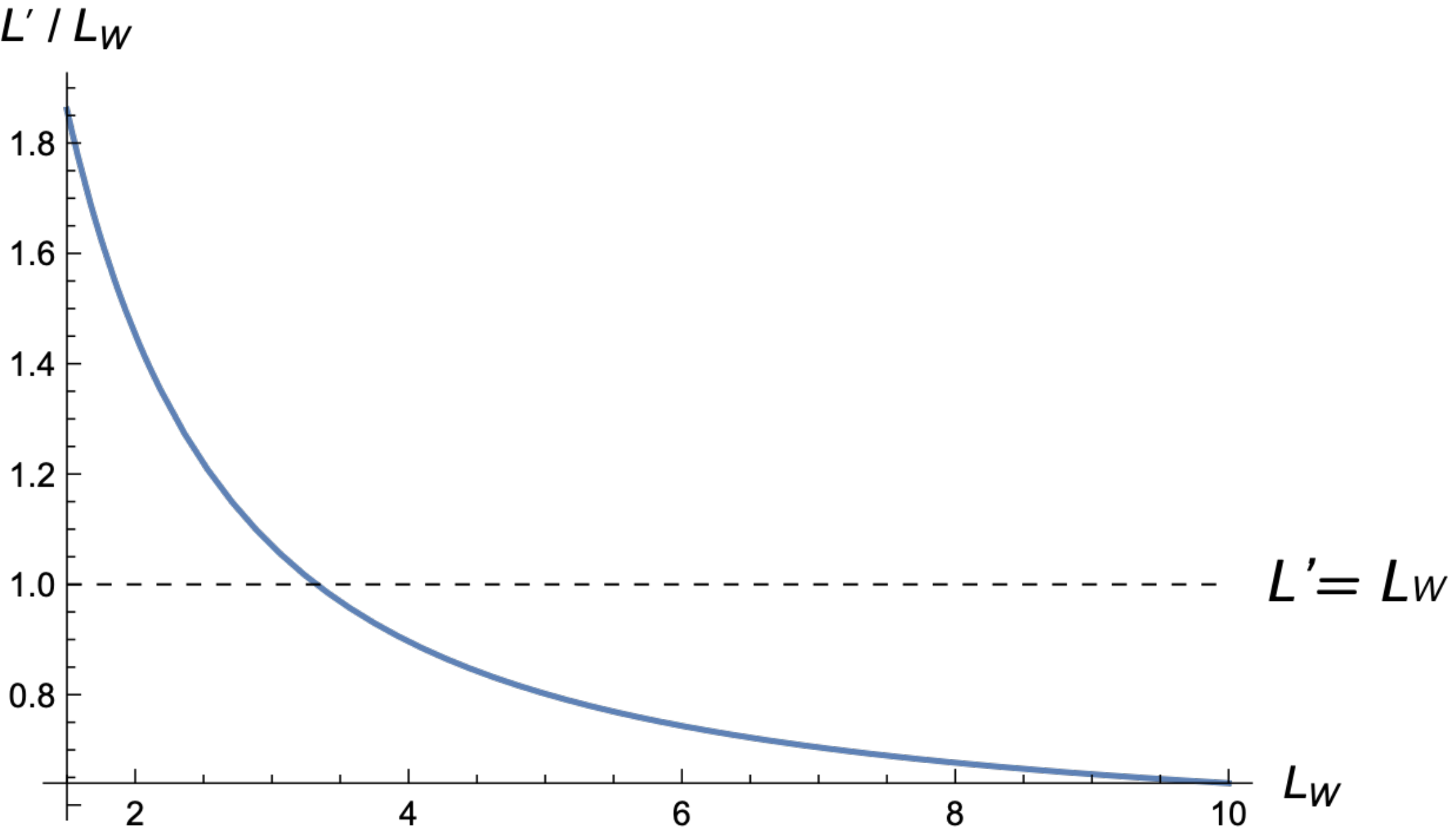}
  \caption{A plot of $L'/L_W$ as a function of $L_W$ for fixed $L_A$ and $L_{C_1}$. The dashed line labels the location where the shift of minimal cross-section happens.}
  \label{fig:plotL}
\end{figure}

On the other hand, it isn't clear whether such saddles are dominant in the gravitational path integral. An analysis similar to Ref.~\cite{Maxfield:2016mwh} would be useful to determine if this is the case. Even if such a saddle is not dominant, we can consider a few alternate options: matter can be used to source a larger interior region in which case it is clear that $L'>\min(L_A,L_W)$. It may also be possible to simply choose to project on the relevant saddles \cite{Akers:2019nfi}.

\section{Temperley-Lieb algebra}
\label{sec:TLalgebra}
In this appendix, we will summarize basic aspects of the Temperley-Lieb (TL) algebra and its representation theory.
There are many different flavors of classifying the representations of TL algebra, with the two most prominent approaches being:
1. As a quotient of the Iwahori-Hecke algebra \cite{10.2307/1971403,10.1112/S0024611597000282}, or 2. As a algebraic module acting on a specific class of ``link diagrams''.
We will take the second approach here, also known under the name of \textit{standard module} \cite{Ridout:2012gg,Westbury1995}, or \emph{cell modules} \cite{Graham1996}.
The goal of this appendix is to provide a pedagogical review on various properties of the TL algebra and the standard module used in our computations.
Due to the length constraint we will omit the proof of many important theorems. Interested reader can refer to \cite{Ridout:2012gg} for a more complete and rigorous treatment on this subject.

\subsection{Basic definitions}
First, we define, for any positive integer $n$, an \textit{n-diagram}. It is constructed by first drawing two parallel lines, each with $n$ marked points. Then, the set of $2n$ points are connected pairwise via $n$ non-crossing strands that lies entirely within the space between two parallel lines.
We also define a pairwise product of two $n$-diagrams by concatenating two diagrams side by side and replacing each closed loop by a factor of $\chi$. For example, for $n=4$ we write:
\begin{equation}
  \begin{matrix}
    \includegraphics[width=.6\textwidth]{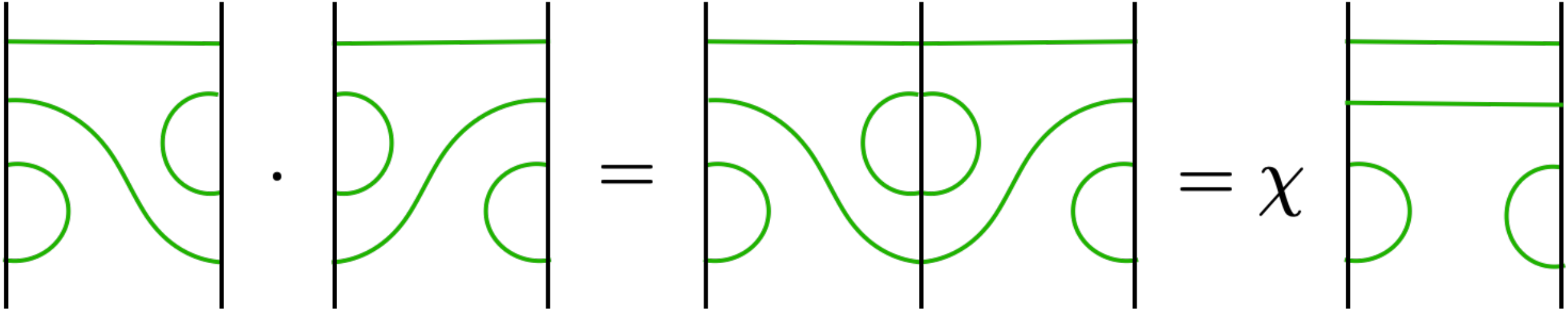}
  \end{matrix}
\end{equation}
This product is associative by construction.
Although the parameter $\chi$ is considered to be complex in more general settings, in what is relevant to this note we will assume $\chi\in \mathbb{R}_{+}$ and we will often consider the limiting case where $\chi\gg 1$.
\begin{definition}[Temperley-Lieb algebra]
\label{def:TL}
  The Temperley-Lieb algebra $\text{TL}_n$ for a positive integer $n$, is the (complex) vector space spanned by $n$-diagrams. This vector space is equipped with a pairwise product by taking the diagrammatic product from above and extend it to the whole vector space bilinearly. 
\end{definition}
\begin{lemma}
\label{lem:q_catalan}
  The dimension of $\text{TL}_n$ vector space is equal to the $n$-th Catalan number $C_n$.
\end{lemma}
\begin{proof}
  To find the dimension of $\text{TL}_n$, we need to count the total number of distinct $n$-diagrams for given $n$.
  There is a bijection $D(\cdot)$ between $n$-diagrams and $\text{NC}_n$, the non-crossing permutations on $n$ elements, defined as follows.
  First, one deforms the $n$-diagram $h$ by rotating the RHS of the diagram by 180 degrees and append it below the LHS.
  Having done so one now has a diagram with non-crossing strands connecting $2n$ marked points arranged on a single line.
  These diagrams are of one to one correspondence to the double line notations of non-crossing permutations on $n$ elements.
  This fact can be made clear by concatenating the left of the diagram with $n$ ``caps'' and assigning a number to each cap.
  After doing so, the diagram factors into a set of non-crossing closed blocks partitioning the set $\{1,\cdots,n\}$, which can be further identified to the cycle decomposition of a non-crossing permutation $D(h)\in\text{NC}_n$.
  We give a example of the map $D$ below.
  \begin{equation*}
    \includegraphics[width=.75\textwidth]{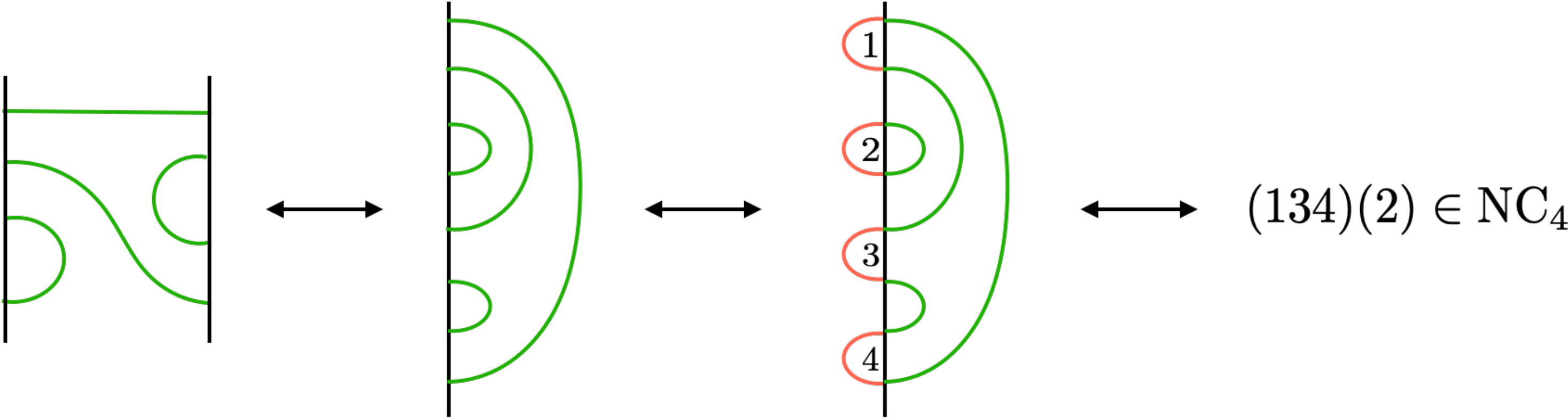}    
  \end{equation*} 
  The number of distinct non-crossing permutations on $n$ elements is given by Catalan numbers $C_n$.
  The detailed form of $C_n$ can be worked out using a simple generating function, but for reference here let us consider a slight generalization by weighing each cycle of $D(h)$ by a power of $q$. Define
  \begin{equation}
    \begin{matrix}
        \includegraphics[scale=.3]{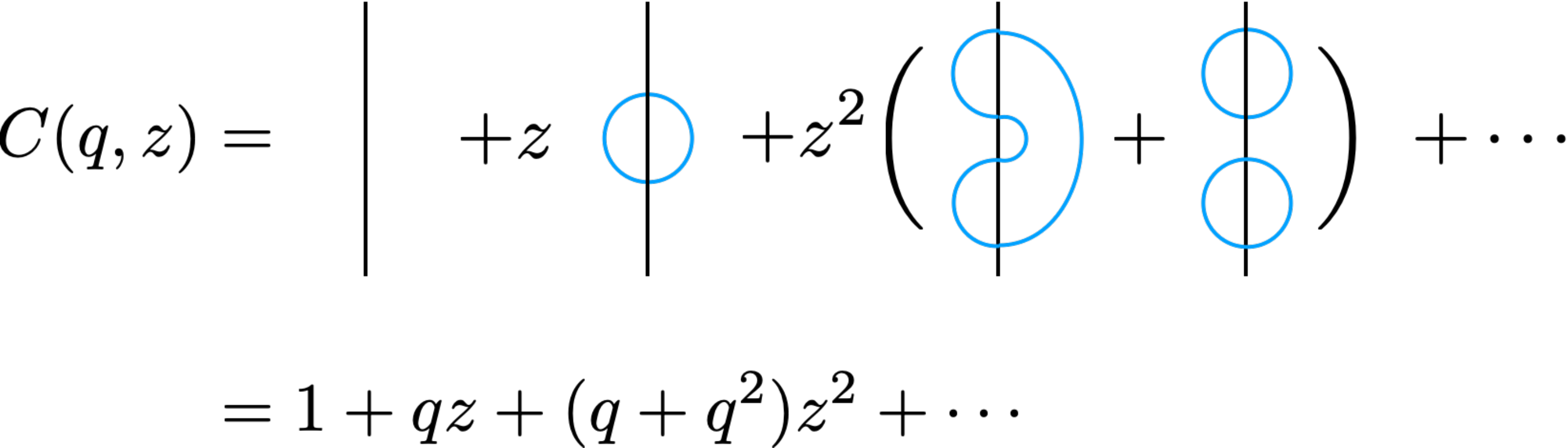}      
    \end{matrix}
  \end{equation}
  to be the $q-$weighted generating function for $NC_n$ and $F(q,z)$ the generating function for connected diagrams, i.e.
  \begin{equation}
    C(q,z) = 1+F(q,z)+F(q,z)^2+\cdots = \frac{1}{1-F(q,z)}
  \end{equation}
  $C(q,z)$ and $F(q,z)$ satisfies the following Schwinger-Dyson equation
  \begin{equation}
    \begin{matrix}
      \includegraphics[scale=.3]{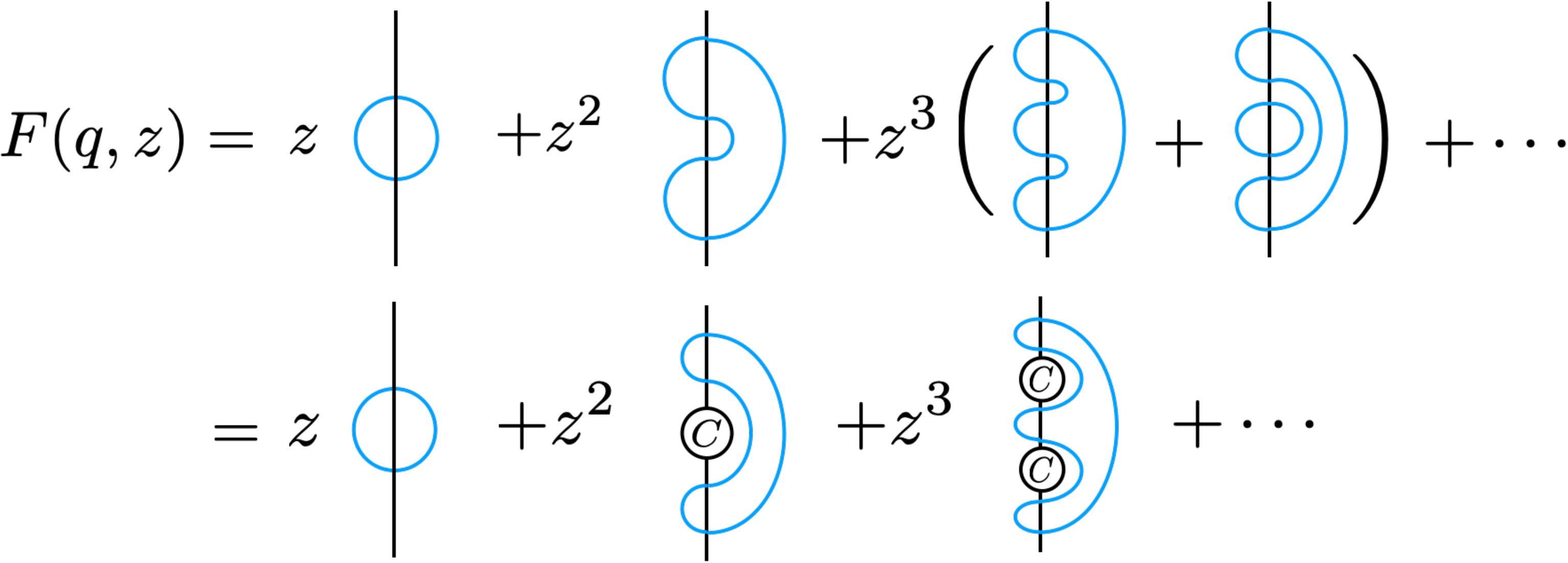}      
    \end{matrix}
  \end{equation}
  or
  \begin{equation}
    F(q,z) = \frac{qz}{1-zC(q,z)}
  \end{equation}
  which can be solved to give
  \begin{equation}
    C(q,z) = \frac{1-z(q-1)-\sqrt{(1+z(q-1))^2-4qz}}{2z}
  \end{equation}
  where we pick the sign of the solution by matching the small $z$ behavior.
  The expansion coefficients for $C(q,z)$ are called \textit{q-Catalan numbers}, which can be extracted using the contour integral trick:
  \begin{equation}
    C_n(q) = \frac{1}{2\pi i} \oint \frac{dz}{z^{n+1}}C(q,z)
    =
    \begin{cases}
      q \;_2F_1(1-n,-n;2;q), \quad &q\le 1, \\
      q^n\;_2F_1(1-n,-n;2;q^{-1}), \quad &q > 1
    \end{cases}
  \end{equation}
  The ordinary Catalan number $C_n$ is related to $C_n(q)$ by
  \begin{equation}
    C_n = C_n(q=1) = \frac{1}{n+1}
    \begin{pmatrix}
      2n \\ n
    \end{pmatrix}
  \end{equation}
\end{proof}

It is easy to check that the diagram
\begin{equation}
    \begin{matrix}
      \includegraphics[scale=.33]{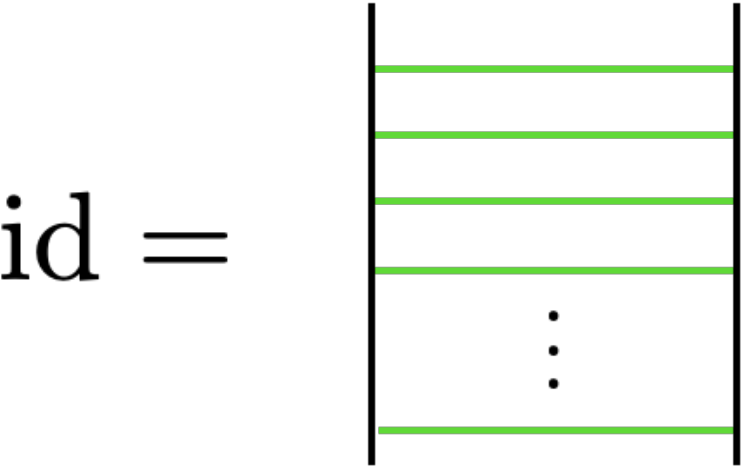}      
    \end{matrix}
\end{equation}
acts as an unit in the algebra.
Aside from the unit, every $n$-diagram can be obtained by multiplying a set of generators $e_1,e_2,\cdots,e_{n-1}$, where
\begin{equation}
  \begin{matrix}
      \includegraphics[scale=.33]{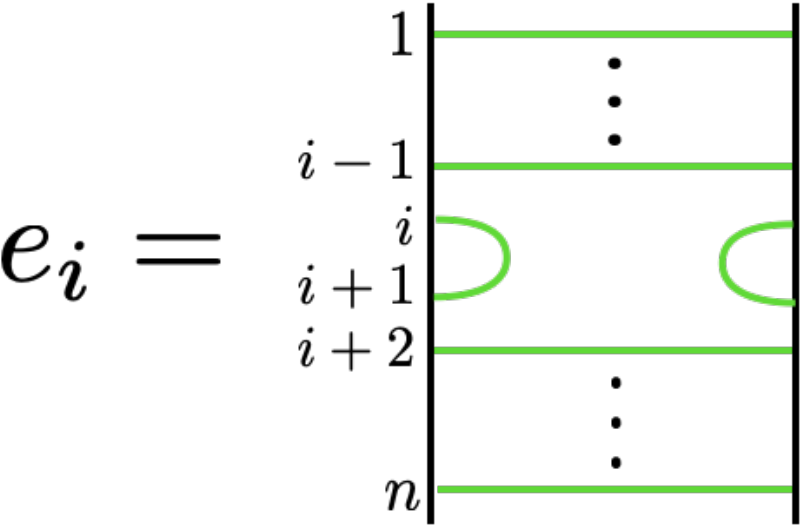}      
    \end{matrix}
\end{equation}
These generators satisfy the following relations
\begin{align}
    e_i^2&=\begin{matrix}
      \includegraphics[scale=.3]{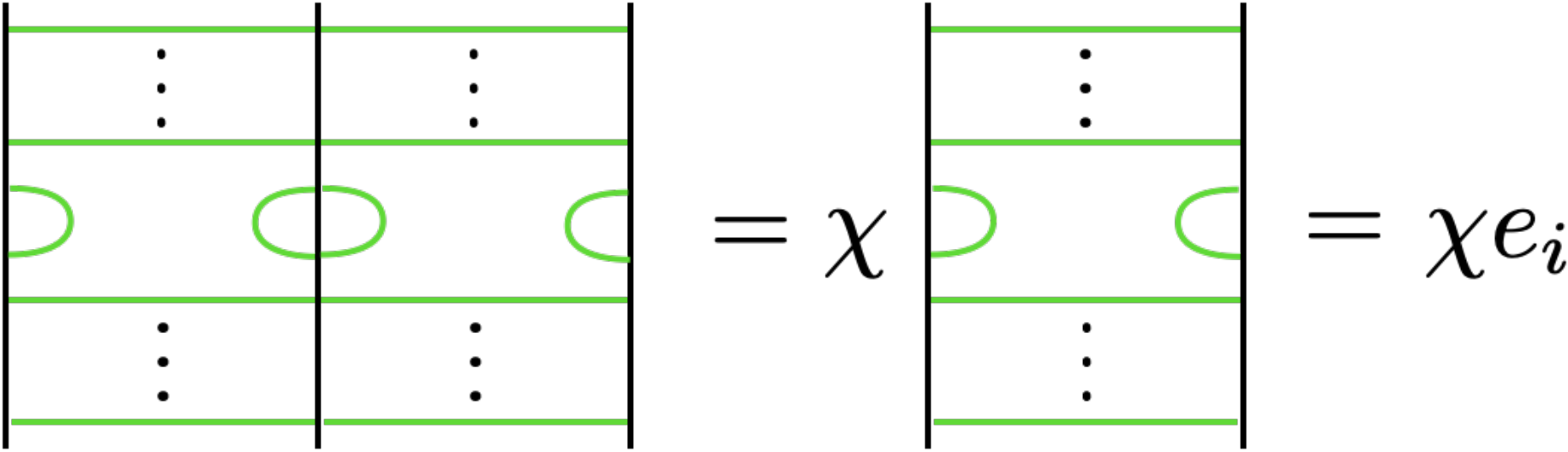}      
    \end{matrix}, \\
    e_ie_{i+1}e_i&=\begin{matrix}
      \includegraphics[scale=.3]{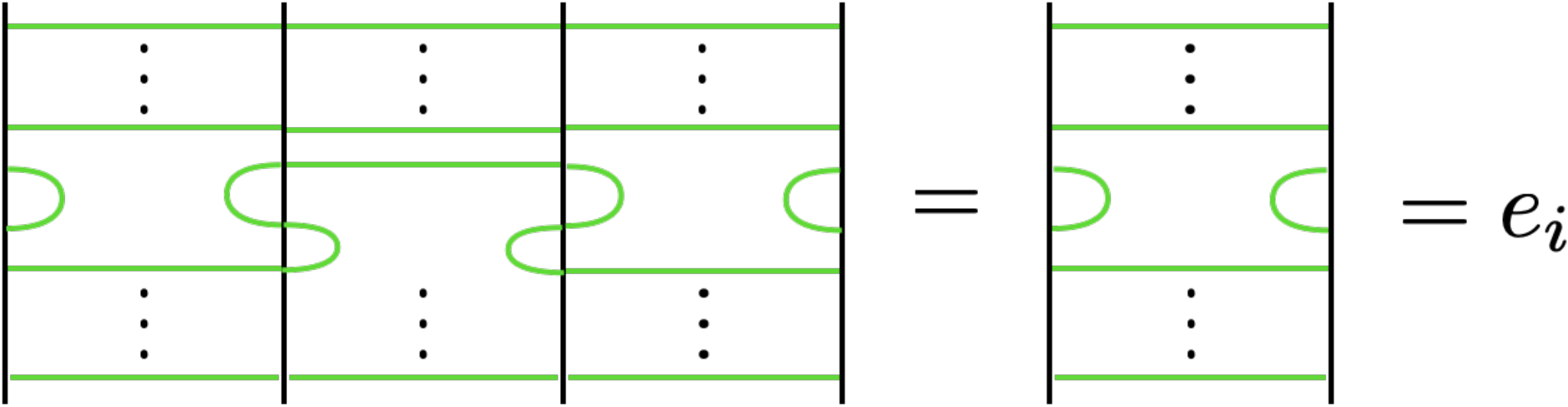}      
    \end{matrix}, \\
    e_ie_j&=\begin{matrix}
      \includegraphics[scale=.3]{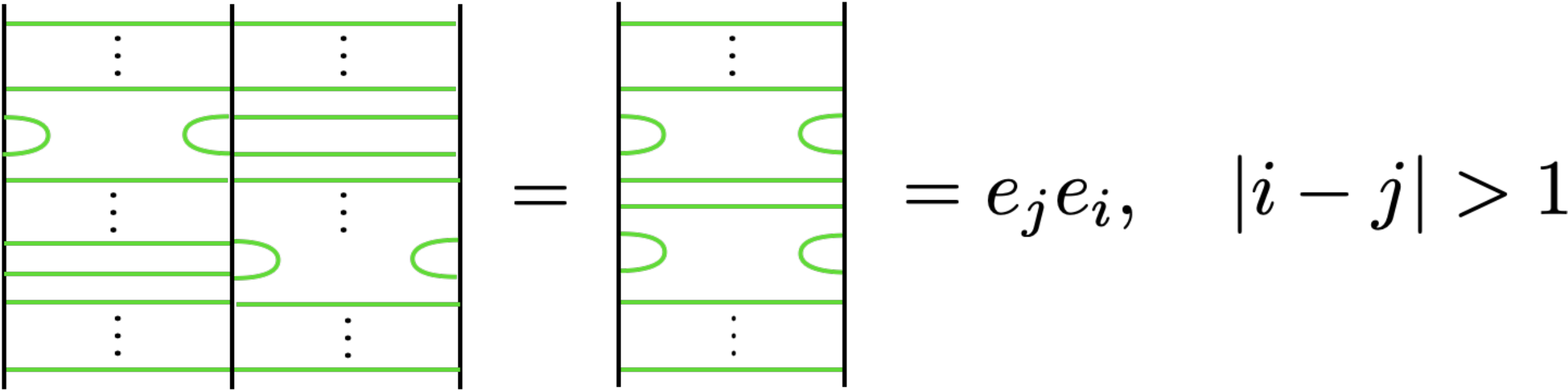}      
    \end{matrix},
\end{align}
and also $e_ie_{i-1}e_i=e_i$, which can be verified through the upside-down mirror of the second identity.
This motivates the another more abstract definition of TL algebra:
\begin{definition}[alternative definition of Temperley-Lieb algebra]
  The Temperley-Lieb algebra $\text{TL}_n$ is the algebra generated by a unit \text{id} and a set of $n-1$ generators $e_i$ with $i=1,\cdots,{n-1}$ which satisfies the following relations:
  \begin{align}
  \label{eq:TL_generator}
    e_i^2 = \chi e_i, \quad e_i e_{i\pm 1}e_i = e_i, \quad e_i e_j = e_j e_i\quad \text{if} |i-j|>1.
  \end{align}
\end{definition}
The proof that the algebra generated from this definition is isomorphic to the diagrammatic one (Definition~\ref{def:TL}) can be found in standard references of the subject.

We conclude this subsection by introducing a trace function on $\text{TL}_n$.
\begin{definition}
  For a positive integer $n$, $\text{Tr}_{\text{TL}_n}: \text{TL}_n\to \mathbb{C}$ is a cyclic linear function on $\text{TL}_n$.
  It is defined diagrammatically on an $n$-diagram by closing the strands on opposing marked points on two lines and assign a power of $\chi$ for each closed loop obtained this way.
  $\text{Tr}_{\text{TL}_n}$ is cyclic from definition, and it is extended to the full algebra using linearity.
\end{definition}
For example,
\begin{equation}
  \begin{matrix}
      \includegraphics[scale=.35]{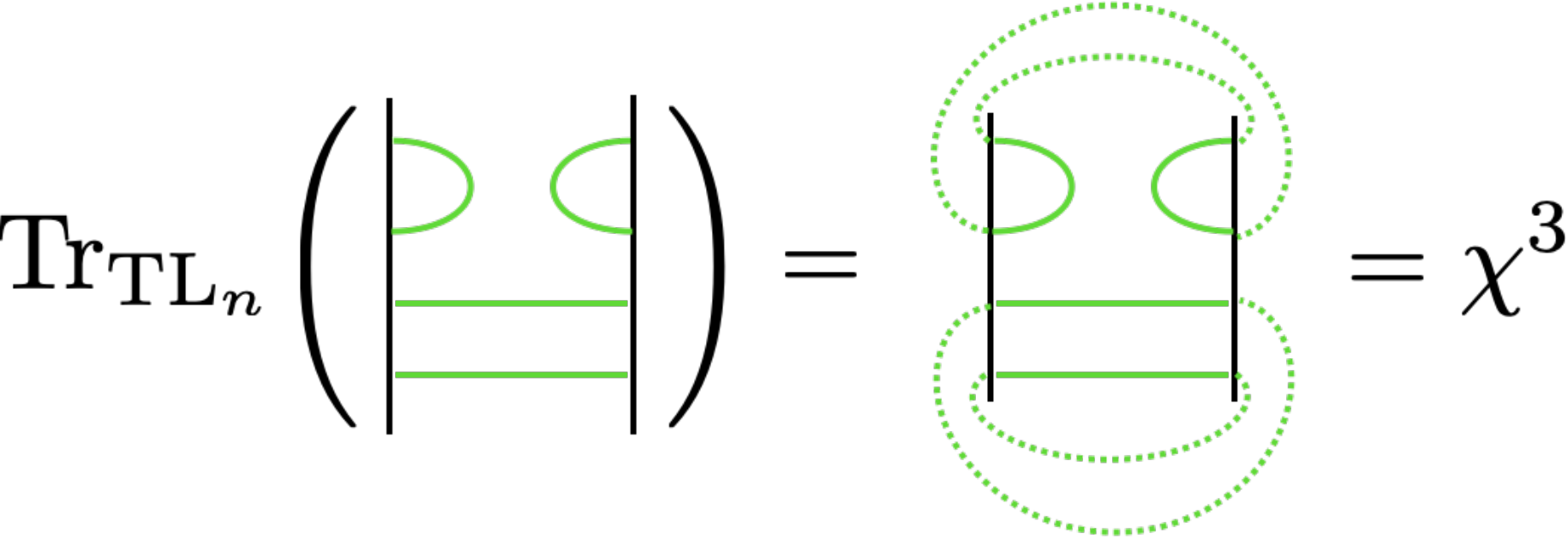}
  \end{matrix}.
\end{equation}
We also define a normalized version of this trace, denoted by $\text{Tr}_{\widehat{\text{TL}}}=\chi^{-n}\textup{Tr}_{\textup{TL}_n}$. 
This normailzed trace is consistent with respect to all $n$ and satisfies the following properties:
\begin{align}
\label{eq:tr_char}
    \text{Tr}_{\widehat{\text{TL}}}(\textup{id}) = 1, \quad 
    \text{Tr}_{\widehat{\text{TL}}}(h) = \chi \text{Tr}_{\widehat{\text{TL}}}(he_{n-1})
\end{align}
where $h\in \text{TL}_{n-1} \subset \text{TL}_n$ is an element in the subalgebra of $\text{TL}_n$ consisting of elements such that the $n$-th marked point on both sides of its diagram are connected via a strand. 
This fact can be illustrated using the diagrams below. 
\begin{equation}
    \begin{matrix}
        \includegraphics[scale=.35]{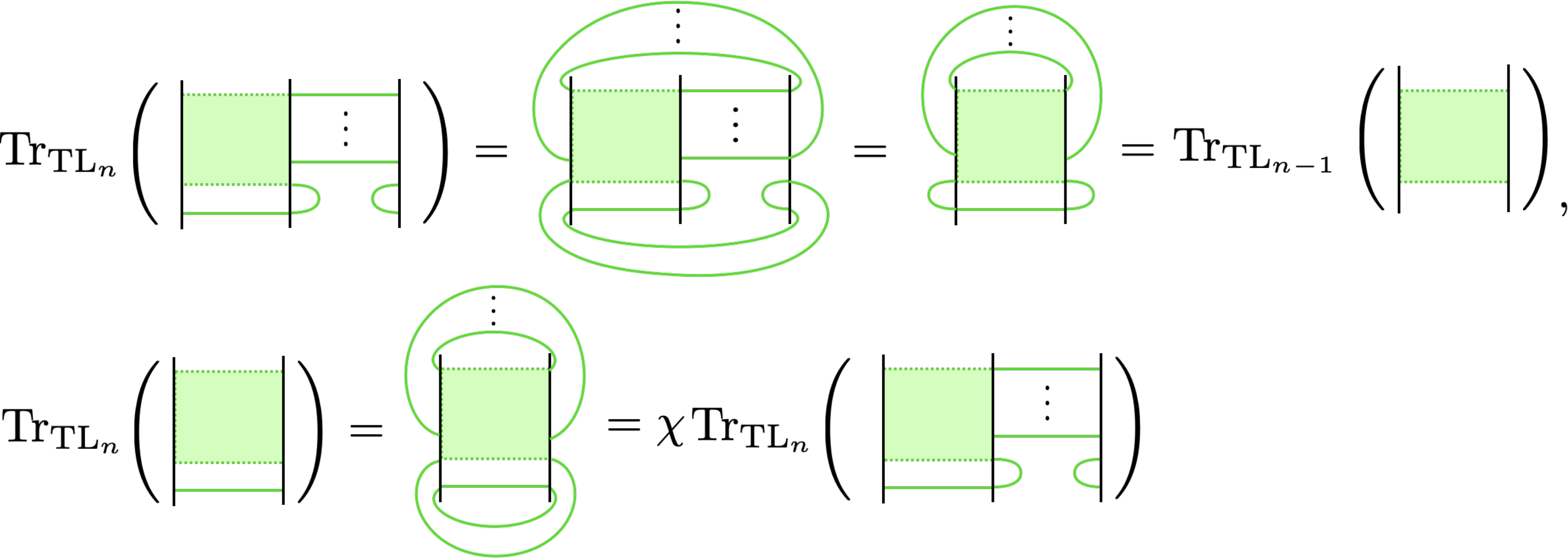}
    \end{matrix}
\end{equation}
where we used colored blocks to indicate arbitrary connections in the region.
After putting back the normalization factors one recovers \Eqref{eq:tr_char}.
In fact the property \Eqref{eq:tr_char} completely characterizes $\text{Tr}_{\widehat{\text{TL}}}$, as shown in the following lemma:
\begin{lemma}
\label{lem:trace}
    $\textup{Tr}_{\widehat{\textup{TL}}}$ is the unique linear function on $\textup{TL}_n$ that is cyclic $(\textup{Tr}(ab) = \textup{Tr}(ba))$ and such that \Eqref{eq:tr_char} is true.
\end{lemma}
\begin{proof}
Every element $h\in \text{TL}_{n}$ can be written as a product of finitely many generators $h = e_{i_1}\cdots e_{i_k}$ where $i_k\in \{1,\cdots,n-1\}$.
The $e_{i_j}$'s in the string can be further made unique by using \Eqref{eq:TL_generator} to permute the list and eliminate any duplicates.
To evaluate $\textup{Tr}_{\widehat{\textup{TL}}}(h)=\textup{Tr}_{\widehat{\textup{TL}}}(e_{i_1}\cdots e_{i_k})$, we use cyclicity of the trace to cycle the possible $e_{n-1}$ to the end of the list, and we pick up a factor of $\chi$ if there is one.
Now since the new list of generators does not contain $e_{n-1}$, the new element $h'$ obtained by multiplying the new list of generators is in the subalgebra $\text{TL}_{n-1}\subset TL_{n}$. Using \Eqref{eq:tr_char} we see that $\text{Tr}_{\widehat{\text{TL}}}(h)=\chi^b\text{Tr}_{\widehat{\text{TL}}}(h')$ where $b=0$ if there is no $e_{n-1}$ in the generating string of $h$ and $b=-1$ if there is one.
Now by repeated iteration of this whole process $n-1$ times, one can eliminate all the generators of $h$ and obtain $\text{Tr}_{\widehat{\text{TL}}}(h) = \chi^{-k}\text{Tr}_{\widehat{\text{TL}}}(\text{id})=\chi^{-k}$ where $k$ is the length of the generating list $h = e_{i_1}\cdots e_{i_k}$ with non-duplicating generators. Thus the condition \Eqref{eq:TL_generator} uniquely fixes the action of $\text{Tr}_{\widehat{\text{TL}}}$ on every element in $\text{TL}_n$.
\end{proof}

\subsection{The standard module}
To define an algebraic module, one must first form the vector space for the algebra to act on.
In the standard module of Temperley-Lieb algebra such vector space is spanned by a collection of diagrams called \textit{link states} (or \textit{cup diagrams} in Westbury \cite{Westbury1995}).

An $n$-diagram can be cut in half across the middle line and expose a number of ``defects'', each corresponding to a connection between a point in the left and another point in the right.
The number of defects, denoted $l$, is always even if $n$ is even and odd if $n$ is odd, which we keep track of using an integer $k\in\mathbb{Z}_{\ge0}$, by saying that $l=2k$ when $n$ is even and $l=2k+1$ when $n$ is odd.
We will call a half diagram obtained this way an $(n,k)$\emph{-link state}.
\footnote{Note that we label link states by $k$, (one half of) the number of \textit{defects}, as opposed to the common practice of labeling by the number of \textit{closed arcs} $p=\floor{n/2}-k$. The main advantage for this notation change is to make the various formulae regarding to different $k$ sector in the main text cleaner.}
A link state will be drawn with the lines on the left and any possible defects facing right, as shown in \figref{fig:linkstates}.
We will refer to the set of all $(n,k)$-link states as $\mathcal{B}^{(n)}_{k} $.
\begin{figure}[h]
  \centering
  \includegraphics[scale=.35]{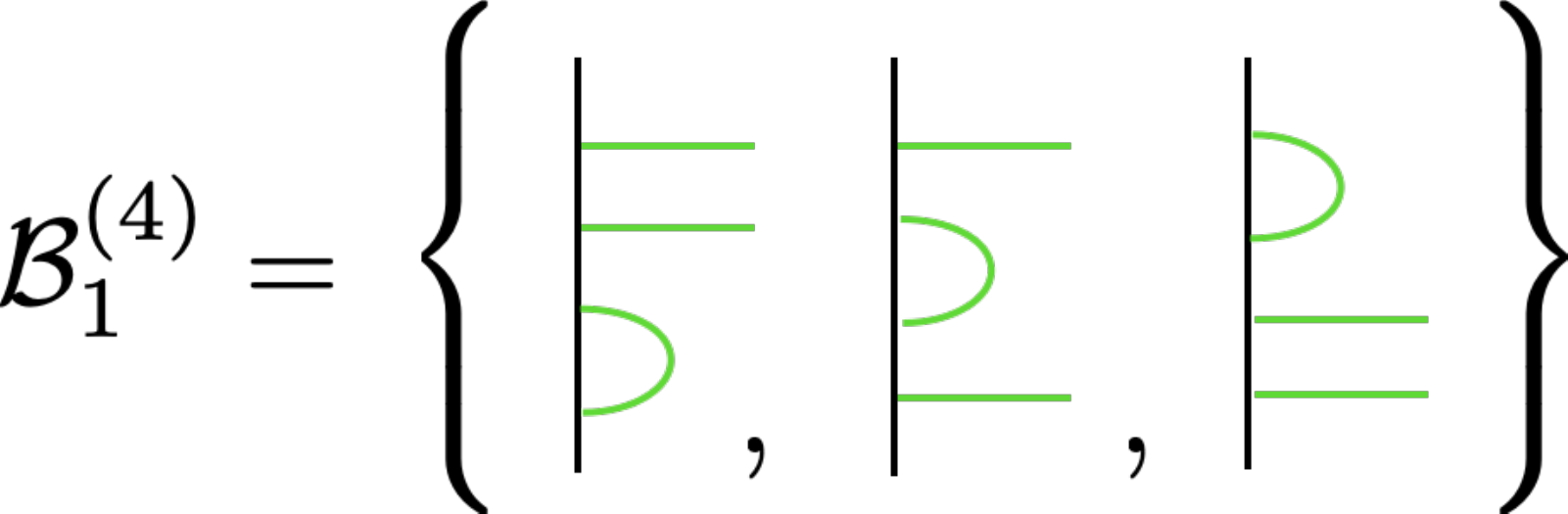}
  \caption{All the possible $(4,1)$-link states.}
  \label{fig:linkstates}
\end{figure}

\begin{lemma}
\label{lem:module_dim}
  For integer $n$ and $k\le \lfloor n/2 \rfloor$, the number of distinct $(n,k)$-link states is equal to the number of standard Young tableaux
  of shape $(\ceil{n/2}+k,\floor{n/2}-k)$.
  \begin{equation}
     |\mathcal{B}^{(n)}_k| = \#\textup{SYT}\left(\ceil*{\frac{n}{2}}+k,\floor*{\frac{n}{2}}-k\right)
  \end{equation}
\end{lemma}
\begin{proof}
  Note that there is a bijection between a link state and a \emph{lattice word} consisting of left and right brackets $\{[,]\}$ such that every prefix has more opening brackets than closing brackets.
  For example, $[\,[\,[\,]\,[\,]\,]\,]\,[$ is a lattice word but $[\,]\,[\,]\,]\,[\,[\,]\,]$ is not since the prefix $[\,]\,[\,]\,]$ has more $]$'s than $[$'s.
  To put it in another way, there can never be a closing bracket without a previous matched opening bracket, but a standalone opening bracket is allowed.
  To see that this is true, simply assign an opening-closing pair to a closed strand in the link state and any standalone opening brackets to defects, e.g.
  \begin{equation}
    \begin{matrix}
       \includegraphics[scale=.4]{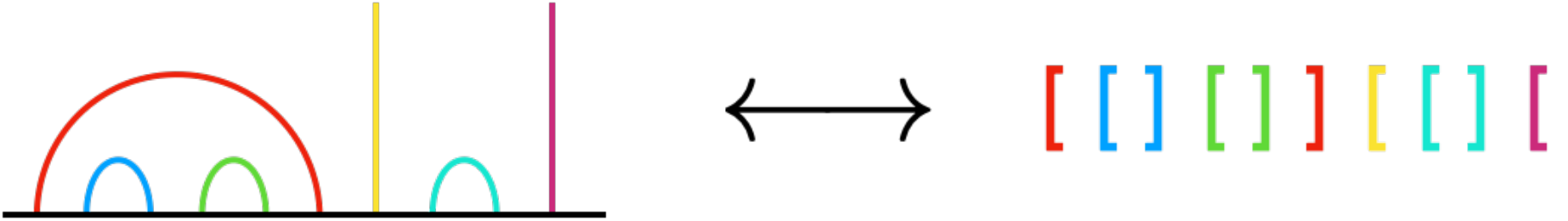}   
    \end{matrix}
  \end{equation}
  where we have color coded the brackets to the matching strands.
  This correspondence works in either way so it is really a bijection.
  Now it is a standard result that the number of lattice words are counted by standard Young tableaux:
  Given a standard Young tableau of shape $(p,q)$, construct the following word $(s_1,\cdots,s_{p+q})$ via
  \begin{equation}
    s_i =
    \begin{cases}
      ~[, \quad \text{if $i$ is in first row}, \\
      ~], \quad \text{if $i$ is in second row}.
    \end{cases}
  \end{equation}
  Then at any given point $i$ there can never be more closing brackets than opening brackets since the tableaux is standard.
  The number of total brackets are $p+q $ whereas the number of standalone opening brackets are $p-q$ from which one immediately sees that $p=\ceil{n/2}+k$ and $q=\floor{n/2}-k$.
\end{proof}

For a $n$-diagram $h\in \text{TL}_n$, we define a left action of $h$ on a link state $v$ in a similar fashion by concatenating $h$ and $v$ from left to right, removing any disconnected open strands and assign a factor of $\chi$ for each closed loop.
We illustrate this definition with an example:
\begin{equation}
  \label{eq:LS_action}
  \begin{matrix}
    \includegraphics[scale=.3]{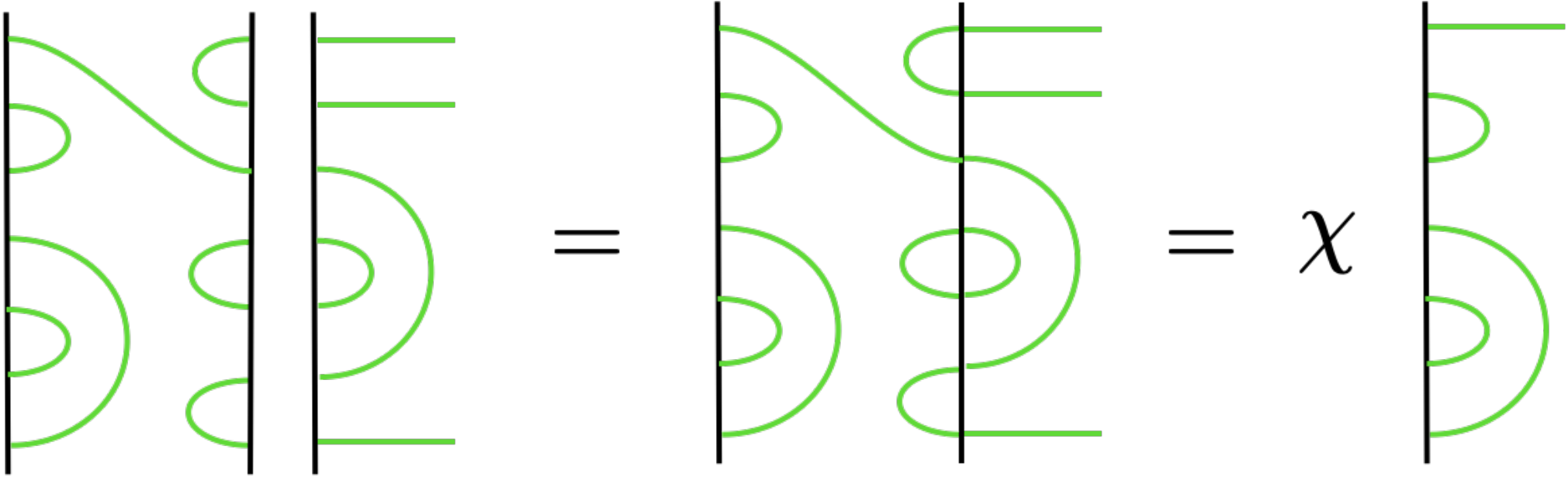}
  \end{matrix}
\end{equation}
This action naturally defines a algebraic module (representation) on the vector space spanned by link states.
We will refer to such module as the \textit{link module}.
\begin{definition}[link module]
  For given integer $n$, the \textit{link module}, denoted by $\mathcal{M}^{(n)}$, is a left $\text{TL}_n$-module on the collective vector space spanned by all link states whose action is given by the diagrammatic action above and is extended to the whole vector space through linearity.
\end{definition}

Note that the number of defects is not necessarily conserved under the action. (\Eqref{eq:LS_action} is an example where it decreases by two)
However, it can only close defects in pairs, so $k$ is always non-increasing under this action.
Acknowledging this fact, one can identify $\text{TL}_n$-submodules $\mathcal{M}^{(n)}_k\subseteq\mathcal{M}^{(n)}$ spanned by all $(n,k')$-link states with $k'<k$. There is a natural chain of submodules
\begin{equation}
  \mathcal{M}^{(n)}_0 \subset \mathcal{M}^{(n)}_1 \subset \cdots \subset \mathcal{M}^{(n)}_{\floor{n/2}} = \mathcal{M}^{(n)}
\end{equation}
Consecutive quotients on this chain determines quotient $\text{TL}_n$-modules called \emph{standard modules}:
\begin{definition}[The standard module]
  For integer $n$ and $k<\floor{n/2}$, the $(n,k)$-\emph{standard module}, denoted by $\mathcal{V}^{(n)}_k$, is the quotient module
  \begin{equation}
    \mathcal{V}^{(n)}_k = \frac{\mathcal{M}^{(n)}_{k+1}}{\mathcal{M}^{(n)}_k}.
  \end{equation}
\end{definition}
The standard module vector space (which we will also refer to as $\mathcal{V}^{(n)}_k$ by an abuse of notation) is the coset $[\mathcal{M}^{(n)}_{k+1}:\mathcal{M}^{(n)}_{k}]$, which is isomorphic to the vector space spanned by $(n,k)$-link states.
In this regard, we will often talk about the action of $(n,k)$-standard module on $(n,k)$-link states directly, forgetting the fact that it is really the coset that $\mathcal{V}^{(n)}_k$ is acting on.
Diagrammatically, the action of an $n$-diagram $h$ on a $(n,k)$-link state $v$ is given by the same action as in the link module, but with the further requirement that $h$ maps $v$ to zero whenever the number of defects $k$ decreases under the action.

There is a natural inner product $\Braket{\cdot,\cdot}$ on the vector space $\mathcal{V}^{(n)}_k$ defined as follows:
\begin{definition}
  \label{def:inner}
  For $x,y\in\mathcal{B}^{(n)}_k$, the inner product between $x$ and $y$, denoted by $\Braket{x,y}$, is obtained by first flipping $x$ around the vertical axis and identifying it with the vertical border of $y$.
  Then, define $\Braket{x,y}$ by assigning a factor of $\chi$ for every closed loop obtained this way.
  Furthermore, we require $\Braket{x,y}$ to be nonzero only when every defect of $x$ is connected to a defect of $y$.
  This product is extended to the full vector space $\mathcal{V}^{(n)}_k$ by requiring the form being linear in the $y$ and anti-linear in $x$.
\end{definition}
As an example, we have
\begin{align}
  \begin{matrix}
      \includegraphics[scale=.3]{Vk_inner}
  \end{matrix}, 
\end{align}
Note that the link state basis $\mathcal{B}^{(n)}_k$ is \emph{not} orthogonal under this inner product.
However, they are approximate orthogonal at large $\chi$, with corrections $\sim O(1/\chi)$.
This inner product is invariant under the TL action.
\begin{lemma}
  For all $x,y\in \mathcal{V}^{(n)}_k$ and $h\in \text{TL}_n$,
  \begin{equation}
    \Braket{x,Uy} = \Braket{U^\dag x,y},
  \end{equation}
  where the ``adjoint'' $U^\dag$ is obtained by flipping the $n$-diagrams across horizontally and conjugating the complex coefficients.
\end{lemma}
Along with the inner product we define another sesquilinear form $|\cdot~\cdot|:\mathcal{V}^{(n)}_k\times\mathcal{V}^{(n)}_k\to \text{TL}_n$ that takes two $(n,k)$-link states and form a $n$-diagram.
\begin{definition}
  \label{def:outer}
  For $x,y\in\mathcal{B}^{(n)}_k$, denote $|x~y|$ to be the unique $n$-diagram formed by flipping $y$ across the vertical axis and identifying its defects with the defects of $x$. This form is extended to the full vector space in the same way as in Definition~\ref{def:inner}.
\end{definition}
For example,
\begin{equation}
  \begin{matrix}
    \includegraphics[scale=.35]{Vk_outer}.
  \end{matrix}
\end{equation}
The usefulness of Definition~\ref{def:outer} comes from the following lemma.
\begin{lemma}
  If $x,y,z\in \mathcal{V}^{(n)}_k$, then
  \begin{equation}
    \label{eq:outer_inner}
    |x~y|\,z = \Braket{y,z} x
  \end{equation}
\end{lemma}
\begin{proof}
  Without loss of generality we can consider the case where all of $x,y,z$ are $(n,k)$-link states.
  If $|x~y|$ decreases the number of defects in $z$, then the LHS of \Eqref{eq:outer_inner} vanishes. This is consistent with the RHS as since there must exist disconnected strands in $\Braket{y,z}$ when such a pair of defect is closed, forcing $\Braket{y,z}=0$.
  It remains to check the case where no defects are closed. In this case the LHS of \Eqref{eq:outer_inner} will be proportional to $x$.
  The proportionality factor is given by $\chi^\#$, where $\#$ is the number of closed loops in the concatenation, which is the same as the number of closed loops in calculating $\Braket{y,z}$.
\end{proof}

It is sometimes useful to think of $|x~y|$ as the outer product of link states $x y^\dag$, with the adjoint $y^\dag\in \bar{V}^{(n)}_k$ in the dual module obtained by flipping $y$ around its vertical axis.
However such analogy is only true when all of $x,y,z$ lie in the same $(n,k)$-standard module.
When it is not the case, say when $x,y \in V^{(n)}_k$ and $z\in V^{(n)}_{k'}$, then not only $\Braket{y,z}$ is not defined, $|x~y|z$ will not be proportional to $x$.
For example,
\begin{equation}
  \begin{matrix}
      \includegraphics[scale=.35]{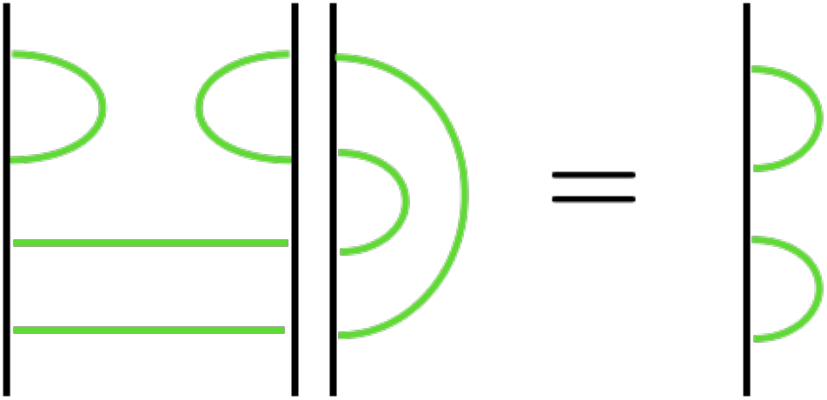}.
  \end{matrix}
\end{equation}

In general, the standard module classify all the finite dimensional irreducible representations of $\text{TL}_n$ for almost all values of $\chi$, except for a discrete set of points when $\chi < 2$.
The key to proving this proposition is by studying the degeneracy of the inner product $\Braket{\cdot,\cdot}$ through its \emph{Gram matrices} $G^{(n)}_k$.
The idea is very similar to the process of determining the reducibility of Verma modules of Virasoro algebra in 2d CFT: $G^{(n)}_k$ being non-singular implies the irreducibility of $\mathcal{V}^{(n)}_k$.
When $G^{(n)}_k$ is singular (i.e. $\det G^{(n)}_k=0$), there will be additional null states that one needs to discard to obtain a irreducible representation.

To end this short review we quote an important theorem, first due to Jones \cite{Jones1983} and studied in more detail by Westbury \cite{Westbury1995}.
\begin{theorem}[Jones]
  When $\chi \neq 2\cos(m\pi/n)$ for some integer $m\ge 3$, the algebra $\text{TL}_n$ is semisimple and the standard modules $\mathcal{V}^{(n)}_k$  form a complete set of finite dimensional irreducible non-isomorphic representations of $\text{TL}_n$.
\end{theorem}
Due to the length constraint we will not present a proof for this theorem. Interested reader can refer to the various literature (e.g. Refs.~\cite{Westbury1995,Ridout:2012gg}) on this subject.
Throughout this note we always work with integer $\chi\ge 2$, and we frequently consider the case $\chi\gg 1$. In this regime the standard modules are always irreducible.

\section{Finite $\chi$ corrections}
\label{sec:finite_chi}
In this appendix we will study the effect of finite bond dimension $\chi$.
As opposed to the corrections from finite external bond dimensions $\chi_{A,B,C}$, the effect of finite $\chi$ acts only on eigenvalues within each $k$ sectors by shifting them by a small amount.
In the following we will give generating functions pertain to these eigenvalue shifts.
Analytically continuing the coefficient of the relevant generating function gives the leading correction the the reflected spectrum.
We give explicit formulae of such analytic continuation for $k=0,1$.
The correction to these two sectors completely characterizes the leading order corrections to the reflected entropy in the finite $\chi$ limit, which matches well with our numerics (\secref{sub:num}).
While the detailed effect of sector mixing is interesting on its own in that there seems to be an additional hierarchy structure from our TL numerics, its effect on the reflected entropy is to the order $O(\chi^{-2})$ which is outside the main interest of this note.

\subsection{Corrections to orthogonality condition}
  We will find the first order $\chi$ corrections to \Eqref{eq:leadinglambda}.
  Denote the related generating function
  \footnote{We will suppress the $q_{A,B}$ dependence in this section unless necessary to make our notation more clear.}
  \begin{equation}
    G(r,z) = \sum_{m\in 2\mathbb{Z}_+,k\in \mathbb{Z}_+}X^{(m)}_{k}(q_A,q_B,\chi)z^{m/2}r^{k} ,\quad X^{(m)}_k(q_A,q_B,\chi) = \sum_{x,y\in \mathcal{B}^{(m)}_k} f_{q_A}(x)f_{q_B}(y)\Braket{x,y}
  \end{equation}
  We write the function $X^{(m)}_k$ as a sum of diagrams:
  \begin{equation}
    \begin{matrix}
       \includegraphics[scale=.33]{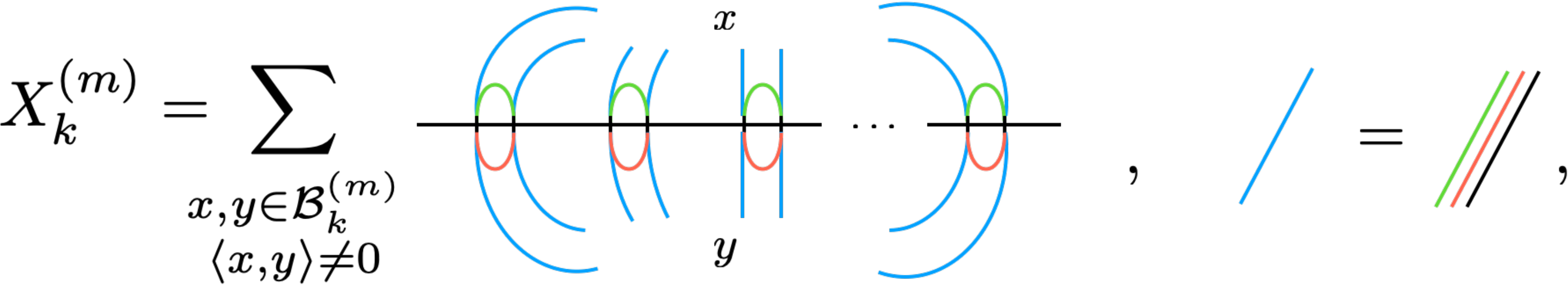}   
    \end{matrix}
  \end{equation}
  where we assign each closed red loop a factor of $q_A$, each closed green loop a factor of $q_B$ and each closed black loop a factor of $\chi$.
  Note that we work in the limit of large $\chi$ and $q_A,q_B\sim O(1)$.
  The contraction pattern of blue strands is determined by the link states $x$ (on top) and $y$ (on bottom) where we exclude all contractions  in which $\Braket{x,y}=0$.
  The dominant contribution comes from $x=y$, reflecting the fact that the link states form a approximately orthogonal basis for $\mathcal{V}^{(m)}_k$.
  We now try to determine the first order correction to this result.
  Begin by splitting $G(r,z)$ in the following fashion
  \begin{equation}
    G(r,z) = G_0(r,z) + G_1(r,z),
  \end{equation}
  where $G_0(r,z)$ contains all the leading diagrams (that is with $x=y$).
  It is related to the link state generating function \Eqref{eq:gen_fun} by $G_0(r,z) = G(q_{AB},r/\chi,z\chi)$.
  $G_1(r,z)$ contains all the next order corrections. 
  Also denote $H(r,z)$ to be the generating function for 1PI diagrams and split it in a similar fashion
  \begin{equation}
    H(r,z) = H_0(r,z) + H_1(r,z)
  \end{equation}
  $G(r,z)$ and $H(r,z)$ are related by
  \begin{equation}
    \label{eq:Gconsistent_dig}
    \begin{matrix}
      \includegraphics[scale=.33]{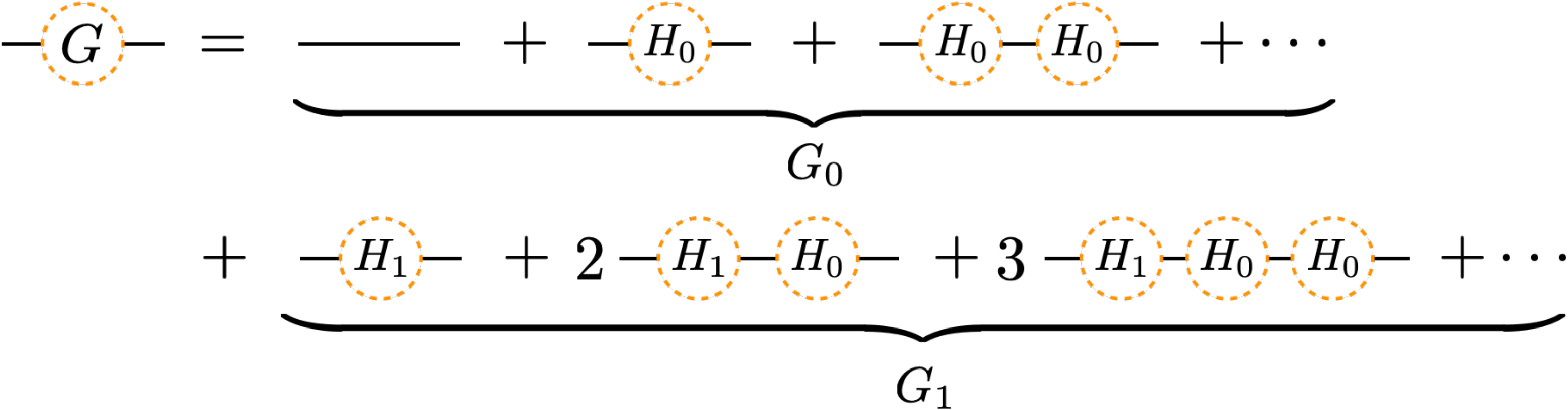}   
    \end{matrix}
  \end{equation}
  The weights in the sum of $G_1$ counts the number of permutations for the 1PI diagrams.
  Algebraically, \Eqref{eq:Gconsistent_dig} can be written as
  \begin{align}
    \label{eq:Gconsistent}
    \begin{split}
       G_0(r,z) &= \sum^\infty_{n=0}H^n_0(r,z) = \frac{1}{1-H_0(r,z)}, \\
    G_1(r,z) &= H_1(r,z)\sum^\infty_{n=0}(n+1)H_0^n(r,z) = G^2_0(r,z)H_1(r,z)   
    \end{split}
  \end{align}
  An important seed for the generating function $G$ is the special case $k=0$, where the link states are non-crossing partitions $x,y\in NC_m$. 
  Schematically the SD equation for $H_1$ and $k=0$ is
  \begin{equation}
    \begin{matrix}
       \includegraphics[scale=.33]{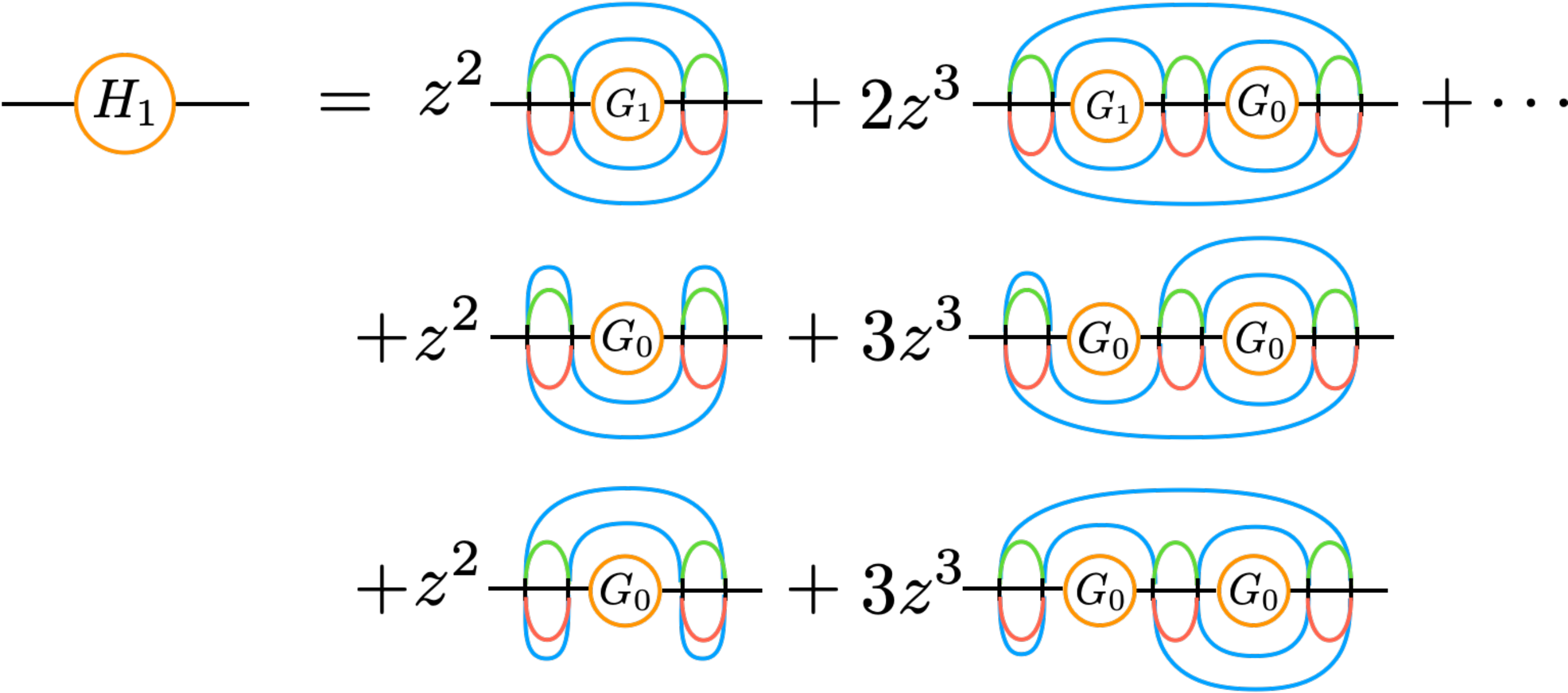}   
    \end{matrix}
  \end{equation}
    For the diagrammatic notations used here: dashed circles indicate generating functions such as $G(r,z)$ and $H(r,z)$; and solid line circles indicate the generating functions for the $k=0$ sector, i.e. $G(0,z)$ and $H(0,z)$.
  In the last two rows the blue connections feature non-crossing permutations $(x,y)$ of the form $(\tau_n, \tau_n\sigma)$ or $(\tau_n\sigma,\tau_n)$, where $\sigma$ is a simple transposition (permutation that swaps two elements and keeping all others invariant).
  The number of these permutations $\sigma\tau_n$ are $n(n-1)/2$.
  We get the following equation
  \begin{align}
    \label{eq:SD_subleading}
    \begin{split}
      H_1(0,z) &= q_{AB}\chi^2 z^2 G_1(0,z)\sum_{n=0}^\infty(n+1)(\chi z G_0(0,z))^n \\
               &\quad + q_{AB}(q_A+q_B)\sum^\infty_{n=0}\frac{n(n-1)}{2}z^n(\chi G_0(0,z))^{n-1} \\
    &= q_{AB}\frac{\chi^2 z^2 G_1(0,z)}{(1-\chi z G_0(0,z))^2} + q_{AB}(q_A+q_B)\frac{\chi z^2G_0(0,z)}{(1-\chi zG_0(0,z))^3}   
    \end{split}
  \end{align}
  This equation, together with \Eqref{eq:Gconsistent} yields a linear equation for $G_1(0,z)$ in terms of $G_0(0,z)$.
  We get
  \begin{equation}
    \label{eq:G1_0}
    G_1(0,z) = \frac{q_{AB}(q_A+q_B)\chi z^2G_0^3(0,z)}{(\chi zG_0(0,z)-1)\left(q_{AB}\chi^2z^2G^2_0(0,z)-(\chi zG_0(0,z)-1)^2\right)}
  \end{equation}
  This generating function characterizes the correction to the single pole $\lambda_{M^{(m)}_0}$ of the $k=0$ sector.
  
  For the $k>0$ sectors we must include in our sum connections with open strands, but only such that each open strand on top is connected to an open strand in bottom, as required by the condition $\Braket{x,y} \neq 0$.
  To better organize the SD equation we define an auxiliary matrix $\mathbb{F}_0(z)$
  \begin{equation}
    \begin{matrix}
      \includegraphics[scale=.35]{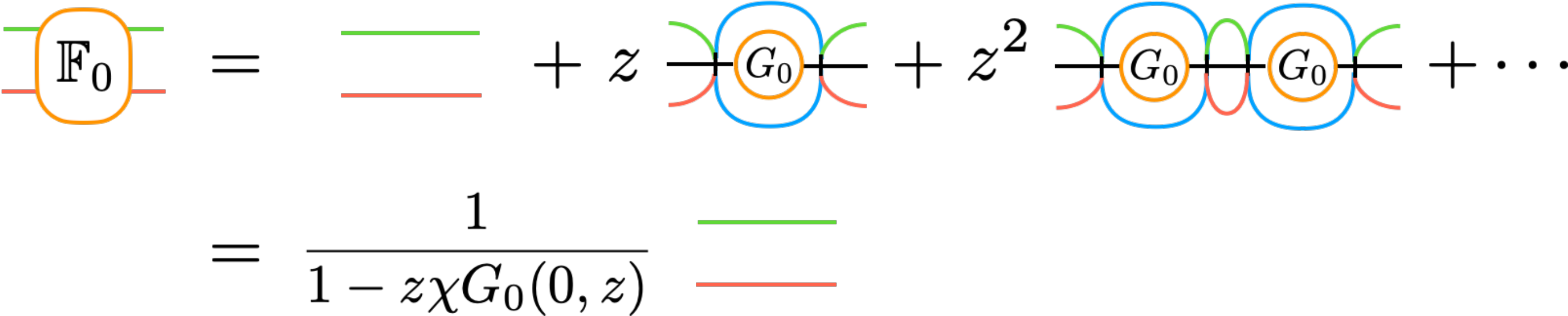} 
    \end{matrix}
  \end{equation}
  $\mathbb{F}_0(z)$ can be thought of as an stripped down version of $H_0(z)$ in the sense of the following equality
  \begin{equation}
    \label{eq:Fconsistent}
    \begin{matrix}
      \includegraphics[scale=.35]{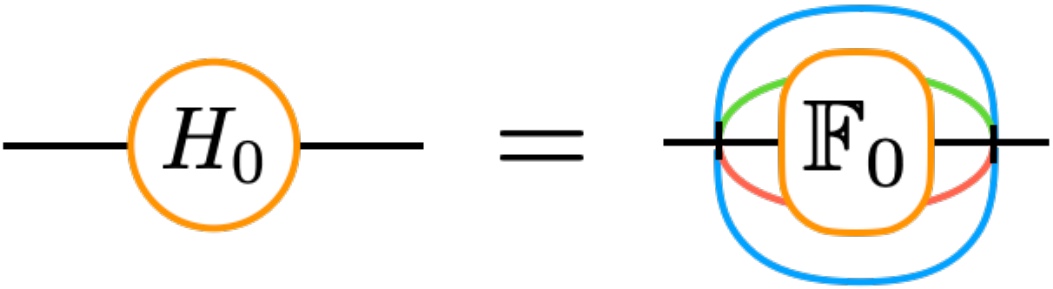} 
    \end{matrix}
  \end{equation}
  Similarly we define the matrix $\mathbb{F}_1(z)$ that is comprised of the next order diagrams
  \begin{equation}
    \begin{matrix}
      \includegraphics[scale=.35]{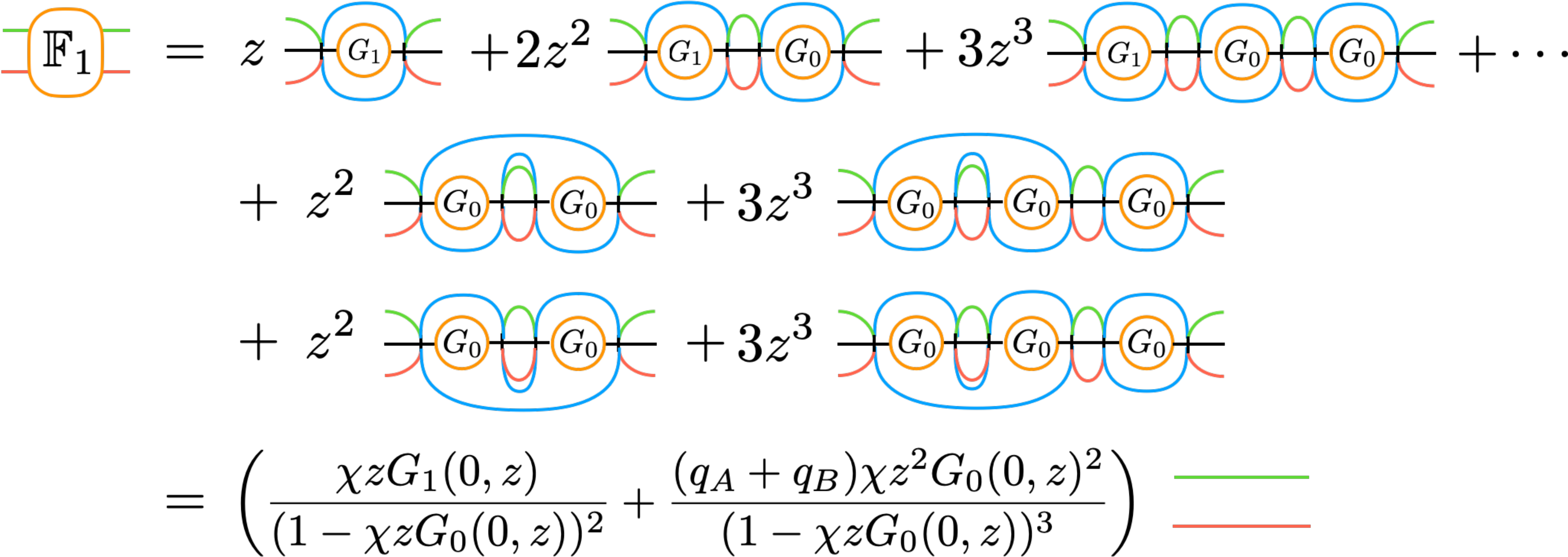}      
    \end{matrix}
  \end{equation}
  Note that despite the similarity, $\mathbb{F}_1(z)$ does not satisfy a simple relation to $H_1(z)$ like \Eqref{eq:Fconsistent}.
  The SD equation for the connected generating function $H(r,z)$ can then be formulated using these matrices
  \begin{equation}
    \begin{matrix}
             \includegraphics[scale=.35]{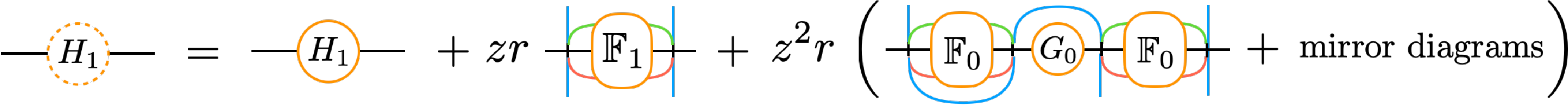}   
    \end{matrix}
  \end{equation}
  As indicated there are three additional diagrams consists of vertical and horizontal reflections of the ``zigzag'' diagram.
  Since each ``zigzag'' decrease the power of $\chi$ by one, diagrams that with more than one ``zigzags'' only contribute at higher orders.
  Algebraically this equation is
  \begin{align}
    \begin{split}
      H_1(r,z) = H_1(0,z) + zr F_1(z) + 2z^2r(q_A+q_B)F_0(z)^2G_0(z)
    \end{split}  
  \end{align}
  where the non-bold version of $\mathbb{F}_i$ denote the common prefactor of the respective matrix.
  Together with \Eqref{eq:Gconsistent} this determines the form of $G_1(r,z)$ (since both $G_1(0,z)$ and $G_0(r,z)$ are known):
  \begin{align}
      G_1(r,z) = G_0(r,z)^2&\left( \frac{G_1(0,z)}{G_0(0,z)^2} + \frac{\chi z^2r G_1(0,z)}{(1-\chi z G_0(0,z))^2}+\frac{(q_A+q_B)\chi z^3r G_0(0,z)^2}{(1-\chi z G_0(0,z))^3} \right. \nonumber \\
       &\quad \left. + \frac{2z^2r(q_A+q_B)G_0(0,z)}{(1-\chi zG_0(0,z))^2} \right)
  \end{align}
  The $k=0$ result \Eqref{eq:G1_0} is recovered by taking $r=0$.
  Denote formally the series coefficient of $G_1(r,z)$ by
  \begin{equation}
    G_1(r,z) = \sum_{\mu,k=0} g^{(1)}_{\mu,k}(q_A,q_B) z^\mu r^k
  \end{equation}
  The correction to the matrix eigenvalue $\lambda_{M^{(m)}_k}$ \Eqref{eq:leadinglambda} from orthogonality is then
  \begin{equation}
  \label{eq:corr_dl1}
    \Delta_{1}\lambda_{M^{(m)}_k} = 2(q_{AB})^{-k}\chi^{m/2-k}g_{m/2,k}(q_{AB}^{-1})g^{(1)}_{m/2,k}(q^{-1}_A,q^{-1}_B) + O(\chi^{m-2k-2})
  \end{equation}
  where $g_{m/2,k}(q_{AB}^{-1})$ is the series coefficient of the link state generating function \Eqref{eq:gen_fun}.

\subsection{Corrections from subleading TL diagrams}
We now write down the corrections from including the subleading TL diagrams in  $M^{(m)}_k(q)$.
  We begin by splitting the sum
  \begin{equation}
    M^{(m)}_k(q) = \sum_{h\in NC_m} q^{-\#(h)}\pi^{(m)}_k(D(h)) = (M_0)^{(m)}_k(q) + \Delta M^{(m)}_k(q)
  \end{equation}
  where $(M_0)^{(m)}_k = \sum_{h\in NC_{m,k}}q^{-\#(h)} \pi^{(m)}_k(D(h))$ contains diagrams of the form $\{|x~y|;x,y\in \mathcal{B}^{(m)}_k\}$
  \footnote{Certainly there are diagrams in $M_0$ that contribute at $O(\chi^{-1})$ but their effect is to alter the orthogonality condition \Eqref{eq:leadinglambda}, which we have already dealt with earlier.}
  and $\Delta M^{(m)}_k = \sum_{h\in NC_m\setminus NC_{m,k}}q^{-\#(h)} \pi^{(m)}_k(D(h))$ contains the rest.
  Then we have
  \begin{equation}
    \Delta\left( M^{(m)}_k(q_A)M^{(m)}_k(q_B) \right) \approx \Delta M^{(m)}_k(q_A)(M_0)^{(m)}_k(q_B) +(M_0)^{(m)}_k(q_A)\Delta M^{(m)}_k(q_B)
  \end{equation}
  up to first order corrections.
  Our goal is to apply first order QM perturbation theory on the matrix $\Delta\left( M^{(m)}_k(q_A)M^{(m)}_k(q_B) \right)$.
  Since $M^{(m)}_k(q_A)M^{(m)}_k(q_B)$ is not Hermitian, we need a slightly modified version of usual first order  perturbation theory.
  
\begin{itemize}
\item correction to leading eigenvalue \\
  Denote $v^{(m)}_k(q) = \sum_{x\in \mathcal{B}^{(k)}_m} f_{q}(x)x$, we have $(M_0)^{(m)}_k(q) = q^{-k} \ket{v^{(m)}_k(q)}\bra{v^{(m)}_k(q)}$
  and the right (left) eigenvector of the unperturbed matrix $\ket{v^{(m)}_k(q_A)}\, (\bra{v^{(m)}_k(q_B)})$.
  \footnote{We will use $v$ and $\ket{v}$ interchangeably as a vector in the module $\mathcal{V}^{(m)}_k$ and $\bra{v}$ as the associated vector in the dual module $\mathcal{V}^{\star(m)}_k$ defined from the inner product $\bra{v}\equiv\Braket{\cdot,v}$ to match the notation of QM perturbation theory. We hope the change of notation is not too confusing for the readers.}.
  The first order correction to the leading eigenvalue is
  \begin{align}
    \label{eq:lambda_corr}
    \begin{split}
      \Delta\lambda_{M^{(m)}_k} &= \frac{\Braket{v^{(m)}_k(q_B)|\Delta M^{(m)}_k(q_A)(M_0)^{(m)}_k(q_B) +(M_0)^{(m)}_k(q_A)\Delta M^{(m)}_k(q_B)|v^{(m)}_k(q_A)}}{\Braket{v^{(m)}_k(q_B)|v^{(m)}_k(q_A)}}   \\
           &= q^{-k}_A\Braket{v^{(m)}_k(q_A)|\Delta M^{(m)}_k(q_B)|v^{(m)}_k(q_A)} 
             +q^{-k}_B\Braket{v^{(m)}_k(q_B)|\Delta M^{(m)}_k(q_A)|v^{(m)}_k(q_B)}
    \end{split}
  \end{align}
  Hence we must evaluate the expectation values of the form
  \begin{equation}
    \Braket{v^{(m)}_k(q)|\Delta M^{(m)}_k(q')|v^{(m)}_k(q)} = \sum_{x,y \in \mathcal{B}^{(m)}_k}  f_{q}(x)f_{q}(y)\sum_{h\in \text{NC}_{m}\setminus \text{NC}_{m,k}}q^{\prime -\#(h)}\Braket{x|\pi^{(m)}_k(D(h))|y}
  \end{equation}
  We can think of the sum $\sum_{h\in \text{NC}_{m}\setminus \text{NC}_{m,k}}\Braket{x|\pi^{(m)}_k(D(h))|y}$ as flipping the diagram of $\ket{x}$ horizontally and try to ``fill in'' the spaces between $\ket{x}$ and $\ket{y}$ by TL diagrams associated with $D(h)$.
  Every $D(h)$ contributes to this sum in different powers of $\chi$ and for our purpose here it suffices to find the diagrams that contribute at leading order.
  
  Again an important seed is the $k=0$ sector.
  For $k=0$ the leading elements in the sum $\sum_{h\in \text{NC}_{m,k>0}} q^{-\#(h)}\Braket{x|\pi^{(m)}_k(D(h))|y}$ are elements such that $D(h) = |z~w|$ where $z,w \in \mathcal{B}^{(m)}_2$ in which they can be thought of as the ``cut-open'' versions of $x$ and $y$
  , as shown in \figref{fig:mirrorcut}.
  \begin{figure}
    \centering
    \includegraphics[width=.7\textwidth]{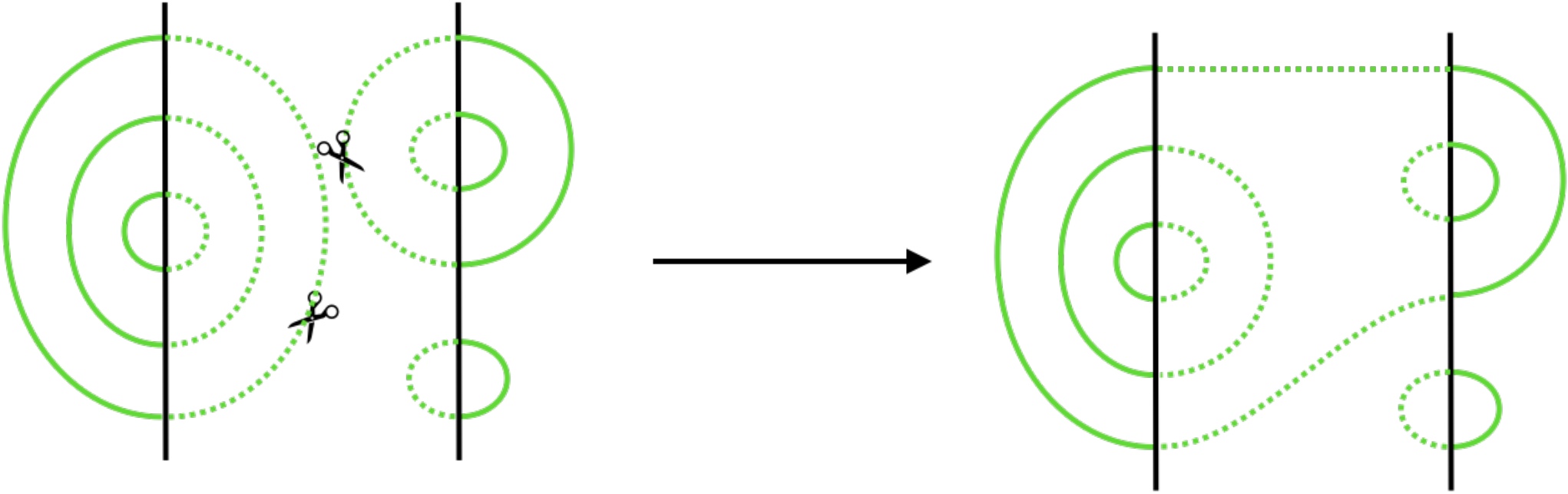}
    \caption{The leading correction diagrams for $\Braket{x|D(h)|y}$ in the $k=0$ sector ($D(h)$ shown as dashed lines in the middle) is constructed by starting with mirroring the $x$ and $y$ link states (shown as solid lines on the left and right sides). Then an outer strand is cut on both mirrored link states before connecting them to form $D(h)$. This construction decreases the contribution of such diagram by one order of $\chi$ compared to $|x~y|$.}
    \label{fig:mirrorcut}
  \end{figure}
  For a link state $x\in \mathcal{B}^{(m)}_k$, denote $c(x)$ to be the number of outer strands it possesses,
  then the number of diagrams we should include to get the leading behavior of $\sum_h\Braket{x|D(h)|y}$ is equal to $c(x)c(y)$.
  We have, up to leading order in $\chi$:
  \begin{equation}
    \sum_{h\in \text{NC}_{m,k>0}} q^{-\#(h)}\Braket{x|\pi^{(m)}_k(D(h))|y} \approx \chi^{2m-1}q f_q(x)f_q(y) c(x)c(y), \quad x,y \in \mathcal{B}^{(m)}_0
  \end{equation}
  $c(x)$ is also equal to the number of connected components $x$ contains, which can be counted using a modified generating function.  
  Recall that $H(q,r=0,z)$ \Eqref{eq:1PIgen} is the 1PI generating function of $k=0$ link states.
  Instead of $G(q,r=0,z) = 1/(1-H(q,0,z))$ we can additionally weigh the sum by the number of 1PI diagrams using
  \begin{equation}
    \label{eq:gen_fun_Y}
    Y(q,z) = H(q,0,z) + 2H^2(q,0,z)+ 3H^3(q,0,z) = \frac{H(q,0,z)}{(1-H(q,0,z))^2}
  \end{equation}
  Denote formally the expansion coefficient of $Y$ by
  \begin{equation}
    Y(q,z) = \sum^\infty_{\mu=0} y_\mu(q)z^\mu
  \end{equation}
  Then one can write $\Braket{v^{(m)}_0(q)|\Delta M^{(m)}_0(q')|v^{(m)}_0(q)}$  (up to leading order) as
  \begin{align}
     \Braket{v^{(m)}_0(q)|\Delta M^{(m)}_0(q')|v^{(m)}_0(q)} \approx q'\,y^2_{m/2}((qq')^{-1})\chi^{m-1} + O(\chi^{m-2})
  \end{align}
  The additional factor of $q'$ in front comes from the fact that the total number of loops decreases by one when performing the surgery shown in \figref{fig:mirrorcut}.
  Finally from \Eqref{eq:lambda_corr} we obtain the first order correction to $\lambda_{M^{(m)}_0}$:
  \begin{align}
    \label{eq:lambda_corr_2}
    \Delta_2 \lambda_{M^{(m)}_0} \simeq\left( q_A+ q_{B}\right)y^2_{m/2}(q_{AB}^{-1})\chi^{m-1} + O(\chi^{m-2})
  \end{align}

  For the $k>0$ sectors one proceeds in a similar fashion but the rule is more complicated.
  As shown in \figref{fig:mirrorcut2}, connecting the open strands dissects $|x~y|$ into different blocks.
  Due to the non-crossing nature of the TL diagram, we must only perform the surgery on outer strands within the same block.
  \begin{figure}
    \centering
    \includegraphics[width=.9\textwidth]{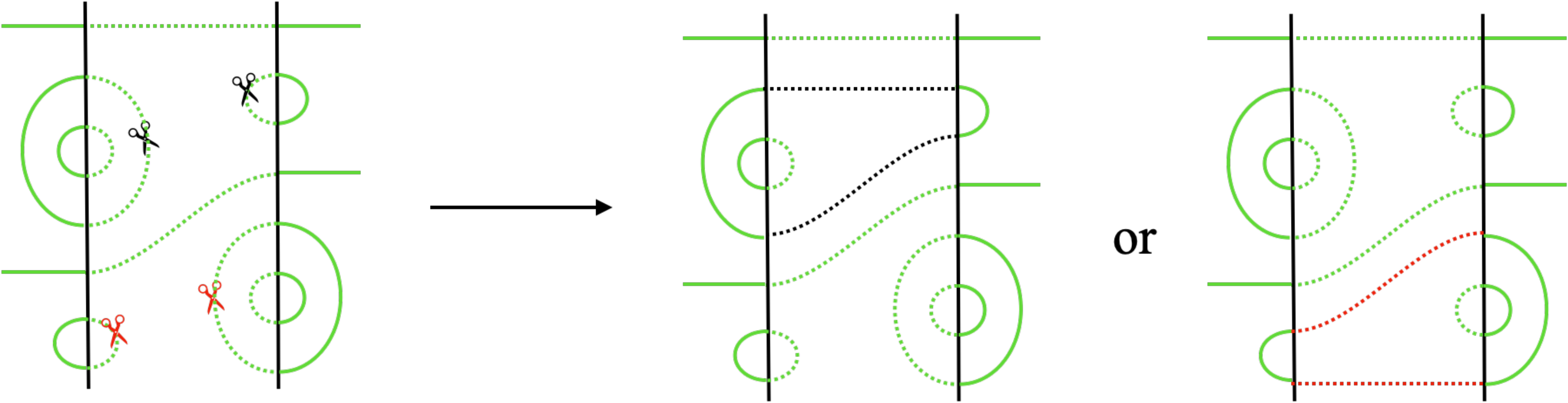}
    \caption{The procedure of constructing leading diagrams for $\Braket{x|D(h)|y}$ for the $k>0$ sector.
      Starting with the diagram $|x~y|$, a similar procedure of cutting and gluing two outer strands are performed.
      However different to the $k=0$ case, the two outer strands being cut must stay inside the same block.
      This is demonstrated using scissors with the same color in the above figure: The surgery on $|x~y|$ can be performed either on the strands with two black scissors (which produces the first diagram on RHS), or strands with two red scissors (which produces  the second diagram on RHS), but performing surgery on a black-red scissor pair is prohibited.}
    \label{fig:mirrorcut2}
  \end{figure}
  To do the required counting here we introduce the two-variable generating functions $\mathcal{G}(z,w)$ and $\mathcal{J}_0(r,z,w)$, defined as
  \begin{equation}
    \begin{matrix}
      \includegraphics[scale=.35]{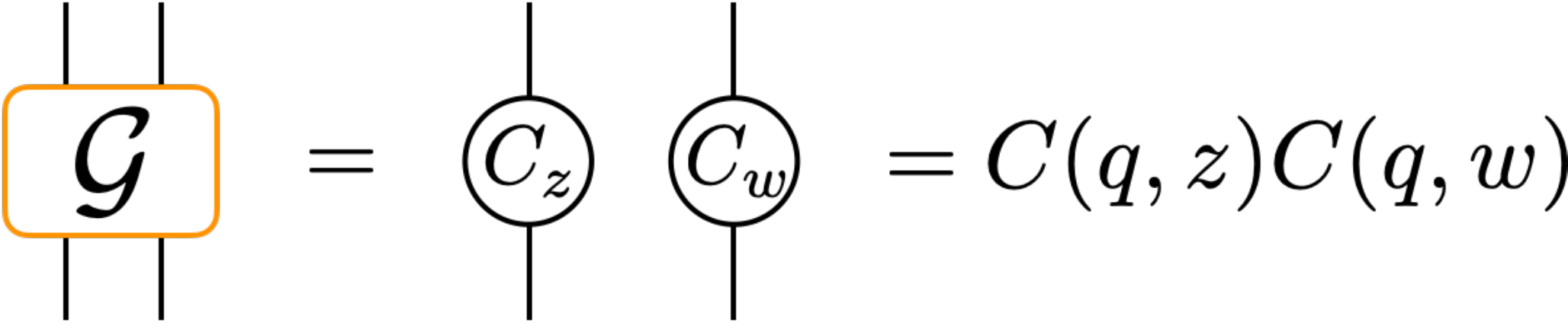}
    \end{matrix}
  \end{equation}
  and
  \begin{equation}
    \begin{matrix}
      \includegraphics[scale=.35]{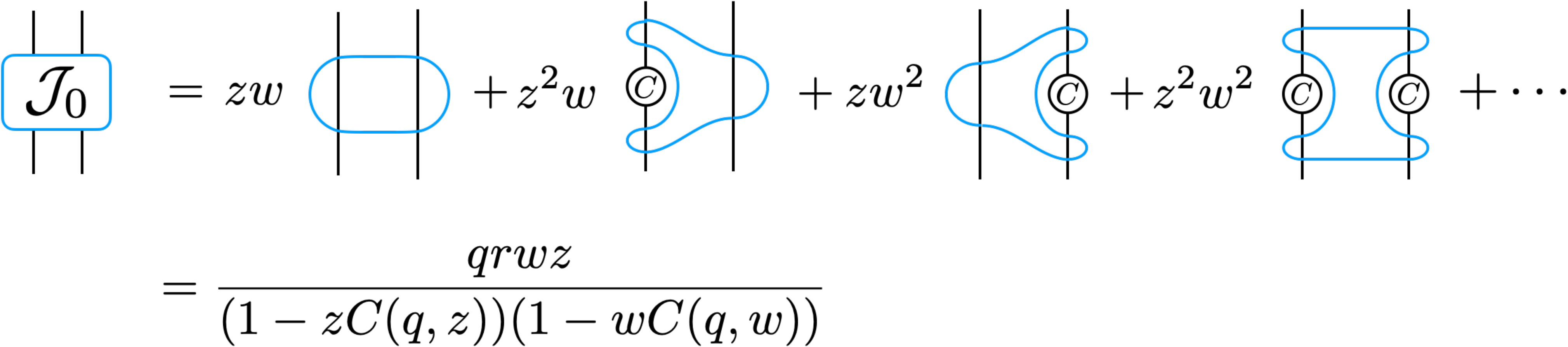}
    \end{matrix}
  \end{equation}
  where the variable $z$ counts the number of TL sites on the left and $w$ counts the number of TL sites on the right.
  One can think of $\mathcal{G}(z,w)$ as a two-point ``propagator'' with $k=0$ and $\mathcal{J}_0(r,z,w)$ as the source which increases $k$ by $1$.
  For example, using these elements we can write down the two-variable generating function for two-sided diagrams
  \begin{equation}
    \begin{matrix}
      \includegraphics[scale=.33]{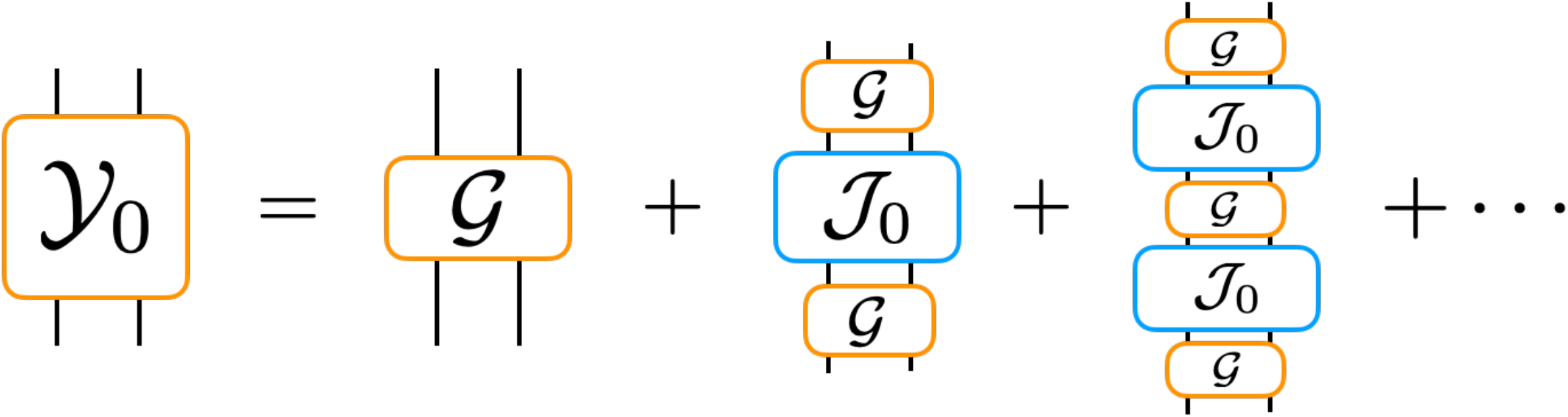}
    \end{matrix}
  \end{equation}
  or in algebraic form 
  \begin{align}
    \begin{split}
        \mathcal{Y}_0(r,z,w) &= \mathcal{G}(z,w)\left(1+\mathcal{G}(z,w)\mathcal{J}_0(r,z,w)+\mathcal{G}(z,w)^2\mathcal{J}_0(r,z,w)^2+\cdots\right) \\
    &=\frac{\mathcal{G}(z,w)}{1-\mathcal{G}(z,w)\mathcal{J}_0(r,z,w)}  
    \end{split}
  \end{align}
  To write down the generating for leading diagrams in $\Braket{x|D(h)|y}$, we introduce additional $c(\cdot)$-weighted source functions $\mathcal{J}_1(z,w)$ and $\mathcal{J}_2(r,z,w)$, defined by
  \begin{equation}
    \begin{matrix}
      \includegraphics[scale=.33]{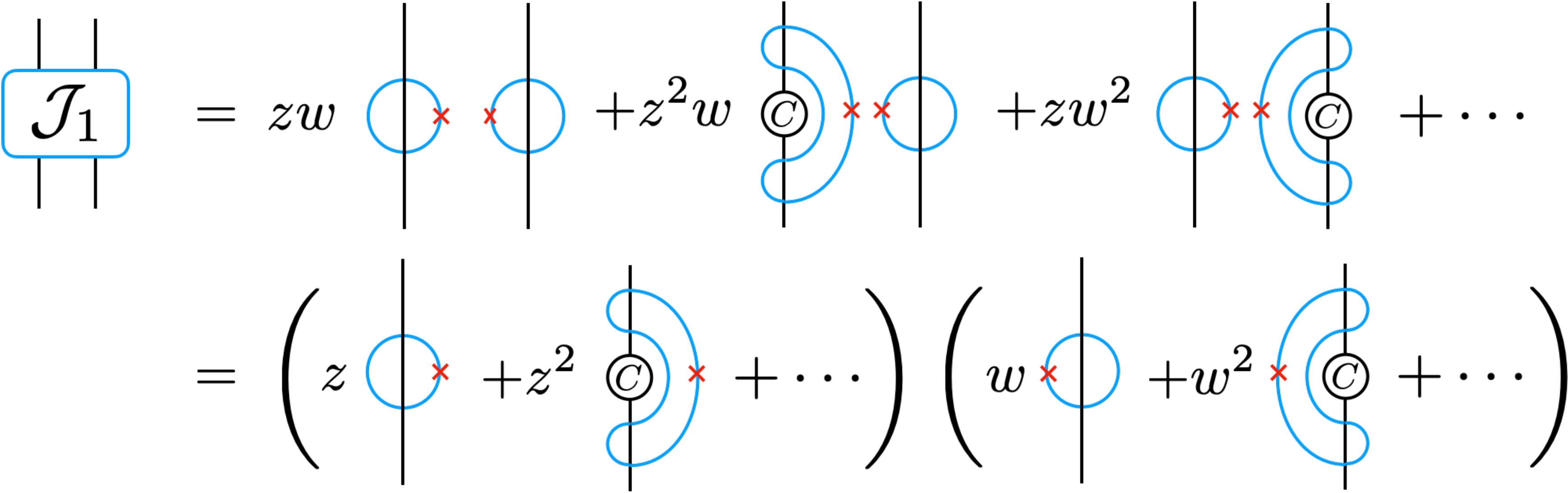}      
    \end{matrix}
  \end{equation}
  \begin{equation}
    \begin{matrix}
      \includegraphics[scale=.33]{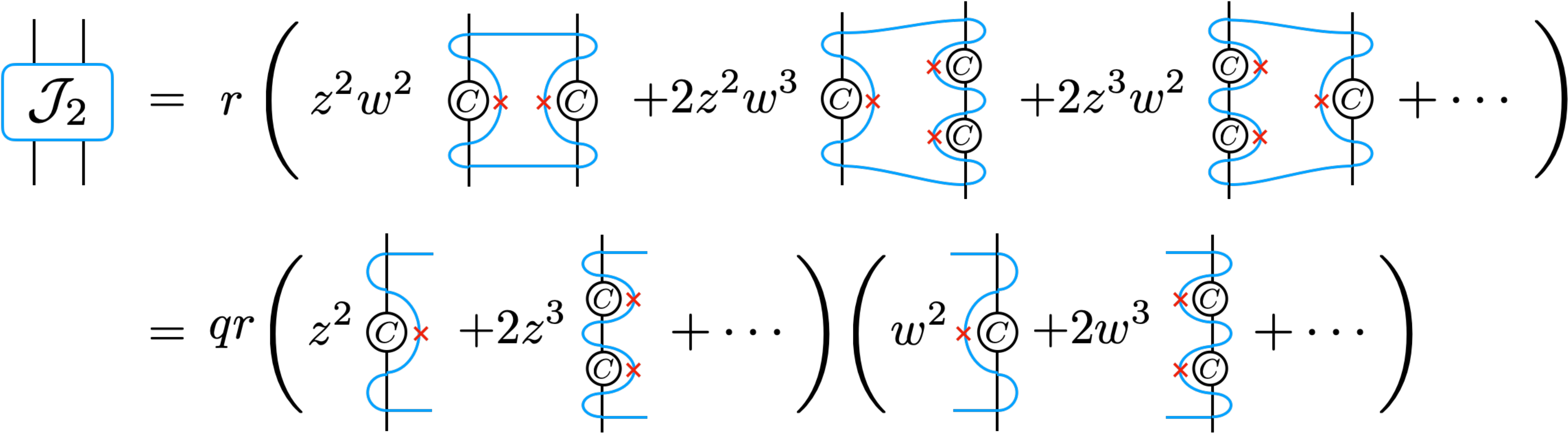}      
    \end{matrix}
  \end{equation}
  where we weight each diagram by product of the number of outer strands on each sides (marked by red crosses above).
  Algebraically these equations read
  \begin{align}
    \mathcal{J}_1(z,w) &= \frac{qz}{1-zC(q,z)}\frac{qw}{1-wC(q,w)} = H(q,0,z)H(q,0,w) \\
    \mathcal{J}_2(r,z,w) &= qr\frac{z^2C(q,z)}{(1-zC(q,z))^2}\frac{w^2C(q,w)}{(1-wC(q,w))^2}
  \end{align}
  Finally we can form the generating function for $\Braket{x|D(h)|y}$ using
  \begin{equation}
    \begin{matrix}
      \includegraphics[scale=.33]{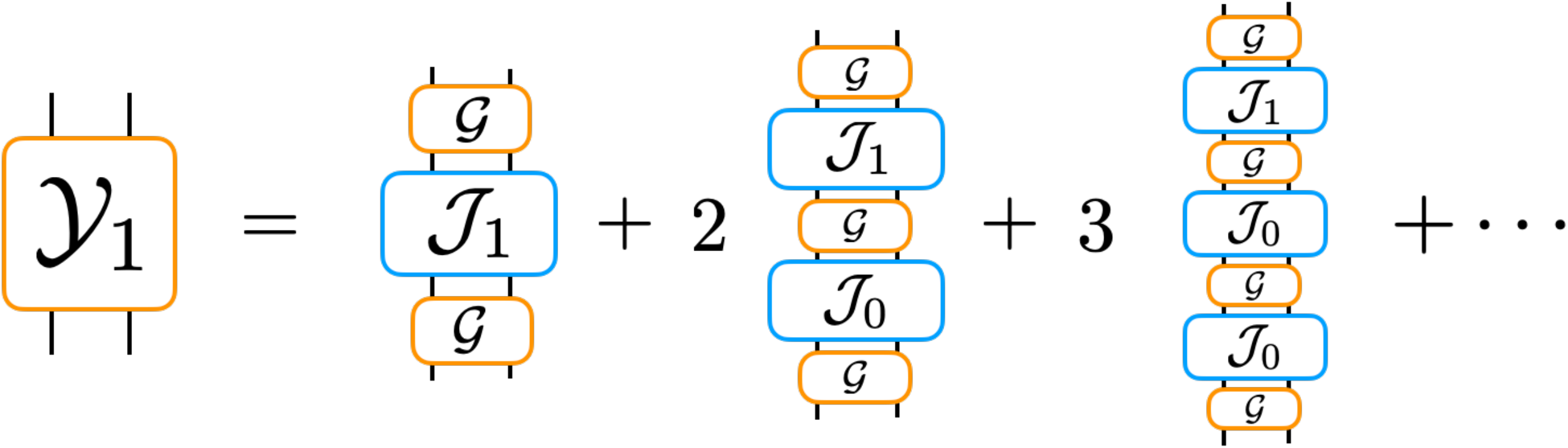}      
    \end{matrix}
  \end{equation}
  and similarly for $\mathcal{J}_2$. Algebraically this equation is
  \begin{equation}
    \mathcal{Y}_i(r,w,z)=\frac{\mathcal{J}_i(r,w,z)\mathcal{G}(w,z)^2}{(1-\mathcal{J}_0(r,w,z)\mathcal{G}(w,z))^2}, \quad i=1,2
  \end{equation}
  The reason we separate the weighted generating function into $\mathcal{Y}_1$ and $\mathcal{Y}_2$ is that after the indicated surgery is performed, the number of loops in diagrams of type $\mathcal{J}_1$ \textit{decreases} by $1$ (this is the only case for $k=0$); whereas diagrams of type $\mathcal{J}_2$ \textit{increases} by $1$.
  Denote the triple expansion coefficient of $\mathcal{Y}_i(r,z,w)$ by
  \begin{equation}
    \mathcal{Y}_i(q,r,w,z) = \sum^\infty_{\mu,\nu, k=0} y^{(i)}_{\mu\nu k}(q)z^\mu w^\nu r^k
  \end{equation}
  This allows us to write down
  \begin{align}
    \begin{split}
        &q^{-k}\Braket{v^{(m)}_k(q)|\Delta M^{(m)}_k(q')|v^{(m)}_k(q)} \\
      &\approx \left(q'y^{(1)}_{m/2,m/2,k}((qq')^{-1}+q^{\prime-1}y^{(2)}_{m/2,m/2,k}((qq')^{-1})\right) \chi^{m-2k-1} + O(\chi^{m-2k-2})
    \end{split}    
    \end{align}
    The first order correction to $\lambda_{M^{(m)}_k}$ from subleading diagrams is thus
  \begin{align}
  \label{eq:corr_dl2}
    \Delta_2 \lambda_{M^{(m)}_k} \approx \left((q_A+q_B)y^{(1)}_{m/2,m/2,k}(q^{-1}_{AB})+(q_A^{-1}+q^{-1}_B)y^{(2)}_{m/2,m/2,k}(q_{AB}^{-1})\right)\chi^{m-2k-1}
  \end{align}
  Although it is not immediately apparent, \Eqref{eq:lambda_corr_2} can be recovered by taking $k=0$.
  
\item sector mixing \\
We now give a brief argument that the correction to the degenerate zero eigenvalues in each sector vanish to the first order.
In short, we need to find the first order corrections to the null states of the product matrix $M^{(m)}_k(q_A)M^{(m)}_k(q_B)$.
We need to apply first order degenerate perturbation theory.
  Denote
  \begin{equation}
    \Pi^{(m)}_k = \text{id} - \frac{\ket{v^{(m)}_k(q_A)}\bra{v^{(m)}_k(q_B)}}{\Braket{v^{(m)}_k(q_A)|v^{(m)}_k(q_B)}}
  \end{equation}
  to be the projector onto the left and right null spaces.
  To find the perturbation one needs to diagonalize the following matrix
  \begin{align}
    \begin{split}
      \Pi^{(m)}_k\left(\Delta M^{(m)}_k(q_A)(M_0)^{(m)}_k(q_B)+(M_0)^{(m)}_k(q_A)\Delta M^{(m)}_k(q_B)\right)\Pi^{(m)}_k = 0,
    \end{split}
  \end{align}
  since $(M_0)^{(m)}_k(q)$ is annihilated by the projector $\Pi^{(m)}_k$.
  Therefore the order $O(\chi^{-1})$ corrections to the zero eigenvalues vanish and the degeneracy is not removed.
  
  The degeneracy is only broken at $O(\chi^{-2})$. 
  At the next order one has the second order contribution of $\Delta M^{(m)}_k(q_A)(M_0)^{(m)}_k(q_B)+(M_0)^{(m)}_k(q_A)\Delta M^{(m)}_k(q_B)$ as well as the first order contribution of $\Delta M^{(m)}_k(q_A)\Delta M^{(m)}_k(q_B)$, and all other contribution is annihilated by the projector.
  While the detailed effect of sector mixing is interesting on its own in that there seems to be an additional hierarchy
  structure from our TL numerics, its effect on the reflected entropy is to the order $O(\chi^{-2})$ which is outside our main interest here.
  
\end{itemize}

\subsection{explicit form for $k=0$ and $k=1$}
Here we will perform explicit analytic continuation $m\to1$ to find the leading order corrections to the spectral eigenvalue $\lambda_0$ and $\lambda_1$.
Together they determine the leading corrections to the reflected entropy in powers of $\chi$.

\begin{itemize}
\item $k=0$ \\
  The generating function for $\Delta_1\lambda$ can be written as
  \begin{align}
    \frac{\chi}{q_A+q_B}G_1(r=0,z)=\frac{-1+(1+q_{AB})\chi z}{2q_{AB}^2\chi^2z^2}+ \frac{1-2(1+q_{AB})\chi z+(1+q_{AB}^2)\chi^2 z^2}{2q_{AB}\chi^2z^2\sqrt{(1+\chi z(q_{AB}-1))^2-4q_{AB} \chi z}}
  \end{align}
  As in the case of main text there is a branch cut induced by the square root.
  To perform the analytic continuation via contour integral, we chose our contour to wrap around the branch cut.
  Additional poles at $z=0$ does not contribute to the integral.
  To set up our notation, define
  \begin{align}
    D_n(q) &\equiv -\frac{1}{\pi}\text{Im}\int^{z_+}_{z_-} \frac{dz}{z^{n}\sqrt{(1+z(q-1))^2-4q z}} \\
    &=
    \begin{cases}
      \;_2F_1(1-n,1-n;1;q), \quad &q\leq1,\\
      q^{n-1}\;_2F_1(1-n,1-n;1;q^{-1}), \quad &q>1\\
    \end{cases}
  \end{align}
  as the analytic continuation for the series coefficients of $z/\sqrt{(1+z(q-1))^2-4q z}$.
  An contour integral gives the expression for $\delta g_{1/2,0}$:
  \begin{align}
    \begin{split}
        \delta g_{1/2,0}(q_A,q_B)&=-\frac{1}{\pi}\text{Im}\int^{z_+/\chi}_{z_-/\chi}G_1(0,z)z^{-3/2}dz \\
                      &=\frac{(q_A+q_B)}{2q_{AB}\sqrt{\chi}} \left((1+q^2_{AB})D_{3/2}(q_{AB})-2(1+q_{AB})D_{5/2}(q_{AB})+D_{7/2}(q_{AB})\right)
    \end{split}
  \end{align}
  and
  \begin{align}
    \begin{split}
        \Delta_1 \lambda_0 &= \frac{2q_{AB}}{\sqrt{\chi}} g_{1/2,0}(q^{-1}_{AB}) g^{(1)}_{1/2,0}(q_A^{-1},q_B^{-1}) \\
                &=  \frac{(q^{-1}_A+q^{-1}_B)}{\chi}C_{1/2}(q^{-1}_{AB})\left((1+q^{-2}_{AB})D_{3/2}(q^{-1}_{AB})-2(1+q^{-1}_{AB})D_{5/2}(q^{-1}_{AB})+D_{7/2}(q^{-1}_{AB})\right)
    \end{split}
  \end{align}
  The expression for $\Delta_2\lambda_0$ is also easy to write down.
  We work with the single variable generating function \Eqref{eq:gen_fun_Y} here.
  Use \Eqref{eq:Cz_quad} one arrives at
  \begin{equation}
    Y(q,z) = \left(\frac{1}{z}-q\right)C(q,z)-\frac{1}{z}
  \end{equation}
  which gives
  \begin{align}
    \begin{split}
        \Delta_2\lambda_0 &= \frac{q_{AB}(q_A+q_B)}{\chi}\left(\frac{1}{\pi}\text{Im}\int^{z^+}_{z^-} Y(q^{-1}_{AB},z)z^{-3/2}dz \right)^2\\    
             &=\frac{q_{AB}(q_A+q_B)}{\chi}(C_{3/2}(q^{-1}_{AB})-q^{-1}_{AB}C_{1/2}(q^{-1}_{AB}))^2  
    \end{split}
  \end{align}
  and the first order correction is
  \begin{align}
    \begin{split}
          \Delta\lambda_0 = \Delta_1\lambda_0+\Delta_2\lambda_0 
    \end{split}
  \end{align}  
  
\item $k=1$ \\
  Extracting the linear $r$ order coefficient of the generating functions we get
  \begin{align}
    \frac{\chi^2}{q_A+q_B}\frac{\partial G_1(r,z)}{\partial r}\Big|_{r=0} &= \frac{-1+(q_{AB}+1)\chi z+q_{AB}^2z^2-(q^2_{AB}-q_{AB})^2\chi^3z^3}{2q_{AB}^2\chi^3z^3} \nonumber \\
     &+\frac{1-2(q_{AB}+1)\chi z+\chi^2z^2+2q_{AB}^3\chi^3z^3-(q^2_{AB}-q_{AB})^2\chi^4z^4}{2q_{AB}^2\chi^3z^3\sqrt{(1+\chi z(q_{AB}-1))^2-4q_{AB} \chi z}}
  \end{align}
  The first term on RHS is irrelevant for determining the contour integral.
  We have
  \begin{align}
    \begin{split}
        g^{(1)}_{1/2,1}(q_A,q_B) &= -\frac{1}{\pi}\text{Im}\int^{z_+/\chi}_{z_-/\chi}\partial_r G_1(0,z)z^{-3/2}dz \\
                             &=\frac{q_A+q_B}{2\chi^{3/2}}\left( -(q_{AB}-1)^2D_{1/2}(q_{AB})+2q_{AB}D_{3/2}(q_{AB})+q_{AB}^{-2}D_{5/2}(q_{AB})\right.\\
      &\qquad\qquad\qquad \left.-2(q_{AB}^{-2}+q^{-1}_{AB})D_{7/2}(q_{AB})+q_{AB}^{-2}D_{9/2}(q_{AB})  \right)
    \end{split}
  \end{align}
  and
  \begin{align}
    \begin{split}
        \Delta_1\lambda_1 &= \frac{2}{\chi^{3/2}}g_{1/2,1}(q^{-1}_{AB})g^{(1)}_{1/2,1}(q^{-1}_A,q^{-1}_B)\\
               &=\frac{(q_{A}^{-1}+q^{-1}_B)}{\chi^3}(-C_{1/2}(q_{AB}^{-1})+q_{AB}C_{3/2}(q_{AB}^{-1}))\\
               &\quad \times (( -(q^{-1}_{AB}-1)^2D_{1/2}(q^{-1}_{AB})+2q^{-1}_{AB}D_{3/2}(q^{-1}_{AB}) +q_{AB}^{2}D_{5/2}(q_{AB}^{-1})     \\
      &\qquad-2(q_{AB}^{2}+q_{AB})D_{7/2}(q^{-1}_{AB})+q_{AB}^{2}D_{9/2}(q_{AB}^{-1}) )
    \end{split}
  \end{align}

  The two variable generating function $\mathcal{Y}_i(q,r,z,w)$ factorizes when expanded as a series of $r$. We are interested in terms linear to $r$ here, which is
  \begin{align}
    \begin{split}
        \frac{\partial\mathcal{Y}_2(q,r,z,w)}{\partial r}\Big|_{r=0} &= \frac{qz^2w^2C(q,z)^3C(q,w)^3}{(1-zC(q,z))^2(1-wC(q,w))^2}, \\
      \frac{\partial\mathcal{Y}_1(q,r,z,w)}{\partial r}\Big|_{r=0} &= 2q^2\frac{\partial\mathcal{Y}_2(q,r,z,w)}{\partial r}\Big|_{r=0}  
    \end{split}
  \end{align}
  Using \Eqref{eq:Cz_quad} we can write the individual factors as
  \begin{align}
    \begin{split}
        \hat{Y}(q,z) \equiv \frac{z^2C(q,z)^2}{(1-zC(q,z))^2} 
                     =  \frac{-1+(1+q)z}{q^2z^2} + \frac{1-(2q+1)z+q^2z^2}{q^2z^2}C(q,z)
    \end{split}
  \end{align}
  $\hat{Y}(q,z)$ has the following contour integral
  \begin{align}
    -\frac{1}{\pi}\text{Im}\int^{z_+}_{z_-}\frac{dz}{z^{3/2}} \hat{Y}(q,z) = C_{1/2}(q)-(2q^{-1}+q^{-2})C_{3/2}(q)+q^{-2}C_{5/2}(q)    
  \end{align}
  which gives the following analytic continuation for the expansion coefficient
  \begin{align}
    \begin{split}
      y^{(2)}_{1/2,1/2,1}(q) &= q\left(C_{1/2}(q)-(2q^{-1}+q^{-2})C_{3/2}(q)+q^{-2}C_{5/2}(q)\right)^2, \\
      y^{(1)}_{1/2,1/2,1}(q) &= 2q^2y^{(2)}_{1/2,1/2,1}(q)
    \end{split}
  \end{align}
  and
  \begin{align}
    \Delta_2\lambda_1 = \frac{3(q_A^{-1}+q^{-1}_B)}{\chi^3}
    \left(C_{1/2}(q_{AB}^{-1})-(2q_{AB}+q^2_{AB})C_{3/2}(q_{AB}^{-1})+q^{2}_{AB}C_{5/2}(q^{-1}_{AB})\right)^2
  \end{align}
  The first order correction to the leading eigenvalue in $k=2$ sector is given by
  \begin{equation}
    \Delta\lambda_1 = \Delta_1\lambda_1+\Delta_2\lambda_1
  \end{equation}
  
\end{itemize}

\section{Proofs}
\label{sec:proof}
\subsection{Proof of Proposition~\ref{lem:dk}}
\begin{proof}
  For given $m$ we have a series of unknowns $d^{(m)}_k$ where $k=0,1,\cdots,m/2$.
  We pick $m/2$ elements $t_i\in\text{TL}_m$ such that the trace conditions $\text{Tr}_{\text{TL}_m}t_i = \sum_kd^{(m)}_k\text{tr}(\pi^{(m)}_k(t_i))$ give rise to $m/2$ linear equations for which we check that the proposed solutions $d^{(m)}_k=[2k+1]_q$ satisfy.
  We pick:
  \begin{equation}
    \includegraphics[scale=.3]{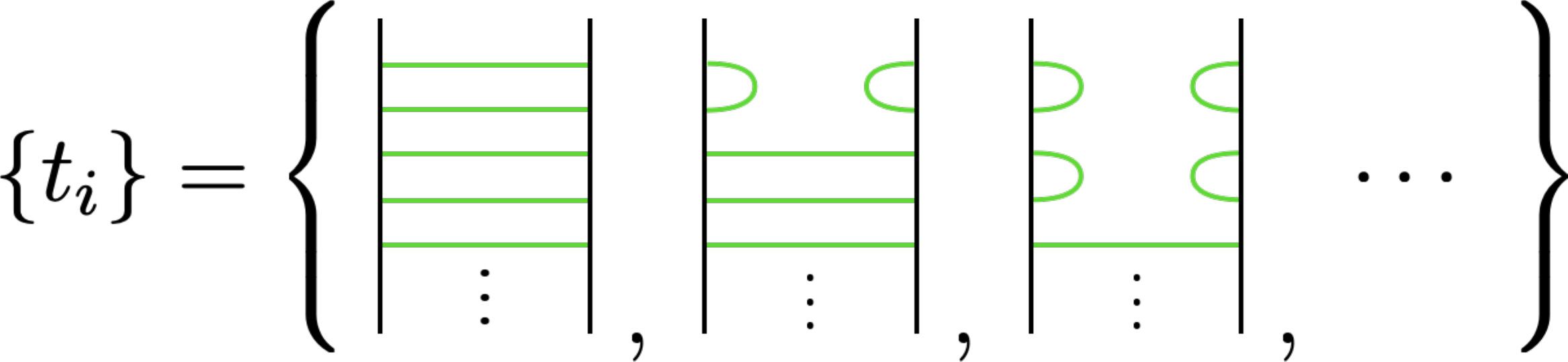}
  \end{equation}
  or in terms of generators of $\text{TL}_m$:
  \begin{equation}
    \{t_i\} = \{\text{id},e_1,e_1e_3,e_1e_3e_5,\cdots\}
  \end{equation}
  The element $\text{id}$ maps to the identity matrix in every submodule, so the trace condition for it is simply
  \begin{equation}
    \sum^{m/2}_{k=0} d^{(m)}_k |\mathcal{V}^{(m)}_k| = \chi^{m}
  \end{equation}
  The element $e_1$ annihilates the module $\mathcal{V}^{(m)}_{m/2}$ since the action of $e_1$ on any link state will always contain at least $2$ closed strands.
  Next, to calculate $\text{tr}(\pi^{(m)}_k(e_1))$ we must find all elements $v_j\in V^{(m)}_k$ such that $\pi^{(m)}_k(e_1)v_j\propto v_j$.
  Such a link state must have a closed strand connecting first two sites, but can otherwise have arbitrary connections in the remaining $m-2$ sites since $e_1$ act on these sites just as identity.
  The closed loop after concatenation contribute a single $\chi$ for all $\pi^{(m)}_k(e_2)v_j$ and the trace condition for $e_2$ is thus
  \begin{equation}
    \sum^{m/2-1}_{k=0} d^{(m)}_k |\mathcal{V}^{(m)}_k| = \chi^{m-2}
  \end{equation} 
  For other elements in the list a similar argument also works as one simply consider link states with progressively more closed strands connecting adjacent pair of sites when constructing the eigenvector of $\pi^{(m)}_k(t_i)$.
  The trace condition for $t_i$ is
  \begin{equation}
    \sum^{m/2-i+1}_{k=0} d^{(m)}_k |\mathcal{V}^{(m)}_k| = \chi^{m-2i+2}
  \end{equation}  
  We observe that all the trace conditions have the same form -- in fact for any given $m$ only the $t_1$ condition is new and the remaining $m/2-1$ conditions coincide with the conditions of the $m-1$ module.
  This means a solution for $d^{(m)}_k$ automatically solves the equations of $d^{(m-1)}_k$. Thus we conclude that $d^{(m)}_k$ is independent of $m$.
  As a result, we will drop the $m$ superscript in the remainder of this proof.

  Let us reorganize the $d_k$ equations we obtained. Introducing the variable $q$ which satisfies $\chi = q+q^{-1}$ and rewriting $|\mathcal{V}^{(m)}_k|$ using Lemma~\ref{lem:module_dim} we have an infinite set of conditions
  \begin{equation}
    \sum_{k=0}^{m/2}d_{k} \,\#\text{SYT}\left(m/2+k,m/2-k\right) = (q+q^{-1})^m, \quad \forall m \in  2\mathbb{Z}_+
  \end{equation}
  We now claim $d_k=[2k+1]=q^{-2k}+q^{-2k+2}+\cdots+q^{2k}$ solves this equation for every even $m$.  
  Plug in $d_k$ and match coefficients with the same power of $q$ we find
  \begin{equation}
    \sum_{k=0}^{n} \#\text{SYT}\left(m/2+k,m/2-k\right) \underset{?}{=}
    \begin{pmatrix}
      m \\ m/2-n
    \end{pmatrix}
  \end{equation}
  Or equivalently
  \begin{equation}
    \#\text{SYT}\left(m/2+n,m/2-n\right) \underset{?}{=}
    \begin{pmatrix}
      m \\ m/2-n
    \end{pmatrix}-
    \begin{pmatrix}
      m \\ m/2-n-1
    \end{pmatrix}
  \end{equation}
  which one can check by explicit computation
  \begin{align}
  \begin{split}
   \#\text{SYT}\left(m/2+n,m/2-n\right)  &=
    \begin{pmatrix}
      m \\ m/2-n
    \end{pmatrix}-
    \begin{pmatrix}
      m \\ m/2-n-1
    \end{pmatrix}\\
    &= \frac{m!(2n+1)}{(m/2+n+1)!(m/2-n)!}
  \end{split}
\end{align}

To show that the matrix trace $\text{tr}=\sum_k d_k \text{tr}_k$ on $\mathcal{V}^{(m)}$ coincides with $\text{Tr}_{\text{TL}_m}$ for all elements in $\mathcal{V}^{(m)}$, it suffices to check that \Eqref{eq:tr_char} is true.
We have $\text{tr} (\pi^{(m)}(\text{id})) = \chi^m$ by construction, and $\forall h\in \text{TL}_{m-1}\subset \text{TL}_{m}$ we have
\begin{equation}
\label{eq:tr_sum_v}
    \text{tr} (\pi^{(m)}(he_{m-1})) = \text{tr} (\pi^{(m)}(e_{m-1}h)) =\sum_{k} d_k\tr\left(\pi^{(m)}_k(e_{m-1}h)\right)
\end{equation}
To evaluate this trace we must find all links states $v\in\mathcal{B}^{(m)}_k$ such that $e_{m-1}hv \propto v$. 
We can classify $v$'s by the contraction pattern of the last two sites. 
We denote $v_1$ to be the set such that the last two sites are connected through a closed strand and $v_2$ to be the remaining ones. Schematically,
\begin{equation}
\begin{matrix}
    \includegraphics[scale=.35]{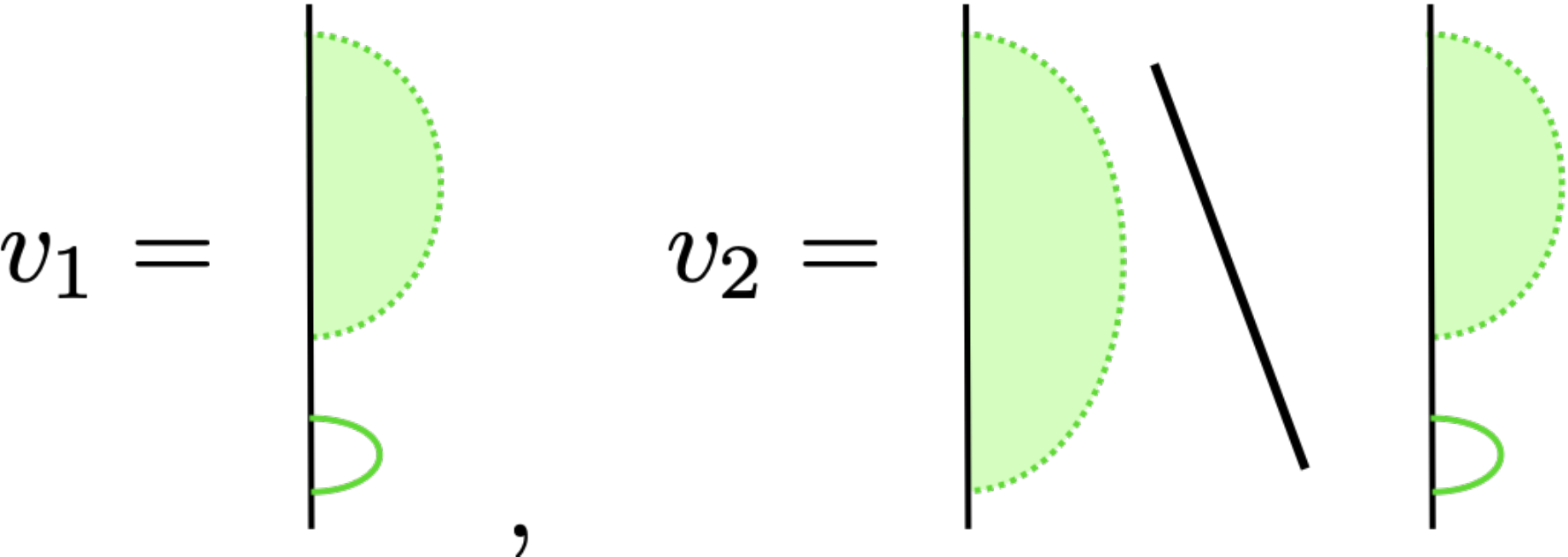},
\end{matrix}
\end{equation}
where we used colored block to represent arbitrary connections.
It is easy to check that 
\begin{equation}
\begin{matrix}
    \includegraphics[scale=.35]{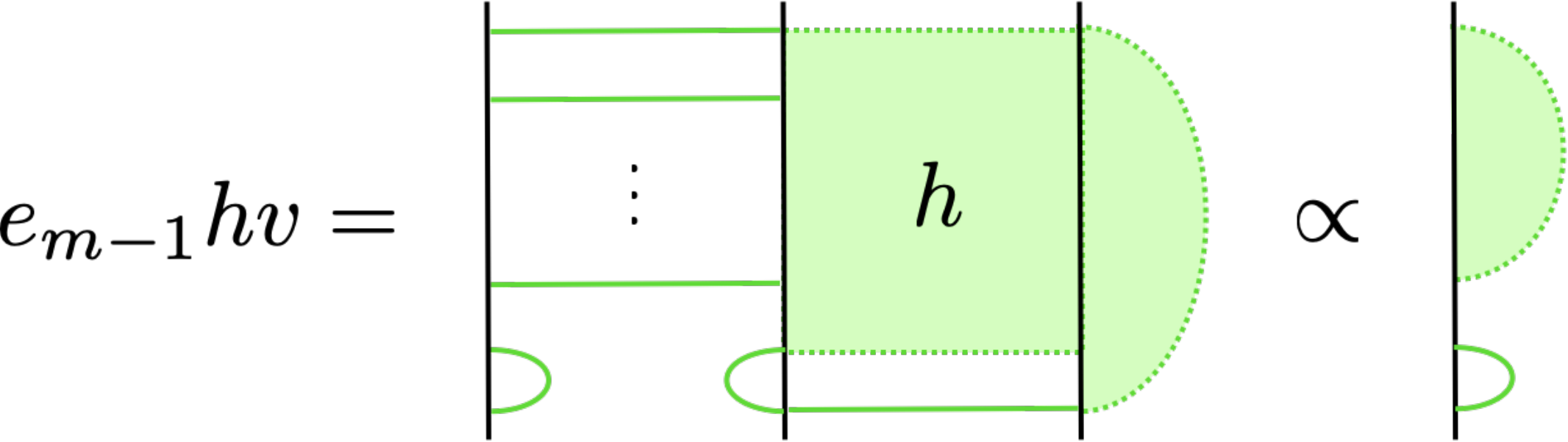},
\end{matrix}
\end{equation}
and thus $he_{m-1}v_2\not\propto hv_2$ so they do not contribute to the trace.
For $v_1$, note that
\begin{equation}
\label{eq:ehv1}
\begin{matrix}
    \includegraphics[scale=.35]{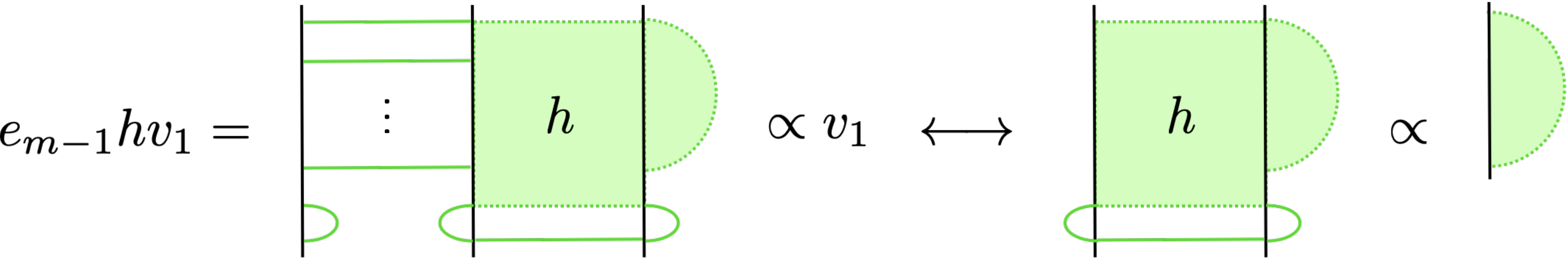}.
\end{matrix}
\end{equation}
But this is only possible if in the diagram of $h$ there is a path connecting the $(m-1)$-th site on the left to the $(m-1)$-th site on the right, otherwise either the RHS of \Eqref{eq:ehv1} vanishes or it is not possible to reproduce the connection pattern of $v_1$.
Also,
\begin{equation}
\begin{matrix}
\label{eq:hv1}
    \includegraphics[scale=.35]{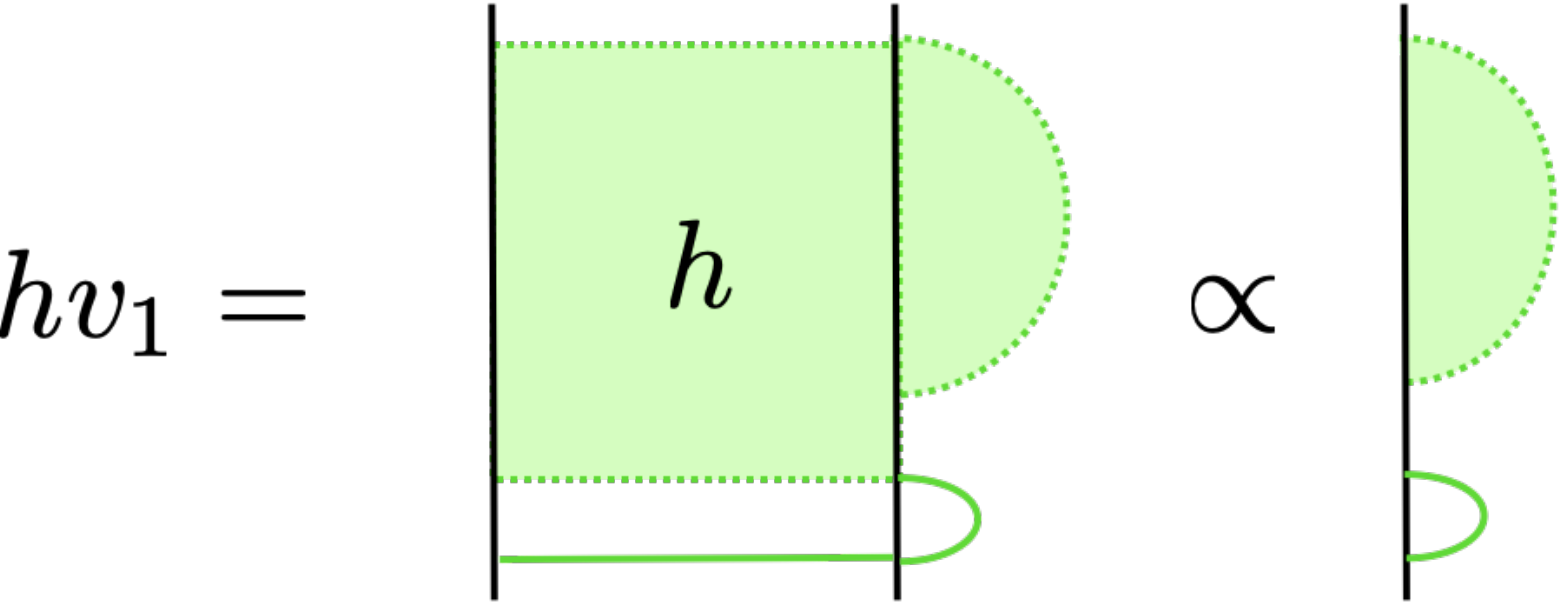}
\end{matrix}
\end{equation}
is only possible for the same subset of $h$ that \Eqref{eq:ehv1} is true for a similar reason.
For such $h$ the contraction in the RHS of \Eqref{eq:ehv1} produces an extra factor of $\chi$, but is otherwise identical to \Eqref{eq:hv1}. 
Hence we find that $hv=av\leftrightarrow e_{m-1}hv=\chi av$ for some $a\in \mathbb{C}$ and thus $\text{tr}(h)=\chi\text{tr}(he_{m-1})$. This result naturally generalizes to all sequential inclusions of algebras in the list $\text{TL}_1\subset\cdots\subset\text{TL}_{m-1}\subset\text{TL}_{m}$.
Using Lemma~\ref{lem:trace} we conclude that the trace function we constructed on $\mathcal{V}^{(m)}$ is indeed the same as $\text{Tr}_{\text{TL}_m}$.

\end{proof}

\subsection{Proof of Proposition~\ref{lem:gen_fun}}
\begin{proof}
  The generating function $G(q,r,z)$ is defined as an infinite sum over diagrams
  \begin{equation}
    \begin{matrix}
       \includegraphics[scale=.3]{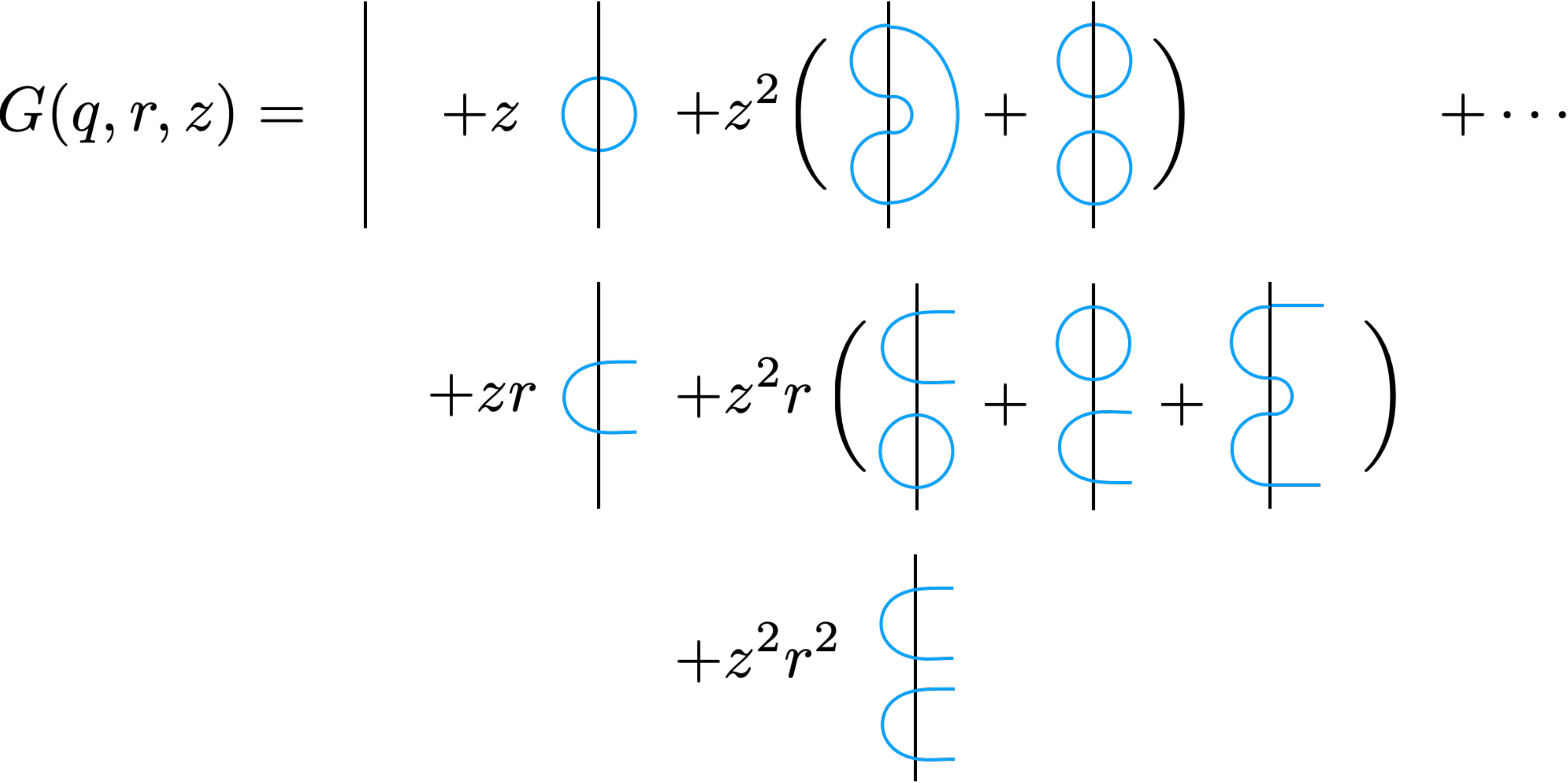}   
    \end{matrix}
  \end{equation}
  where the power of $z$ counts (half) the number of elements $m/2$ and power of $r$ counts (half) the number of defects $k$.
  Each diagram is evaluated in a similar fashion as the definition of linear functional $f_q$ in \Eqref{eq:f_q} -- for every closed loop in diagram assign a (positive rather than negative) power of $q$ to the result.
  Also introduce 1PI generating function $H(q,r,z)$ defined as the sum over all connected diagrams:
  \begin{equation}
    \begin{matrix}
       \includegraphics[scale=.3]{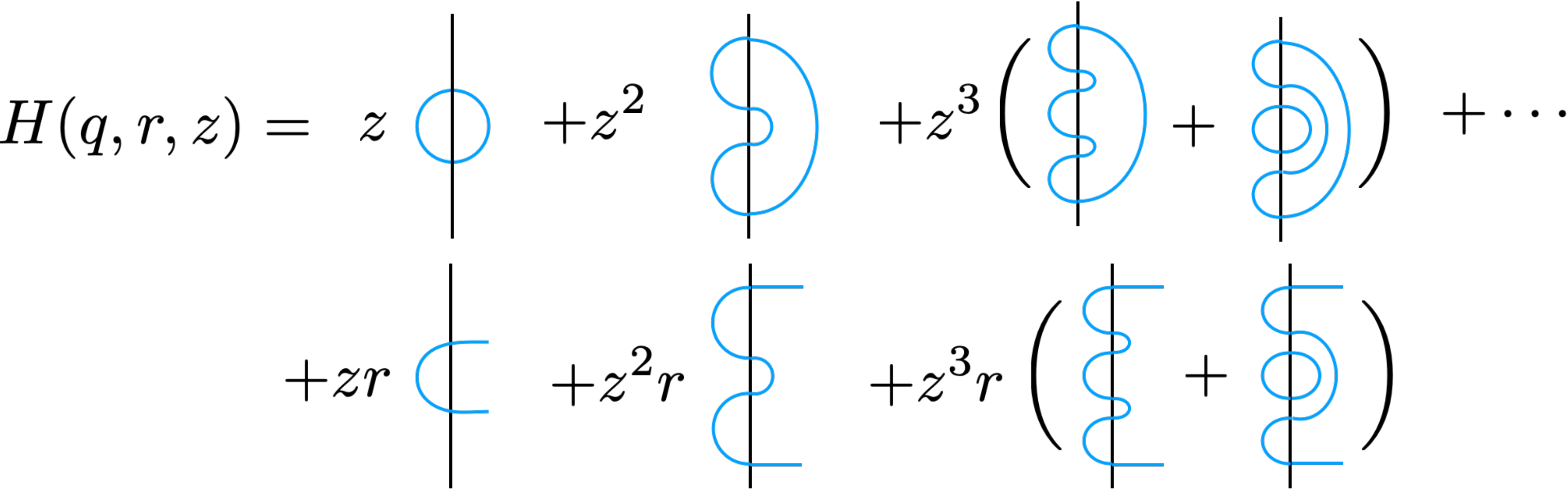}   
    \end{matrix}
  \end{equation}
  The form of $H$ allows us to reorganize it as
  \begin{equation}
    \begin{matrix}
       \includegraphics[scale=.3]{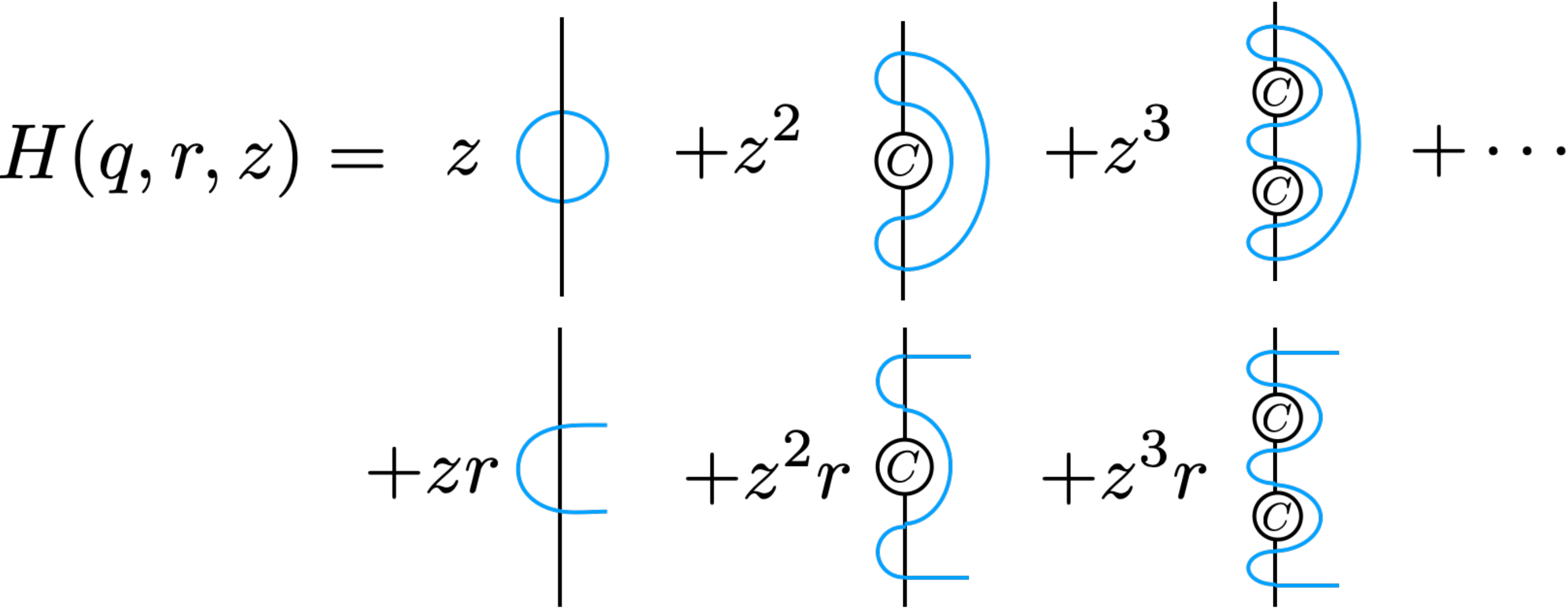}   
    \end{matrix}
  \end{equation}
  or
  \begin{align}
  \label{eq:1PIgen}
    \begin{split}
      H(q,r,z) &= z(q+r)(1+zC(q,z)+z^2C^2(q,z)+\cdots) \\
      &= \frac{z(q+r)}{1-zC(q,z)}   
    \end{split}
  \end{align}
  $G$ and $H$ are related by
  \begin{equation}
    \label{eq:genfun_app}
    \begin{split}      
       G(q,r,z) &= 1+H(q,r,z) + H^2(q,r,z) + \cdots = \frac{1}{1-H(q,r,z)} \\
   &= \frac{1-zC(q,z)}{1-z(q+r)-zC(q,z)}
    \end{split}
  \end{equation} 
  The function $C(q,z)$ is the generating function of link states with no defects:
  \begin{equation}
    \label{eq:Cz_identity}
    \begin{matrix}
       \includegraphics[scale=.3]{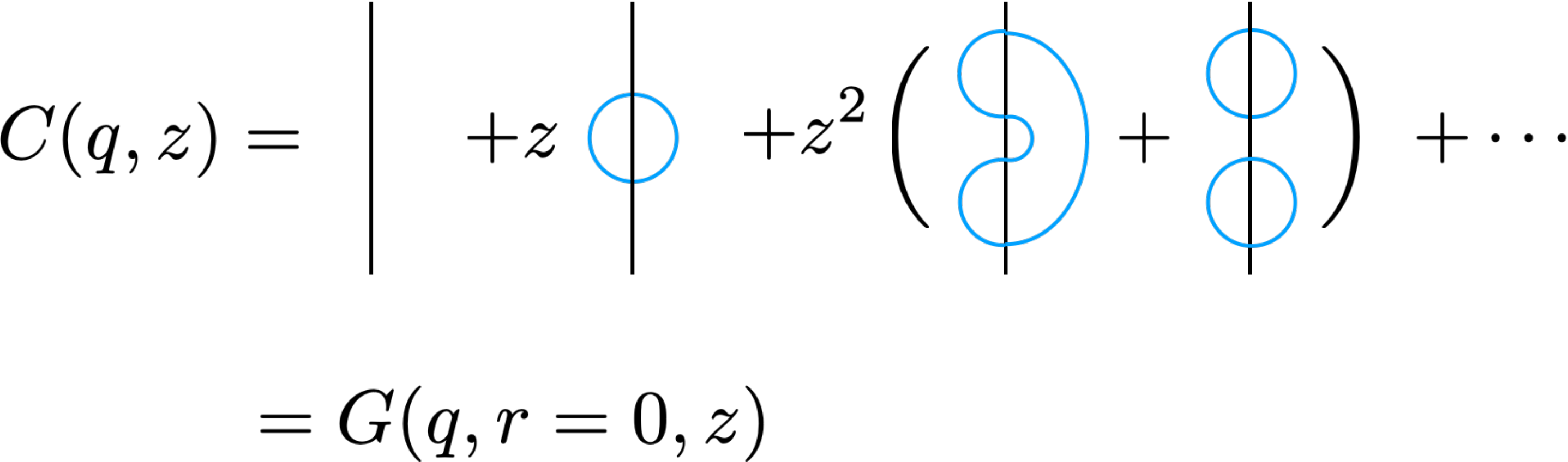}   
    \end{matrix}
  \end{equation}
  $C(q,z)$ is also known under the name of the generating function of \textit{q-Catalan numbers}.
  \begin{equation}    
    C(q,z) = \frac{1-z(q-1)-\sqrt{(1+z(q-1))^2-4qz}}{2z}
  \end{equation}
  For a derivation of $C(q,z)$ please refer to the proof of Lemma~\ref{lem:q_catalan}.
\end{proof}

\subsection{Proof of Proposition~\ref{lem:gen_fun_2}}
\begin{proof}
  Define $\tilde{C}(q,z)=C(q,z)-1$. It satisfies a different quadratic equation
  \begin{equation}
    \tilde{C}(q,z)+(q+1-\frac{1}{z})\tilde{C}(q,z)+q=0
  \end{equation}
  Consider the generating function of $\tilde{C}^k(q,z)$:
  \begin{align}
    \begin{split}
      X(q,z,t) &= \sum^\infty_{k=0} \tilde{C}^k(q,z)t^k \\
               &=(\frac{1}{z}-q-1)\sum^\infty_{k=2}\tilde{C}^{k-1}t^k-q\sum^\infty_{k=2}\tilde{C}^{k-2}t^k+\tilde{C}(q,z)t+1 \\
               &=(\frac{1}{z}-q-1)t(X(q,z,t)-1)-qt^2X(q,z,t)+\tilde{C}(q,z)t+1
    \end{split}
  \end{align}
  Solving for $X$ gives
  \begin{align}
    \begin{split}
       X(q,z,t) &= \frac{\tilde{C}(q,z)t+(q+1-\frac{1}{z})t+1}{qt^2+(q+1-\frac{1}{z})t+1} \\
    &= \frac{C(q,z)t+(q-\frac{1}{z})t+1}{qt^2+(q+1-\frac{1}{z})t+1}   
    \end{split}
  \end{align}
  Now we can write $G(q,r,z)$ as
  \begin{align}
    \begin{split}
      G(q,r,z) &= C(q,z)X(q,z,r/q) \\
      &= \frac{(r+q)C(q,z)-\frac{r}{z}}{r^2+(q+1-\frac{1}{z})r+q}   
    \end{split}
  \end{align}
  where we have again used \Eqref{eq:Cz_quad} to swap out $C^2(q,z)$. 

\end{proof}

\bibliographystyle{jhep}
\bibliography{mybibliography}
 
\end{document}